\documentclass{article}
\usepackage[utf8]{inputenc}
\usepackage[english]{babel}
\usepackage[utf8]{inputenc}

\usepackage{lscape}
\usepackage{amssymb}
\usepackage{amsmath}
\usepackage{mathtools}
\usepackage{stmaryrd}
\usepackage[hypcap=false]{caption}
\usepackage{multirow}

\usepackage{amsthm}
\usepackage{makeidx}
\usepackage[color,all]{xy}
\usepackage{xcolor,colortbl}

\usepackage[hmargin={3cm,3cm},vmargin={2cm,3.5cm}]{geometry}

\usepackage[pdftex]{graphicx}
\usepackage{comment}
\usepackage{tikz}

\usepackage[colorlinks=true,linkcolor=black]{hyperref}

\hypersetup{urlcolor=blue, citecolor=red}

\def \talpha{\alpha\!\!\!\!\alpha \!\!\!\!\alpha}
\def \tbeta{\beta\!\!\!\!\beta \!\!\!\!\beta}

\newtheorem{theorem}{Theorem}[section]
\newtheorem{corollary}{Corollary}

\newtheorem{lemma}[theorem]{Lemma}
\newtheorem{proposition}{Proposition}

\theoremstyle{definition} 
\newtheorem{example}{Example}
\newtheorem{definition}[theorem]{Definition}

\newcommand{\wg}{\wedge}

\newcommand\xqed[1]{%
  \leavevmode\unskip\penalty9999 \hbox{}\nobreak\hfill
  \quad\hbox{#1}}
\newcommand\demo{\xqed{$\triangle$}}

\numberwithin{equation}{section}
\numberwithin{proposition}{section}
\numberwithin{example}{section}
\numberwithin{corollary}{section}

\begin{document}

\centerline{\Large \bf A Grassmann and graded approach to coboundary Lie bialgebras,}
\vskip 0.25cm
\centerline{\Large \bf their classification, and Yang-Baxter equations}
\vskip 0.20cm
\centerline{J. de Lucas and D. Wysocki}
\vskip 0.20cm
\centerline{Department of Mathematical Methods in Physics, University of Warsaw,}
\centerline{ul. Pasteura 5, 02-093, Warsaw, Poland}

\begin{abstract} We devise geometric, graded algebra, and Grassmann methods to study and to classify finite-dimensional coboundary Lie bialgebras. Several mathematical structures on Lie algebras, e.g. Killing forms, root decompositions, or gradations are extended to their Grassmann algebras. The classification of real three-dimensional coboundary Lie bialgebras is retrieved throughout devised methods. The structure of modified classical Yang-Baxter equations on $\mathfrak{so}(2,2)$ and $\mathfrak{so}(3,2)$ are studied and $r$-matrices are found. Our methods are extensible to other coboundary Lie bialgebras of even higher dimension.
\end{abstract}



\section{Introduction}
Lie bialgebras \cite{Pressley,Yvette,Le90,Ro91}, defined precisely by Drinfeld \cite{Dr83, Dr87}, emerged in the study of integrable systems \cite{Fa84, Fa87}. A {\it Lie bialgebra} consists of a Lie algebra $\mathfrak{g}$ and a Lie algebra structure on its dual, $\mathfrak{g}^*$, that are compatible in a certain sense. Lie bialgebras occur in quantum gravity, where they lead to non-commutative space-times \cite{BCH00,BGGH17, BHM13,MS11, MS03}, quantum group theory \cite{Pressley}, and other topics \cite{Yvette,Op98}. 

The classification of Lie bialgebras for a fixed $\mathfrak{g}$ is an unfinished task. Lie bialgebras with $\dim \mathfrak{g}=2$ and $\dim \mathfrak{g}=3$ have been classified \cite{Farinati, Gomez}. Specific instances of higher-dimensional Lie bialgebras, e.g. for semi-simple $\mathfrak{g}$, have also been studied \cite{ARH17,BLT16,BLT17,LT17,Op98,Op00,Za97}. So far, employed techniques are mostly algebraic and not much effective to analyse higher-dimensional Lie bialgebras \cite{ARH17, Pressley,Farinati,Gomez}. Hence, new approaches to the study and  determination of such Lie bialgebras are interesting.

{\it Coboundary Lie bialgebras} represent a remarkable type of Lie bialgebras. They are characterised by solutions, the so-called $r$-{\it matrices}, to the {\it modified classical Yang-Baxter equations} (mCYBEs) \cite{Pressley,Gomez}. This work introduces novel geometric, graded algebra, and Grassmann algebra procedures to determine, to help in classifying, and to investigate coboundary Lie bialgebras.  As shown in examples, devised methods can be applied to Lie bialgebras on a relatively high-dimensional, not necessarily semi-simple, Lie algebra $\mathfrak{g}$. Let us survey more carefully the techniques introduced in our work.

First, the hereafter called $\mathfrak{g}$-\textit{invariant multilinear maps} on $\mathfrak{g}$-modules generalise Killing forms on $\mathfrak{g}$ to $\Lambda \mathfrak{g}$, and describe other structures, e.g. types of presymplectic forms \cite{HL09} or Casimir invariants \cite{BLT16}, and other invariants on $\mathfrak{g}$-modules (cf. \cite{RS97}) as particular cases. Note that $\mathfrak{g}$-invariant multilinear maps need not be symmetric.

Second, we endow each Lie algebra $\mathfrak{g}$ with a $G$-graded Lie algebra structure, namely a decomposition $\mathfrak{g}=\bigoplus_{\alpha\in G}\mathfrak{g}^{(\alpha)}$ for a commutative group $(G,\star)$, where $G\subset \mathbb{R}^n$ but the composition law $\star$ need not be the standard addition in $\mathbb{R}^n$, such that $[\mathfrak{g}^{(\alpha)},\mathfrak{g}^{(\beta)}]\subset \mathfrak{g}^{(\alpha\star \beta)}$. We call this structure a {\it $G$-gradation} on the Lie algebra $\mathfrak{g}$. We show that  the space, $\Lambda^k\mathfrak{g}$, of $k$-vectors on $\mathfrak{g}$ has a decomposition $\Lambda^k\mathfrak{g}=\bigoplus_{\alpha\in G}(\Lambda^k\mathfrak{g})^{(\alpha)}$ induced by the $G$-gradation of $\mathfrak{g}$. Then, we prove that the {\it algebraic Schouten bracket} $[\cdot,\cdot]_S:\Lambda \mathfrak{g}\times \Lambda \mathfrak{g}\rightarrow \Lambda\mathfrak{g}$ (see \cite{Yvette,Va94}), is such that $[(\Lambda^m\mathfrak{g})^{(\alpha)},(\Lambda^l\mathfrak{g})^{(\beta)}]_S\subset (\Lambda^{m+l-1}\mathfrak{g})^{(\alpha\star \beta)}$ for every $\alpha,\beta \in G$ and $m,l\in \mathbb{Z}$. Our gradations are applicable to relatively high-dimensional Lie algebras, as illustrated by our study of the Lie algebras $\mathfrak{so}(2,2)$ and $\mathfrak{so}(3,2)$ (see Figures \ref{so22_diags} and  \ref{so32_diags}, Tables \ref{so32_bases_2} and \ref{so32_bases_3}, and Example \ref{ex:decomp_so22}). Our gradations can also be applied to not necessarily semi-simple Lie algebras as witnessed by Table \ref{tabela3w}, where $G$-gradations for all three-dimensional Lie bialgebras and the induced decompositions on their Grassmann algebras are detailed.
 
A type of generalisation of root decompositions for general Lie algebras, the  {\it root gradations}, are suggested and briefly studied so as to study the determination of mCYBEs and CYBEs for three-dimensional Lie algebras. 

Previous structures are applied to studying and classifying coboundary Lie bialgebras up to Lie algebra automorphisms in an algorithmic way. Let us sketch this procedure. Let $(\Lambda \mathfrak{g})^\mathfrak{g}\subset \Lambda \mathfrak{g}$ be the space of elements  commuting with all elements of $\Lambda\mathfrak{g}$ relative to the algebraic Schouten bracket. If we denote $(\Lambda^m\mathfrak{g})^\mathfrak{g}:=\Lambda^m\mathfrak{g}\cap (\Lambda\mathfrak{g})^{\mathfrak{g}}$, then spaces $(\Lambda^2\mathfrak{g})^{\mathfrak{g}}$ and  $(\Lambda^3\mathfrak{g})^\mathfrak{g}$ are analysed through the decomposition in $\Lambda\mathfrak{g}$ induced by a gradation in $\mathfrak{g}$ and other new findings detailed in Section \ref{gInvar} relating the structures of $\mathfrak{g}$, $\Lambda\mathfrak{g}$, and $(\Lambda\mathfrak{g})^\mathfrak{g}$.

A coboundary Lie bialgebra on $\mathfrak{g}$ is determined through an  $r\in \Lambda^2\mathfrak{g}$ satisfying the mCYBE on $\mathfrak{g}$, namely $[r,r]_S\in (\Lambda^3\mathfrak{g})^\mathfrak{g}$. Since $r$-matrices differing in an element of $(\Lambda^2\mathfrak{g})^\mathfrak{g}$ give rise to the same Lie bialgebra (cf. \cite{Farinati}), the space of coboundary Lie bialgebras must be investigated through $\Lambda^2_R\mathfrak{g}:=\Lambda^2\mathfrak{g}/(\Lambda^2\mathfrak{g})^\mathfrak{g}$, whose elements are called {\it reduced multivectors}. We prove how previous $\mathfrak{g}$-invariant multilinear maps, gradations, algebraic Schouten brackets, and other introduced structures can be defined on $\Lambda_R^m\mathfrak{g}$. 

Next, $\mathfrak{g}$-invariant $k$-linear structures are employed to observe the equivalence up to inner automorphisms of the coboundary Lie bialgebras on $\mathfrak{g}$.  This is more general than standard techniques based on Casimir elements \cite{BR86}. It also enables us to describe more geometrically the problem of classification up to automorphisms of Lie bialgebras. The determination of automorphisms of Lie algebras is a complicated problem by itself (cf. \cite{Farinati}), but it will be rather unnecessary in our approach. We generally restricted ourselves to studying the equivalence under inner Lie algebra automorphisms. Then, the determination of very few not inner Lie algebra automorphisms leads to obtaining the classification. 

The classification of real three-dimensional coboundary Lie bialgebras up to Lie algebra automorphisms is approached in an algorithmic way (see \cite{BPS74,Ha79} for related topics). Although this problem has been treated somewhere else in the literature \cite{Farinati,Gomez}, we accomplish such a classification to illustrate our techniques, to fill in some gaps of previous works, and to give a new more geometrical approach.  Our results are written in detail in Table \ref{tabela3w} and sketched in Figure \ref{Sum}, where all equivalent reduced $r$-matrices are coloured in the same way. 

{\small \begin{figure}[h]
$\qquad$
\begin{minipage}{0.2\textwidth}
\begin{center}
\includegraphics[scale=0.20]{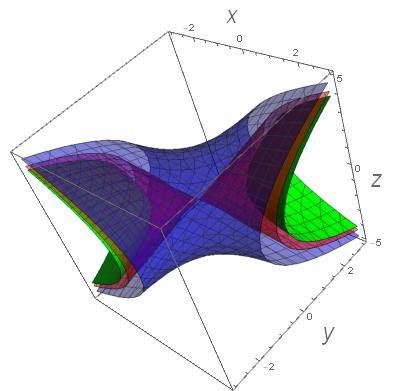}
\captionof*{figure}{$\mathfrak{sl}_2$}
\end{center}
\end{minipage}
$\quad$
\begin{minipage}{0.20\textwidth}
\begin{center}
\includegraphics[scale=0.18]{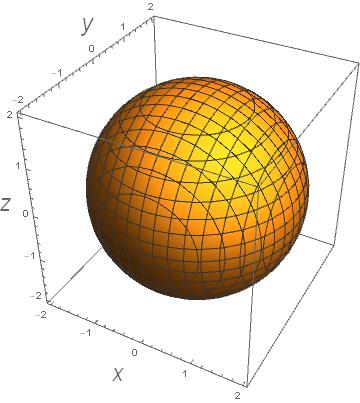}
\captionof*{figure}{$\mathfrak{su}_2$ }
\end{center}
\end{minipage}
$\quad$
\begin{minipage}{0.2\textwidth}
\begin{center}
\includegraphics[scale=0.25]{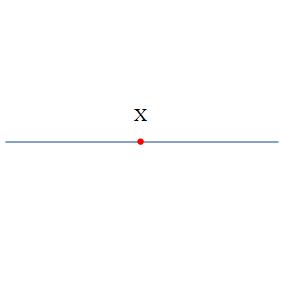}
\captionof*{figure}{$\mathfrak{h}_3$ }
\end{center}
\end{minipage}
$\,\,$
\begin{minipage}{0.2\textwidth}
\begin{center}
\includegraphics[scale=0.15]{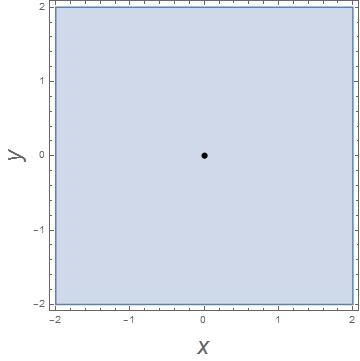}
\captionof*{figure}{$\mathfrak{r}'_{3,0}$}
\end{center}
\end{minipage}
\newline
$\qquad\quad$
\begin{minipage}{0.16\textwidth}
\begin{center}
\includegraphics[scale=0.15]{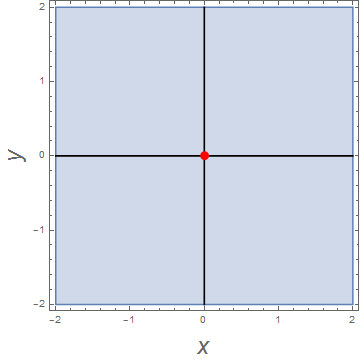}
\captionof*{figure}{$\mathfrak{r}_{3,-1}$}
\end{center}
\end{minipage}
\begin{minipage}{0.18\textwidth}
\begin{center}
\includegraphics[scale=0.15]{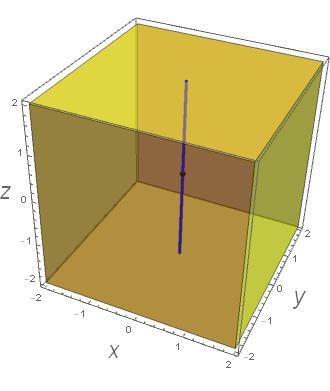}
\captionof*{figure}{$\mathfrak{r}_{3,1}$}
\end{center}
\end{minipage}
\begin{minipage}{0.22\textwidth}
\begin{center}
\includegraphics[scale=0.20]{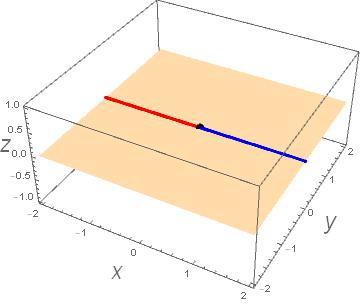}
\captionof*{figure}{$\mathfrak{r}_3$}
\end{center}
\end{minipage}
\begin{minipage}{0.18\textwidth}
\begin{center}
\includegraphics[scale=0.15]{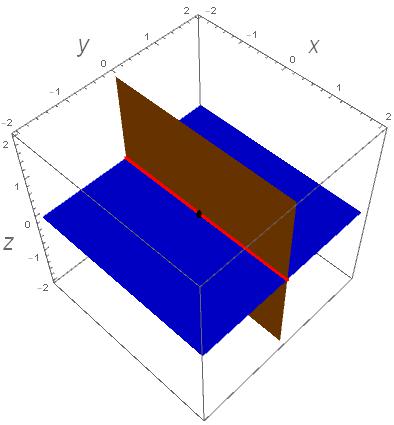}
\captionof*{figure}{$\mathfrak{r}_{3,\lambda}$}
\end{center}
\end{minipage}
\begin{minipage}{0.16\textwidth}
\begin{center}
\includegraphics[scale=0.15]{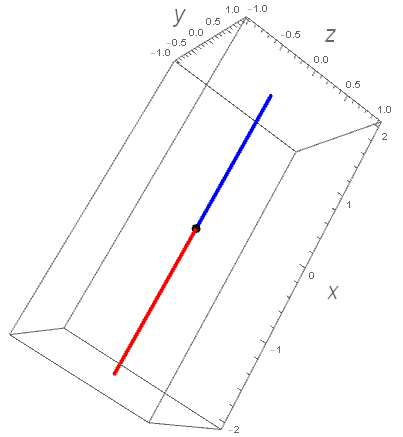}
\captionof*{figure}{$\mathfrak{r'}\!_{3,\lambda\neq 0}$}
\end{center}
\end{minipage}
\caption{The above drawings represent equivalent reduced $r$-matrices in the spaces $\Lambda^2_R\mathfrak{g}=\Lambda^2\mathfrak{g}/(\Lambda^2\mathfrak{g})^\mathfrak{g}$ for all three-dimensional Lie algebras. Families of equivalent reduced $r$-matrices related to a Lie algebra are denoted by  submanifolds of the same color. We assume that $\lambda\in(-1,1)$.}\label{Sum}
\end{figure}}

The structure of the paper goes as follows. Section 2 surveys the main notions on Lie bialgebras and presents the notation to be used. Section 3 introduces $\mathfrak{g}$-modules, proposes new structures related to them, and  gives several examples to be employed. Section 4 defines $\mathfrak{g}$-invariant maps and analyses its applications to Grassmann algebras. In particular, it provides methods to generate such maps on subspaces of Grassmann algebras through ad-invariant maps on Lie algebras. Section 5 studies properties of Killing-type metrics, namely metrics on a Grassmann algebra $\Lambda\mathfrak{g}$ whose definition is a generalization or extension, in the sense given in Section 4, of the standard Killing metric on $\mathfrak{g}$. The existence of $\mathfrak{g}$-invariant bilinear maps in $\mathfrak{g}$-modules is assessed in Section 6. Meanwhile, Section 7 proves that a root decomposition on a Lie algebra induces a new decomposition in its corresponding Grassmann algebra and the algebraic Schouten bracket respects this decomposition. The results of previous sections are employed in Section 8 to investigate the properties of $\mathfrak{g}$-invariant elements in $\Lambda\mathfrak{g}$ and to develop methods for their calculation. The problem of classification of coboundary Lie bialgebras is simplified in Section 9 to a certain quotient of their Grassmann algebras. Section 10 details several results on the existence of automorphisms of Lie algebras. Section 11 applies all previous methods to the classification problem up to Lie algebra automorphisms of three-dimensional coboundary Lie bialgebras. Finally, Section 12 resumes our achievements and sketches future lines of research.

\section{On Lie bialgebras and $r$-matrices}
Let us briefly survey the theory of Lie bialgebras (see \cite{Pressley,Yvette} for details) and establish the notation to be used. We employ a more geometric approach than in standard works, e.g. \cite{Pressley,Yvette}. We hereafter assume that all structures are real. Complex structures can be studied similarly.

Let $\mathcal{V}^m M$ be the space of $m$-vector fields on a manifold $M$. The {\it Schouten-Nijenhuis bracket} \cite{Marle,Va94} on $\mathcal{V}M:=\oplus_{m\in \mathbb{Z}}\mathcal{V}^m M$ is the unique bilinear map $[\cdot, \cdot]_{S}: \mathcal{V}M \times \mathcal{V}M \to \mathcal{V}M$ satisfying that: a) $[f,g]=0$ for arbitrary $f,g\in C^\infty(M)$, b) if $X$ is a vector field on $M$, then $[X,f]_S=Xf=-[f,X]_S$, c) we have 
\vskip -0.6cm
\begin{equation}\label{multi_sn}
[X_1 \wedge \ldots \wedge X_s, Y_1 \wedge \ldots \wedge Y_l]_{S} := \sum_{i=1}^s \sum_{j =1}^l (-1)^{i+j} [X_i, Y_j] \wedge X_1 \wedge \ldots  \wedge\widehat{X}_i \wedge \ldots  X_s \wedge Y_1\wedge \ldots \wedge\widehat{Y}_j \wedge \ldots \wedge Y_l,
\end{equation}
where $X_1,\ldots,X_s,Y_1,\ldots,Y_l$ are vector fields on $M$, the $\widehat{X}_i,\widehat{Y}_j$ are omitted in the exterior products in (\ref{multi_sn}), and $[\cdot, \cdot]$ is the Lie bracket of vector fields. If $X \in \mathcal{V}^k M$, $Y \in \mathcal{V}^l M$, $Z \in \mathcal{V} M$, then \cite{Marle}:
$$\begin{gathered}
[X, Y]_{S} = -(-1)^{(k-1)(l-1)} [Y, X]_{S},\qquad
[X, Y \wedge Z]_{S} = [X, Y]_{S} \wedge Z + (-1)^{(k-1)l} Y \wedge [X, Z]_{S},\\
[X, [Y, Z]_{S}]_{S} = [[X, Y]_{S}, Z]_{S} + (-1)^{(k-1)(l-1)} [Y, [X, Z]_{S}]_{S}.
\end{gathered}$$

Expression (\ref{multi_sn}) yields that the Schouten bracket of  left-invariant elements of $\mathcal{V}G$ for a Lie group $G$ is   left-invariant. Thus,  $\mathcal{V}G$ can be identified with the Grassman algebra $\Lambda\mathfrak{g}$ of the Lie algebra, $\mathfrak{g}$, of $G$. The Schouten-Nijenhuis bracket on $\mathcal{V}G$ can be restricted to left-invariant elements of $\mathcal{V}G$ giving rise to the \textit{algebraic Schouten bracket} on $\Lambda\mathfrak{g}$ \cite{Va94}, which is also denoted by $[\cdot,\cdot]_S$ for simplicity. 

A \textit{Lie bialgebra} is a pair $(\mathfrak{g}, \delta)$, where $\mathfrak{g}$ is a Lie algebra with a Lie bracket $[\cdot, \cdot]_\mathfrak{g}$ and  $\delta: \mathfrak{g} \to \Lambda^2 \mathfrak{g}$ is a linear map, called the \textit{cocommutator}, whose transpose 
$\delta^*: \Lambda^2\mathfrak{g}^* \to \mathfrak{g}^*$ is a Lie bracket on $\mathfrak{g}^*$ and 
\begin{equation}\label{cocycle_cond}
\delta([v_1,v_2]_\mathfrak{g}) = [v_1, \delta(v_2)]_{S} +[\delta(v_1),v_2]_{S}, \qquad \forall v_1,v_2 \in \mathfrak{g}.
\end{equation}

\textit{A Lie bialgebra homomorphism}  is a Lie algebra homomorphism $\phi: \mathfrak{g} \to \mathfrak{h}$ such that $(\phi \otimes \phi) \circ \delta_{\mathfrak{g}} = \delta_{\mathfrak{h}} \circ \phi$ for the cocommutators $\delta_\mathfrak{g}$ and $\delta_\mathfrak{h}$ of $\mathfrak{g}$ and $\mathfrak{h}$, respectively. A \textit{coboundary Lie bialgebra} is a Lie bialgebra $(\mathfrak{g}, \delta_r)$ such that $\delta_r(v) := [v, r]_S$ for an $r \in \Lambda^2 \mathfrak{g}$ and every $v \in \mathfrak{g}$. We call $r$ an \textit{$r$-matrix}.

To characterise those $r \in \Lambda^2 \mathfrak{g}$ giving rise to a cocommutator $\delta_r$, we need the following notions. The identification of $\mathfrak{g}$ with the Lie algebra of left-invariant vector fields on $G$ allows us to understand the tensor algebra $T(\mathfrak{g})$ as the algebra of left-invariant tensor fields on $G$ relative to the tensor product. This gives rise to a Lie algebra representation  $\textrm{ad}: v\in \mathfrak{g} \mapsto {\rm ad}_v\in \mathfrak{gl}\left(T(\mathfrak{g})\right)$, where ${\rm ad}_v(w)=\mathcal{L}_vw$ for every $w\in T(\mathfrak{g})$ and $\mathcal{L}_vw$ is the Lie derivative of the left-invariant tensor field $w$ relative to the left-invariant vector field $v$. This expression is more compact than the algebraic one appearing in standard works (cf. \cite{Pressley}). A $q \in T(\mathfrak{g})$ is called \textit{$\mathfrak{g}$-invariant} if $\mathcal{L}_vq=0$ for all $v \in \mathfrak{g}$. We denote the set of $\mathfrak{g}$-invariant elements of $T(\mathfrak{g})$ by $T(\mathfrak{g})^\mathfrak{g}$. The map ${\rm ad}$ admits a  restriction ${\rm ad}:\mathfrak{g}\rightarrow \mathfrak{gl}(\Lambda^m\mathfrak{g})$. We write  $(\Lambda^m\mathfrak{g})^{\mathfrak{g}}$ for the space of $\mathfrak{g}$-invariant $m$-vectors.

\begin{theorem}\label{thm1} 
The map $\delta_r: v \in \mathfrak{g} \mapsto [v, r]_S \in \Lambda^2\mathfrak{g}$, for $r\in \Lambda^2\mathfrak{g}$, is a cocommutator if and only if 
$[r, r]_{S} \in (\Lambda^3 \mathfrak{g})^{\mathfrak{g}}$.
\end{theorem}
The condition $[r, r]_{S} \in (\Lambda^3 \mathfrak{g})^{\mathfrak{g}}$ is called the \textit{modified classical Yang-Baxter equation} (mCYBE). The equation $[r, r]_{S} = 0$ is called the \textit{classical Yang-Baxter equation} (CYBE) and its solutions are called {\it triangular $r$-matrices}. Geometrically, triangular r-matrices amount to left-invariant Poisson bivectors on Lie groups \cite{Va94}. The next proposition establishes when two $r$-matrices induce the same coproduct. 
\begin{proposition}\label{prop:requiv} Two $r$-matrices $r_1,r_2\in \Lambda^2\mathfrak{g}$ satisfy that $\delta_{r_1}=\delta_{r_2}$ if and only if  $r_1-r_2\in(\Lambda^2\mathfrak{g})^{\mathfrak{g}}$.
\end{proposition}
\begin{proof} If $\delta_{r_1}=\delta_{r_2}$, then $[v, r_1]_S=[v, r_2]_S$ for every $v\in \mathfrak{g}$ and $r_1-r_2\in (\Lambda^2\mathfrak{g})^{\mathfrak{g}}$. The converse is immediate.
	\end{proof}

Proposition \ref{prop:requiv} shows that what really matters to the determination of coboundary Lie bialgebras is not $r$-matrices, but their equivalence classes in $\Lambda^2_R\mathfrak{g}=\Lambda^2\mathfrak{g}/(\Lambda^2\mathfrak{g})^\mathfrak{g}$.

\section{Structures on \texorpdfstring{$\mathfrak{g}$}{}-modules}
Let us discuss $\mathfrak{g}$-modules \cite{Pressley} and some new related structures that are necessary to our purposes. Subsequently, $GL(V)$  and $\mathfrak{gl}(V)$ stand for the Lie group of automorphisms and the Lie algebra of endomorphisms on the linear space $V$, respectively.

A {\it $\mathfrak{g}$-module} is a pair $(V,\rho)$, where $V$ is a linear space and $\rho:v\in \mathfrak{g}\mapsto \rho_v\in \mathfrak{gl}(V)$ is a Lie algebra morphism. A $\mathfrak{g}$-module $(V,\rho)$ will be represented just by $V$ and $\rho_v(x)$ will be written simply as $vx$ for any $v\in \mathfrak{g}$ and $x\in V$ if $\rho$ is understood from context. 

\begin{example}\label{ex:Lie_alg_der}
Let ${\rm ad}: v \in \mathfrak{g}\mapsto [v,\cdot]_\mathfrak{g} \in \mathfrak{gl}(\mathfrak{g})$ be the adjoint representation of $\mathfrak{g}$. Then, $(\mathfrak{g},{\rm ad})$ is a $\mathfrak{g}$-module \cite{Fulton}. Since each $[v,\cdot]_\mathfrak{g}$, with $v\in\mathfrak{g}$, is a derivation of the Lie algebra $\mathfrak{g}$ \cite{Fulton}, the map ${\rm ad}$ can be considered as a mapping ${\rm ad}:\mathfrak{g}\rightarrow {\rm Der}(\mathfrak{g})$, where ${\rm Der}(\mathfrak{g})$ is the Lie algebra of derivations on $\mathfrak{g}$. \demo
\end{example}

\begin{example}\label{ex:Lie_alg_aut} The group ${\rm Aut}(\mathfrak{g})$ of Lie algebra automorphisms of $\mathfrak{g}$ is the intersection of the Lie group $GL(\mathfrak{g})$ with $\Phi^{-1}(0)$ for $\Phi: F \in \mathfrak{gl}(\mathfrak{g})\mapsto [F(\cdot),F(\cdot)]_\mathfrak{g}-F[\cdot,\cdot]_\mathfrak{g}\in (\mathfrak{g}\otimes\mathfrak{g})^* \otimes \mathfrak{g}$. Hence, ${\rm Aut}(\mathfrak{g})$ is a closed subgroup of $GL(\mathfrak{g})$ and it becomes a Lie group \cite{DK00,Lee}. Let $\mathfrak{aut}(\mathfrak{g})$ be the Lie algebra of ${\rm Aut}(\mathfrak{g})$. The tangent map at ${\rm id}_\mathfrak{g}\in {\rm Aut}(\mathfrak{g})$ of the injection $\iota:{\rm Aut}(\mathfrak{g})\rightarrow GL(\mathfrak{g})$ induces a Lie algebra morphism $\widehat {\rm ad}:\mathfrak{aut}(\mathfrak{g})\simeq {\rm T}_{{\rm id}_\mathfrak{g}}{\rm Aut}(\mathfrak{g})\rightarrow \mathfrak{gl}(\mathfrak{g})\simeq {\rm T}_{{\rm id}_\mathfrak{g}}GL(\mathfrak{g})$ and $(\mathfrak{g},\widehat{\rm ad})$ becomes an $\mathfrak{aut}(\mathfrak{g}$)-module. 
\demo
\end{example}

\begin{example}\label{GrAl}
Each Grassmann algebra of a Lie algebra $\mathfrak{g}$ admits a $\mathfrak{g}$-module structure $(\Lambda\mathfrak{g},{\rm ad})$, where ${\rm ad}:v\in \mathfrak{g}\mapsto [v,\cdot]_{S}\in \mathfrak{gl}(\Lambda \mathfrak{g})$ is a Lie algebra homomorphism due to the properties of the algebraic Schouten bracket \cite{Va94}.
\demo
\end{example}

Example \ref{GrAl} can be understood as a consequence of the  Proposition \ref{Prop1} to be proved next. To understand this fact,  recall that every $T\in \mathfrak{gl}(V)$ gives rise to the mappings  $\Lambda^mT:\lambda \in \Lambda^mV\mapsto 0\in\Lambda^mV$ for $m\leq 0$, and the maps $\Lambda^mT\in \mathfrak{gl}(\Lambda^mV)$, for $m>0$, of the form
\begin{equation}\label{loc1}
\Lambda^mT:=\stackrel{m\,\, {\rm operators}}{\overbrace{T\otimes {\rm id}\otimes\ldots\otimes {\rm id}}}+\ldots+\stackrel{m\,\, {\rm operators}}{\overbrace{{\rm id}\otimes\ldots\otimes{\rm id}\otimes T}},
\end{equation}
where the tensor products of $m$ operators are restricted to $\Lambda^mV$ and ${\rm Id}$ is the identity on $V$. Moreover,  $\Lambda T:=\bigoplus_{m\in \mathbb{Z}}\Lambda^mT$ belongs to $\mathfrak{gl}(\Lambda V)$. 

\begin{proposition} \label{Prop1}
Let $(V, \rho)$ be a $\mathfrak{g}$-module. For every $m \in \mathbb{N}$, the pair $(\Lambda^m V,\Lambda^m\rho)$, where $\Lambda^m \rho: v \in \mathfrak{g} \mapsto \Lambda^m \rho_v \in \mathfrak{gl}(\Lambda^m V)$ and $(\Lambda V,\Lambda \rho:v\in \mathfrak{g}\mapsto \Lambda\rho_v\in \mathfrak{gl}(\Lambda V))$ are $\mathfrak{g}$-modules.
\end{proposition}
\begin{proof}
Proving that $(\Lambda^mV,\Lambda^m\rho)$ is a $\mathfrak{g}$-module reduces to showing that $\Lambda^m\rho$ is a Lie algebra homomorphism. Let us show that $\Lambda^m \rho_{[v_1,v_2]_{\mathfrak{g}}} = [\Lambda^m \rho_{v_1}, \Lambda^m \rho_{v_2}]_{\mathfrak{gl}(\Lambda^mV)}$ for every $v_1,v_2\in\mathfrak{g}$. As $(V, \rho)$ is a $\mathfrak{g}$-module and $\rho$ is therefore a Lie algebra morphism, it follows that
\begin{equation}\label{loc1}
\begin{split}
\Lambda^m \rho_{[v_1,v_2]_{\mathfrak{g}}} &= \rho_{[v_1, v_2]_{\mathfrak{g}}}\otimes {\rm id}\otimes\ldots\otimes {\rm id}+\ldots+ {\rm id}\otimes\ldots\otimes{\rm id}\otimes \rho_{[v_1,v_2]_{\mathfrak{g}}} \\
&= [\rho_{v_1}, \rho_{v_2}]_{\mathfrak{gl}(V)}\otimes {\rm id}\otimes\ldots\otimes {\rm id}+\ldots+{\rm id}\otimes\ldots\otimes{\rm id}\otimes [\rho_{v_1}, \rho_{v_2}]_{\mathfrak{gl}(V)}.
\end{split}
\end{equation}
If  $T_i:={\rm id} \otimes \ldots \otimes \stackrel{i-\textrm{th}}{\overbrace{T}} \otimes \ldots \otimes \textrm{id}$, for $T\in \mathfrak{gl}(V)$, then $T_i\circ S_j=S_j\circ T_i$ for $i \neq j$ and every $T,S\in \mathfrak{gl}(V)$. Thus,
\begin{equation}\label{loc2}
[\Lambda^m \rho_{v_1}, \Lambda^m \rho_{v_2}]_{\mathfrak{gl}(\Lambda^mV)}= [\rho_{v_1}, \rho_{v_2}]_{\mathfrak{gl}(V)}\otimes {\rm id}\otimes\ldots\otimes {\rm id}+\ldots+{\rm id}\otimes\ldots\otimes{\rm id}\otimes [\rho_{v_1}, \rho_{v_2}]_{\mathfrak{gl}(V)}
\end{equation}
for arbitrary $v_1, v_2 \in \mathfrak{g}$. Comparing (\ref{loc1}) and (\ref{loc2}), we get that $(\Lambda^mV,\Lambda^m\rho)$ is a $\mathfrak{g}$-module. It is immediate that $(\Lambda V,\Lambda\rho)$ is a $\mathfrak{g}$-module also.
\end{proof}

The following lemma is immediate.

\noindent
\begin{minipage}{0.65\textwidth}
\begin{lemma}\label{Lio} Let $(V,\rho)$ be a $\mathfrak{g}$-module and let $G$ be a connected Lie group with Lie algebra $\mathfrak{g}$. If $\Phi:G\rightarrow GL(V)$ is a Lie group morphism making commutative the diagram (\ref{diag}),  where $\exp_G$ and $\exp$ are exponential maps on $\mathfrak{g}$ and $\mathfrak{gl}(V)$ respectively, then $\Phi(G)$ is an immersed Lie subgroup of $GL(V)$ generated by the elements $\exp(\rho(\mathfrak{g}))$.
\end{lemma}
\end{minipage}
\begin{minipage}{0.35\textwidth}
\begin{center}
\vskip-0.5cm
\begin{equation}\label{diag}
\xymatrix{\mathfrak{g}\ar[rr]^{\rho\quad\quad}\ar[d]^{\exp_G}&&\mathfrak{gl}(V)\ar[d]^{\exp}\\G\ar[rr]^{\Phi}&&GL(V)}
\end{equation}
\end{center}
\end{minipage}

A Lie algebra homomorphism $\rho:\mathfrak{g}\rightarrow \mathfrak{gl}(V)$ gives rise to a Lie group morphism $\Phi:G\rightarrow GL(V)$, where $G$ is  connected and simply connected, so that (\ref{diag}) is commutative \cite{DK00}.  Since $\Phi(G)$ is generated by the elements $\exp(\rho(\mathfrak{g}))$, the group $\Phi(G)$ is the smallest group containing $\exp(\rho(\mathfrak{g}))$.  If (\ref{diag}) holds for other connected Lie group $G'$ with Lie algebra $\mathfrak{g}$, then $\Phi(G')$ is still generated by elements $\exp(\rho(\mathfrak{g}))$ and $\Phi(G')=\Phi(G)$. Hence, $\Phi(G')$ depends indeed only on $\rho$. Moreover, $\Phi(G')$ may not be an embedded submanifold of $GL(V)$: it is only an immersed Lie subgroup. These facts justify the following definition.

\begin{definition}
Given a $\mathfrak{g}$-module $(V,\rho)$, the {\it Lie group} of $(V, \rho)$ is the Lie subgroup $GL(\rho)$ of $GL(V)$ generated by $\exp (\rho(\mathfrak{g}))$. 
\end{definition}

\begin{proposition}\label{Inn} Let ${\rm Ad}:g\in G\mapsto {\rm Ad}_g\in GL(\mathfrak{g})$ be the adjoint action of a connected Lie group $G$ on its Lie algebra $\mathfrak{g}$. The Lie group of the $\mathfrak{g}$-module $(\mathfrak{g},{\rm ad})$ is  equal to ${\rm Ad}(G)$.
\end{proposition}
\noindent
\begin{minipage}{0.65\textwidth}
\begin{proof} Every connected Lie group $G$ with a Lie algebra $\mathfrak{g}$ is such that ${\rm ad}$ is the tangent map to ${\rm Ad}$ at the neutral element of $G$. Hence, one obtains the commutative diagram (\ref{diag2}). In view of Lemma \ref{Lio} and previous comments, $GL({\rm ad})={\rm Ad}(G)$. 
\end{proof}
\end{minipage}
\begin{minipage}{0.35\textwidth}
\begin{center}
 \vskip-0.6cm
 \begin{equation}\label{diag2}
 \xymatrix{\mathfrak{g}\ar[rr]^{{\rm ad}\quad\quad}\ar[d]^{\exp_G}&&\mathfrak{gl}(\mathfrak{g})\ar[d]^{\exp}\\{G}\ar[rr]^{{\rm Ad}}&&GL(\mathfrak{g})}
 \end{equation}
 \end{center}
 \end{minipage}

Since $GL({\rm ad})$ depends only on ${\rm ad}(\mathfrak{g})$, it makes sense to denote $GL({\rm ad})$ by Inn$(\mathfrak{g})$. Moreover, the space of inner automorphisms of $\mathfrak{g}$ is also given by ${\rm Ad}(G)$, which is equal to $ {GL}({\rm ad})$.

\begin{proposition}\label{Out} The Lie group of the $\mathfrak{aut}(\mathfrak{g})$-module $(\mathfrak{aut}(\mathfrak{g}),\widehat {\rm ad})$ is given by the connected component, ${\rm Aut}_c(\mathfrak{g})$, of the neutral element of ${\rm Aut}(\mathfrak{g})$.
\end{proposition}
\begin{proof} 
The inclusion $\iota:{\rm Aut}_c(\mathfrak{g})\hookrightarrow GL(\mathfrak{g})$ has a tangent map $\widehat{\rm ad}:\mathfrak{aut}(\mathfrak{g})\rightarrow \mathfrak{gl}(\mathfrak{g})$ at ${\rm id}_\mathfrak{g}\in {\rm Aut}_c(\mathfrak{g})$. 
\noindent
\begin{minipage}{0.65\textwidth}
This leads to the commutativity between the right and central columns  of the diagram aside. Let $\widetilde{{\rm Aut}}(\mathfrak{g})$ be the simply connected Lie group associated with $\mathfrak{aut}(\mathfrak{g})$. The commutativity of the left and central columns of the diagram aside comes from the properties of ${\rm Ad}$ for $\widetilde{\rm Aut}(\mathfrak{g})$ and $\widehat{{\rm ad}}$. From the commutativity of the\! diagram and using Lemma \ref{Lio},  it follows that $GL(\widehat{\rm ad})={\rm Ad}(\widetilde{{\rm Aut}}(\mathfrak{g}))={\rm Aut}_c(\mathfrak{g})$.
\end{minipage}
\begin{minipage}{0.35\textwidth}
\vskip-0.4cm
\begin{center}
\begin{equation*}\label{prop:diag}
\xymatrix{\mathfrak{aut}(\mathfrak{g})\ar[r]^(0.5){\widehat{\rm ad}}\ar[d]^{\exp_{\widetilde{{\rm Aut}}(\mathfrak{g})}}&\mathfrak{gl}(\mathfrak{g})\ar[d]^{\exp}&\mathfrak{aut}(\mathfrak{g})\ar[l]_{\widehat{{\rm ad}}}\ar[d]_{\exp_{{{\rm Aut}}_c(\mathfrak{g})}}\\\widetilde{{\rm Aut}}(\mathfrak{g})\ar[r]^{{\rm Ad}}&GL(\mathfrak{g})&{\rm Aut}_c(\mathfrak{g})\ar@{->}[l]_{\iota}}
\end{equation*}
\end{center}
\end{minipage}
\noindent \vskip -0.4cm
\end{proof}

Propositions \ref{Inn} and \ref{Out} show that the Lie groups of the $\mathfrak{g}$-modules described in them are ${\rm Inn}(\mathfrak{g})$ and ${\rm Aut}_c(\mathfrak{g})$. These Lie groups play a relevant role in the classification of Lie bialgebras up to Lie algebra automorphisms and previous  constructions will be employed to study this problem.
As a $\mathfrak{g}$-module, $(V,\rho)$ induces new ones $(\Lambda^m V,\Lambda^m\rho)$, the $GL(\rho)$ is related to $GL(\Lambda^m\rho)$ as explained next.
 
\begin{proposition}\label{Prop:invk}If $(V,\rho)$ is a $\mathfrak{g}$-module, $(\Lambda^m V,\Lambda^m\rho)$ satisfies that $GL(\Lambda^m\rho)=\left\{\Lambda^m T\,| \,\,T\in GL(\rho)\right\}.$

\end{proposition}
\begin{proof}By definition, $GL(\Lambda^m\rho)$ is generated by the composition of elements of $\exp(\Lambda^m\rho(\mathfrak{g}))$, namely
$$
\exp(\Lambda^m \rho_v) = \exp(\stackrel{m\,\,{\rm elements}}{\overbrace{\rho_v\otimes {\rm id}\otimes\ldots\otimes {\rm id}}}+\ldots+ \stackrel{m\,\, {\rm elements}}{\overbrace{{\rm id}\otimes\ldots\otimes{\rm id}\otimes \rho_v}}),\qquad \forall v\in \mathfrak{g}.
$$
As  operators $T_i$ commute for $i\neq j$ and any $T\in \mathfrak{gl}(\mathfrak{g})$, one has that
$\exp(\Lambda^m \rho_v) = \exp(\rho_v) \otimes \ldots \otimes \exp(\rho_v).$ Hence, $GL(\Lambda^m \rho)$ is generated by the composition of  operators $T \otimes \ldots \otimes T$ ($m$-times), where $T$ is a composition of operators $\exp(\rho_v)$ with $v\in \mathfrak{g}$. Since the $\exp(\rho_v)$ generate $GL(\rho)$, then $T$ is any element of $GL(\rho)$, which finishes the proof. 
\end{proof}

The following proposition will be employed to determine equivalent solutions to mCYBEs. 

\begin{proposition}\label{proporb}
	The dimension of the orbit $\mathcal{O}_w$ of the ${\rm Inn}(\mathfrak{g})$-action on $\Lambda^m\mathfrak{g}$ through $w\in \Lambda^m\mathfrak{g}$ is 
	$
	\dim  {\rm Im}\,\Theta^m_w,
	$
	with $\Theta^m_w:v\in \mathfrak{g}\mapsto [v, w]_S \in \Lambda^m\mathfrak{g}$.
\end{proposition}
\begin{proof}
	The adjoint action of $G$ on each $\Lambda^m \mathfrak{g}$ is given by $g \cdot w := \Lambda^m {\rm Ad}_g w$. Define $\exp(tv) =:g_t$, $g_1 :=g$ for $v\in\mathfrak{g}$. 
	Then, $\dim G\cdot w=\dim G-\dim G_w$, where $G_w$ is the isotropy group of $w\in \Lambda^m\mathfrak{g}$. The Lie algebra $\mathfrak{g}_w$ of $G_w$ is given by the $v$ of $\mathfrak{g}$ such that 
	$\frac{d}{dt}\bigg|_{t=0}\Lambda^m{\rm Ad}_{g_t}(w)=[v,w]_S=0$. Ths amounts to the fact that $v \in \ker\,\Theta^{m}_w$.
	Hence, $\dim \mathcal{O}_w = \dim \mathfrak{g}-\dim \mathfrak{g}_w=\dim {\rm Im}\, \Theta_{w}^{m}$.
\end{proof}

\section{The \texorpdfstring{$\mathfrak{g}$}{}-invariant maps on Grassmann algebras \texorpdfstring{$\Lambda\mathfrak{g}$}{}}

This section extends and analyses standard notions on Lie algebras, like the {\rm ad}-invariance, to $\mathfrak{g}$-modules. Special attention is paid to the extension to Grassmann algebras. Our findings will permit us to study  Lie bialgebras  in following sections.

\begin{definition}
A $k$-linear map $b:V^{\otimes k}\rightarrow \mathbb{R}$ is \textit{$GL(\rho)$-invariant} relative to a $\mathfrak{g}$-module $(V,\rho)$ if $T^* b = b$ for every $T \in GL(\rho)$, i.e. $b(Tx_1, \ldots, Tx_k) = b(x_1, \ldots, x_k)$ for every $x_1, \ldots, x_k \in V$.
\end{definition}

To characterise $GL(\rho)$-invariant maps, we will use the following  notion.

\begin{definition}\label{ginv}
 A $k$-linear map $b: V^{\otimes k} \to \mathbb{R}$ is {\it $\mathfrak{g}$-invariant} relative to the $\mathfrak{g}$-module $(V, \rho)$ if 
\begin{equation}\label{ginva}
b(\rho_v(x_1), \ldots, x_k) + \ldots+b(x_1,\ldots,\rho_v(x_k))= 0,\qquad \forall v\in \mathfrak{g},\quad \forall x_1,\ldots,x_k\in V.
\end{equation}
\end{definition}

\begin{example} 
The Killing metric on $\mathfrak{g}$, namely $\kappa_{\mathfrak{g}}(v_1, v_2) := \textrm{tr}(\textrm{ad}_{v_1} \circ \textrm{ad}_{v_2})$ with $v_1,v_2 \in \mathfrak{g}$, satisfies that $\kappa_{\mathfrak{g}}(\textrm{ad}_{v}v_1, v_2) + \kappa_{\mathfrak{g}}(v_1, \textrm{ad}_{v}v_2) = 0$ for all $v, v_1, v_2 \in \mathfrak{g}$ \cite{Fulton} . For this reason Killing metrics are called {\it ad-invariant}. In view of Definition \ref{ginv}, the Killing metric is $\mathfrak{g}$-invariant with respect to $(\mathfrak{g}, \textrm{ad})$. 
\demo
\end{example}

Thus, $\mathfrak{g}$-invariance can be interpreted as an extension of ad-invariance to $\mathfrak{g}$-modules. As shown in Proposition \ref{prop:glrho_ginv}, the invariance of a $k$-linear map on a $\mathfrak{g}$-module $(V,\rho)$ relative to $GL(\rho)$ can be characterised by the $\mathfrak{g}$-invariance of the $k$-linear map. The proof is not detailed as it is quite immediate.

\begin{proposition}\label{prop:glrho_ginv}
A $k$-linear map $b:V^{\otimes k} \to \mathbb{R}$ is $GL(\rho)$-invariant relative to a $\mathfrak{g}$-module $(V,\rho)$ if and only if $b$ is $\mathfrak{g}$-invariant relative to $(V,\rho)$.
\end{proposition}

Subsequently, we assume that $\{v_1,\ldots,v_r\}$ is a basis of $\mathfrak{g}$ and define $v_{J}:=v_{J(1)}\wedge \ldots \wedge v_{J(m)}$, where $J:=(J(1),\ldots,J(m))$ with $J(1), \ldots, J(m) \in {1,\ldots, r}$ represents a multi-index of length $|J|=m$, the $S_m$ is the permutation group of $m$ elements, and ${\rm sg}(\sigma)$ stands for the sign of the permutation $\sigma \in S_m$.  

 \begin{theorem}\label{extension}
Every $\mathfrak{g}$-invariant $k$-linear map $b:V^{\otimes k}\rightarrow \mathbb{R}$ relative to a $\mathfrak{g}$-module $V$ induces a $\mathfrak{g}$-invariant $k$-linear map, $b_{\Lambda V}$, on $\Lambda V$ relative to the induced $\mathfrak{g}$-module on $\Lambda V$ by imposing that
 \begin{enumerate}
 \item the spaces $\Lambda^mV$, with $m\in\mathbb{Z}$, are orthogonal between themselves relative to $b_{\Lambda V}$,
\item $b_{\Lambda V}(1,\ldots,1)=1$, 
 
 \item  the restriction, $b_{\Lambda^mV}$, of $b_{\Lambda V}$ to $\Lambda^m V$, with $m\in\mathbb{N}$, satisfies
\vskip -0.4cm{
\begin{equation*}\label{Lambdag}
b_{\Lambda^{m} V}(v_{J_1},\ldots ,v_{J_k}):=\!\!\!\!\!\!\!\sum_{\sigma_1,\ldots,\sigma_{k} \in S_m}\!\!\!\!\!\!{\rm sg }(\sigma_1\ldots\sigma_k)\frac{1}{m!}\prod_{r=1}^m b\left( v_{J_1(\sigma_1^{-1}(r))},\ldots ,v_{J_k(\sigma_{k}^{-1}(r))}\right).
\end{equation*}}
 \end{enumerate}
\end{theorem}
\begin{proof} Since $1\in \Lambda^0V$ and the decomposable elements $v_J=v_{J(1)}\wedge\ldots\wedge v_{J(m)}$ span $\Lambda V$ and $b_{\Lambda V}$ is $k$-linear, then the conditions 1, 2, and 3 establish $b_{\Lambda V}$. The condition 3 establishes a well-defined value of $b_{\Lambda V}$ independently of the representative for each $v_{J_s}$, with $s\in \overline{1,k}$. Indeed, defining $\sigma v_{J}:=v_{J(\sigma^{-1}(1))}\wedge\ldots\wedge v_{J(\sigma^{-1}(m))}$ and $\tilde{\tilde{\sigma}}_j :=\tilde\sigma_j\cdot \sigma_j $, we obtain
	\vskip -0.4cm
$$
\begin{aligned}
b_{\Lambda^m V}(\tilde{\sigma}_1v_{J_1},\ldots, \tilde{\sigma}_kv_{J_k})&=\!\!\!\!\sum_{\sigma_1,\ldots,\sigma_k\in S_m}\!\!\!\!\!\!{\rm sg }(\sigma_1\ldots\sigma_k)\frac{1}{m!}\prod_{r=1}^mb\left(v_{J_1( \sigma_1^{-1}\tilde{\sigma}_1^{-1} (r))},\ldots,v_{J_k( \sigma_k^{-1}\tilde{\sigma}_k^{-1} (r))}\right)\\
&=\!\!\!\!\!\sum_{\tilde{\tilde{\sigma}}_1,\ldots,\tilde{\tilde{\sigma}}_k\in S_m}\!\!\!\!\!{\rm sg }(\tilde{\tilde{\sigma}}_1 \ldots\tilde{\tilde{\sigma}}_k){\rm sg }(\tilde\sigma_1\ldots\tilde\sigma_k)\frac{1}{m!}\prod_{r=1}^mb\left(v_{J_1 (\tilde{\tilde{\sigma}}_1^{-1}(r))},\ldots,v_{J_k (\tilde{\tilde{\sigma}}_k^{-1}(r))}\right)\\&={\rm sg }(\tilde\sigma_1\ldots\tilde \sigma_k)b_{\Lambda^m V}(v_{J_1},\ldots, v_{J_k}).
\end{aligned}
$$

Let us prove that $b_{\Lambda V}$ is $\mathfrak{g}$-invariant relative to the $\mathfrak{g}$-module structure on $\Lambda V$ induced by the $\mathfrak{g}$-module structure on $V$. Proposition \ref{prop:glrho_ginv} yields that the $\mathfrak{g}$-invariance of $b_{\Lambda V}$ is inferred from its $GL(\Lambda\rho)$-invariance. This also reduces to the $GL(\Lambda^m\rho)$-invariance of the restrictions $b_{\Lambda^mV}$ for $m\in \overline{\mathbb{N}}=\mathbb{N}\cup \{0\}$. Using the $GL(\rho)$-invariance of $b$ and defining $e^{\rho_v}:=\exp(\rho_v)$ for every $v\in \mathfrak{g}$ and $m\in \mathbb{N}$, we get
\vskip -0.2cm
\begin{multline*}
b_{\Lambda^mV}(\Lambda^m e^{\rho_v}(v_{J_1}),\ldots,\Lambda^m e^{\rho_v}(v_{J_k}))
= \!\!\!\!\!\!\sum_{\sigma_1,\ldots,\sigma_{k} \in S_m}\!\!\!\!\!\!\!\!\!\!{\rm sg }(\sigma_1\ldots\sigma_k)\frac{1}{m!}\prod_{r=1}^m b\left( e^{\rho_v}v_{J_1(\sigma_1^{-1}(r))},\ldots , e^{\rho_v}v_{J_k(\sigma_{k}^{-1}(r))}\right)\\=\!\!\!\!\!\!\sum_{\sigma_1,\ldots,\sigma_{k} \in S_m}\!\!\!\!\!\!\!\!\!\!{\rm sg }(\sigma_1 \ldots\sigma_k)\frac{1}{m!}\prod_{r=1}^m b\left(v_{J_1(\sigma_1^{-1}(r))},\ldots , v_{J_k(\sigma_{k}^{-1}(r))}\right)= b_{\Lambda^mV}(v_{J_1},\ldots, v_{J_k}).
\end{multline*}
Since the invariance of $b_{\Lambda^0V}$ is obvious, $b_{\Lambda V}$ is $GL(\Lambda\rho)$-invariant and Proposition \ref{prop:glrho_ginv} ensures that is $\mathfrak{g}$-invariant.
\end{proof}

Since each Killing metric is $\mathfrak{g}$-invariant (relative to ($\mathfrak{g},{\rm ad})$), it can be extended to each $\Lambda^m\mathfrak{g}$. Its extensions  to $\Lambda^2\mathfrak{g}$ and $\Lambda^3\mathfrak{g}$ are called the {\it double} and {\it triple Killing metrics} of $\mathfrak{g}$, respectively.

Next corollary gives an immediate consequence of Proposition \ref{prop:glrho_ginv} and the proof of Theorem \ref{extension}. 

\begin{corollary}\label{ExtInvRep} Let $b$ be a $\mathfrak{g}$-invariant $k$-linear map on $V$ and let $T\in GL(\rho)$. Then, $b_{\Lambda^mV}$ is invariant with respect to $GL(\Lambda^m \rho)$, i.e. $b_{\Lambda^mV}(\Lambda^m T\cdot , \ldots, \Lambda^m T\cdot )=b_{\Lambda^mV}( \cdot ,\ldots, \cdot )$.
\end{corollary}
As shown next, certain extensions of a $\mathfrak{g}$-invariant metric  are trivial and, therefore, useless.

\begin{proposition} If $b$ is a $\mathfrak{g}$-invariant $k$-linear map on $V$, then $b_{\Lambda^mV}=0$ for $m >1$ and odd $k>1$. 
\end{proposition}
\begin{proof} Let us first prove that we can gather the summands appearing in  $b_{\Lambda^mV}$ into families that sum up to zero. We introduce the equivalence relation on $S_m^k := {S_m \times \stackrel{k\,\,{\rm times}}{\ldots} \times S_m}$ given by
$$
(\sigma_1, \ldots, \sigma_k) \equiv (\tilde\sigma_1, \ldots, \tilde\sigma_k) \iff \exists \sigma \in S_m: (\tilde\sigma_1, \ldots, \tilde\sigma_k) =  (\sigma \sigma_1, \ldots, \sigma \sigma_k).
$$
Let $[(\sigma_1,\ldots,\sigma_k)]$ be the equivalence class of $(\sigma_1,\ldots,\sigma_k)\in S_m^k$ and let $\mathcal{R}$ be the space of equivalence classes. The map $b_{\Lambda^mV}$ can be written as 
\begin{equation*}
b_{\Lambda^{m}V}(v_{J_1},\ldots ,v_{J_k}):=\sum_{[{\bf w}]\in \mathcal{R}}\sum_{(\sigma_1,\ldots,\sigma_{k}) \in [{\bf w}]}\!\!\!\!\!\!\!\!\!\!{\rm sg }(\sigma_1\ldots\sigma_k)\frac{1}{m!}\prod_{r=1}^m b\left( v_{J_1(\sigma_1^{-1}(r))},\ldots ,v_{J_k(\sigma_{k}^{-1}(r))}\right).
\end{equation*}
Since every equivalence class in $\mathcal{R}$ is of the form $[(\sigma_1,\ldots,\sigma_k)]=\{(\sigma\sigma_1,\ldots,\sigma\sigma_k):\sigma \in S_m\}$, one has
\begin{equation*}
b_{\Lambda^{m} V}(v_{J_1},\ldots ,v_{J_k}):=\sum_{[(\sigma_1 ,\ldots,\sigma_k)]\in \mathcal{R}}\sum_{\sigma\in S_m}\!\!{\rm sg }(\sigma\sigma_1 \ldots\sigma\sigma_k)\frac{1}{m!}\prod_{r=1}^m b\left( v_{J_1(\sigma_1^{-1}\sigma^{-1}(r))},\ldots ,v_{J_k(\sigma_{k}^{-1}\sigma^{-1}(r))}\right).
\end{equation*}
Let us show that the above sum vanishes for every equivalence class of $\mathcal{R}$.
First,
$$
\prod_{r=1}^mb\left(v_{J_1({\sigma}_1^{-1} \sigma^{-1} (r))},\ldots,v_{J_k({\sigma}_k^{-1} \sigma^{-1} (r))}\right)=\prod_{r=1}^mb\left(v_{J_1({\sigma}_1^{-1} (r))},\ldots,v_{J_k({\sigma}_k^{-1} (r))}\right).
$$
Let us define ${\rm sg}({\bf w}):={\rm sg}(\sigma_1\ldots\sigma_m)$. Then
$$
\prod_{r=1}^m{\rm sg }(\sigma {\bf w})b\left(v_{J_1({\sigma}_1^{-1} \sigma^{-1} (r))},\ldots,v_{J_k({\sigma}_k^{-1} \sigma^{-1} (r))}\right)=\prod_{r=1}^m{\rm sg }(\sigma {\bf w})b\left(v_{J_1({\sigma}_1^{-1}(r))},\ldots,v_{J_k({\sigma}_k^{-1} (r))}\right)
$$
and ${\rm sg}(\sigma {\bf w})={\rm sg}(\sigma)^{k}{\rm sg}({\bf w})={\rm sg}(\sigma){\rm sg}(\bf w)$ since $k$ is odd. Therefore,
\begin{equation}\label{eq1}
{\rm sg}(\sigma)^{k}{\rm sg }({\bf w})\prod_{r=1}^mb\left(v_{J_1(\sigma_1^{-1}(r))},\ldots,v_{J_k(\sigma_k^{-1}(r))}\right)=
{\rm sg}(\sigma){\rm sg }({\bf w})\prod_{r=1}^mb\left(v_{J_1(\sigma_1^{-1} (r))},\ldots,v_{J_k(\sigma_k^{-1}(r))}\right).
\end{equation}
Every equivalence class of $\mathcal{R}$ has $m!$ elements. All of them have the same absolute value. Half of them is odd and the other half is even. Hence, 
$$
\sum_{\sigma\in S_m}\!\!{\rm sg }(\sigma {\bf w})\frac{1}{m!}\prod_{r=1}^m b\left(v_{J_1(\sigma_1^{-1}(r))},\ldots , v_{J_k(\sigma_{k}^{-1}(r))}\right)=0\,\,\Longrightarrow \,\,b_{\Lambda^mV}=0.
$$\vskip-0.6cm
\end{proof}

\begin{example} 
Consider the Lie algebra $\mathfrak{su}_2$ and its Killing form $\kappa_{\mathfrak{su}_2}$, which is a $\mathfrak{su}_2$-invariant, bilinear, symmetric map on $\mathfrak{su}_2$. Take the basis $\{e_1, e_2, e_3\}$ of $\mathfrak{su}_2$ given in Table \ref{tabela3w}. Theorem \ref{extension}  extends $\kappa_{\mathfrak{su}_2}$ to the double- and triple-Killing metrics $\kappa_{\Lambda^2 \mathfrak{su}_2}$, $\kappa_{\Lambda^3 \mathfrak{su}_2}$ on $\Lambda^2\mathfrak{su}_2$ and $\Lambda^3\mathfrak{su}_2$, respectively. In the bases $\{e_{12}, e_{13}, e_{23}\}$, $\{e_{123}\}$ for the spaces $\Lambda^2\mathfrak{su}_2, \Lambda^3 \mathfrak{su}_2$ (see Table \ref{tabela3w}), we obtain
	$$
	[\kappa_{\mathfrak{su}_2}]:=\left(\begin{array}{ccc}
	-2&0&0\\
	0&-2&0\\
	0&0&-2\\
	\end{array}\right),\qquad 
	[\kappa_{\Lambda^2\mathfrak{su}_2}]:=\left(\begin{array}{ccc}
	4&0&0\\
	0&4&0\\
	0&0&4\\
	\end{array}\right),\qquad \qquad 
	[\kappa_{\Lambda^3\mathfrak{su}_2}]:=\left(-8\right).
	$$\vskip-0.6cm\demo
\end{example}

The previous example shows that the Killing metric and its extensions to $\Lambda^2\mathfrak{su}_2$ and $\Lambda^3\mathfrak{su}_2$ are simultaneously diagonal and non-degenerate. The corollary below provides an explanation of this fact.

 \begin{corollary}\label{Cor:Diatwoform} If $b$ is a symmetric $\mathfrak{g}$-invariant $k$-linear mapping on an $r$-dimensional $\mathfrak{g}$-module $V$, then $b_{\Lambda V}$ is symmetric. If $b$ diagonalizes in the basis $\{e_i\}_{i\in \overline{1,r}}$, then $b_{\Lambda^mV}$ diagonalizes in the basis $\{e_J\}_{|J|=m}$. Additionally, $b$ is non-degenerate if and only if $b_{\Lambda V}$ is so. 
\end{corollary}
\begin{proof} If $b$ is a symmetric $\mathfrak{g}$-invariant $k$-linear mapping on $V$, then Theorem \ref{extension} ensures that $b_{\Lambda V}$ is symmetric on the elements of a basis $v_J$ of $\Lambda V$. Indeed, the condition 3 guarantees the symmetry of $b_{\Lambda^m V}$ on decomposable elements of $\Lambda^m V$, $m\in \mathbb{N}$, whereas the condition 2 ensures the same for $m=0$.  Since $b_{\Lambda V}$ is additionally multilinear, it becomes symmetric on the whole $\Lambda V$.

If $b$ is bilinear and symmetric, it can always be put into diagonal form in a certain basis $\{e_1,\ldots, e_r\}$ for $V$. This gives rise to a basis $\{e_J\}$ of $\Lambda V$. Using the expression for $b_{\Lambda V}$, we see that this metric also becomes diagonal. The elements on the diagonal read
	$
	\prod_{j=1}^{|J|} b(e_{J(j)},e_{J(j)})$ for every multi-index $J$. Thus, $b$ is non-degenerate if and only if the induced symmetric metric $b_{\Lambda^mV}$ on each $\Lambda^mV$ is so as well. 
	\end{proof}
    
\begin{example}\label{ExExsl2}Consider the Lie algebra $\mathfrak{sl}_2$ with a basis $\{e_1,e_2,e_3\}$ satisfying the commutation relations in Table \ref{tabela3w}. In the induced bases $\{e_{12},e_{13},e_{23}\}$ and $\{e_{123}\}$ in $\Lambda^2\mathfrak{sl}_2$ and $\Lambda^3\mathfrak{sl}_{2}$, respectively, one has
\begin{equation}\label{sl2A}
[\kappa_{\mathfrak{sl}_{2}}] \!=\! \left(\begin{array}{ccc}
	2&0&0\\
	0&0&2\\
	0&2&0\\
	\end{array}\right),\qquad
[\kappa_{\Lambda^2 \mathfrak{sl}_{2}}]\! =\!\left(\begin{array}{ccc}
	0&4&0\\
	4&0&0\\
	0&0&-4\\
	\end{array}\right),\qquad\, [\kappa_{\Lambda^3\mathfrak{sl}_{2}}]\!=\!(-8).
\end{equation}
Since $\mathfrak{sl}_{2}$ is simple, the Cartan criterion states that $\kappa_{\mathfrak{sl}_2}$ is non-degenerate. Then, Corollary \ref{Cor:Diatwoform} ensures that $\kappa_{\Lambda^2 \mathfrak{sl}_{2}}$ and $\kappa_{\Lambda^3 \mathfrak{sl}_{2}}$ must be non-degenerate. This agrees with their expressions showed in (\ref{sl2A}).
\demo
\end{example}

\section{Killing-type metrics}
This section describes the invariance properties of certain multilinear metrics on the spaces $\Lambda^m\mathfrak{g}$ induced by Killing metrics. Our methods give rise to metrics invariant under the action of ${\rm Aut}(\mathfrak{g})$, which will be of interest in the description of coboundary cocommutators in Sections 7, 9, and 11.

\begin{proposition} The Killing metric $\kappa_\mathfrak{g}$ on $\mathfrak{g}$ is $\mathfrak{aut}(\mathfrak{g})$-invariant.
\end{proposition}
\begin{proof}In view of Proposition \ref{prop:glrho_ginv}, this proposition amounts to proving that  $\kappa_\mathfrak{g}$ is $GL(\widehat{{\rm ad}})$-invariant, which in turn means that 
$
\kappa_{\Lambda\mathfrak{g}}(\Lambda T\cdot ,\Lambda T\cdot )=\kappa_{\Lambda\mathfrak{g}}(\cdot ,\cdot)$ for every $T\in {\rm Aut}_c(\mathfrak{g})$, where ${\rm Aut}_c(\mathfrak{g})$ stands for the connected part of the neutral element of ${\rm Aut}(\mathfrak{g})$. The Killing metric is invariant relative to the action of ${\rm Aut}(\mathfrak{g})$ \cite{Ha15}. If $v_{J_1},v_{J_2}$ are decomposable elements of $\Lambda^m \mathfrak{g}$, then
{
\begin{multline*}
\kappa_{\Lambda \mathfrak{g}}(\Lambda^m Tv_{J_1},\Lambda^m Tv_{J_2}):=\!\!\!\!\!\!\sum_{\sigma_1, \sigma_2\in S_m}\!\!\!\!{\rm sg }(\sigma_1\sigma_2)\frac1{m!}\prod_{r=1}^m \kappa_\mathfrak{g}\left( Tv_{J_1(\sigma_1^{-1}(r))},Tv_{J_2(\sigma_{2}^{-1}(r))}\right)\\
=\!\!\!\!\!\!\sum_{\sigma_1, \sigma_2\in S_m}\!\!\!\!{\rm sg }(\sigma_1\sigma_2)\frac 1{m!}\prod_{r=1}^m \kappa_\mathfrak{g}\left( v_{J_1(\sigma_1^{-1}(r))},v_{J_2(\sigma_{2}^{-1}(r))}\right)=\kappa_{\Lambda \mathfrak{g}}(v_{J_1},v_{J_2}).
\end{multline*}}Since $\kappa_{\Lambda\mathfrak{g}}$ is bilinear and the above is satisfied for decomposable elements of $\Lambda^m\mathfrak{g}$, which span $\Lambda^m\mathfrak{g}$, the mapping $\kappa_{\Lambda^m \mathfrak{g}}$ is invariant relative to the action of ${\rm Aut}(\mathfrak{g})$ on $\Lambda^m \mathfrak{g}$. Since this fact is true for every $m$ and the spaces $\Lambda^m\mathfrak{g}$ for different $m$ are orthogonal relative to $\kappa_{\Lambda\mathfrak{g}}$, the proposition follows.
\end{proof}

Since $\kappa_{\Lambda\mathfrak{g}}$ is invariant under the maps $\Lambda T$ with $T\in {\rm Aut}(\mathfrak{g})$, it is therefore invariant under $\Lambda T$ with $T\in {\rm Inn}(\mathfrak{g})$. In view of Proposition \ref{prop:glrho_ginv}, the $\kappa_{\Lambda\mathfrak{g}}$ is also $\mathfrak{g}$-invariant.

\begin{proposition}
The $k$-linear  symmetric map $b(v_1, \ldots, v_k) := \sum_{\sigma \in S_m} {\rm Tr}\left({\rm ad}_{v_{\sigma (1)}} \circ \ldots \circ {\rm ad}_{v_{\sigma (k)}}\right)$ for every $v_1, \ldots, v_k \in \mathfrak{g}$ is $\mathfrak{aut}(\mathfrak{g})$-invariant.
\end{proposition}
\begin{proof}
From Proposition \ref{prop:glrho_ginv}, the map $b$ is $\mathfrak{aut}(\mathfrak{g})$-invariant if and only if $T^* b = b$ for every $T\in {\rm Aut}_c(\mathfrak{g})$. If $T \in {\rm Aut}(\mathfrak{g})$, then
\begin{equation*}
{\rm ad}_{Tv_1}v_2 = [Tv_1, v_2] = [Tv_1, TT^{-1}v_2] = T[v_1, T^{-1}v_2] = T \circ {\rm ad}_{v_1} \circ T^{-1} v_2,\qquad \forall v_1,v_2\in \mathfrak{g}.
\end{equation*}
For arbitrary $v_1,\ldots, v_k \in \mathfrak{g}$, one gets
\begin{equation*}
T^* b(v_1, \ldots, v_k)\!\! =\!\!\! \sum_{\sigma \in S_k} {\rm Tr}\left({\rm ad}_{Tv_{\sigma (1)}} \circ \ldots \circ {\rm ad}_{Tv_{\sigma (k)}}\right)\!
=\!\!\!  \sum_{\sigma \in S_k} {\rm Tr}(T \circ {\rm ad}_{v_{\sigma (1)}} \circ \ldots \circ {\rm ad}_{v_{\sigma (k)}}\circ T^{-1}) \!\!=\!\! b(v_1, \ldots, v_k).
\end{equation*}\vskip-0.6cm
\end{proof}

Following the idea of the proof of the latest proposition, we obtain the following corollary.

\begin{corollary}
The $k$-linear totally anti-symmetric map $b:\mathfrak{g}^k\rightarrow \mathbb{R}$ given by
$
b(v_1, \ldots, v_k) := \sum_{\sigma \in S_k} {\rm sg}(\sigma){\rm Tr}({\rm ad}_{v_{\sigma (1)}} \circ \ldots \circ {\rm ad}_{v_{\sigma (k)}})$, for all $v_1, \ldots, v_k \in \mathfrak{g}
$ 
is $\mathfrak{aut}(\mathfrak{g})$-invariant with respect to $(\mathfrak{aut}(\mathfrak{g}), \widehat{{\rm ad}})$.
\end{corollary}

Let us  prove that a polynomial Casimir element of order $k$, i.e. an element $C \in S(\mathfrak{g}^{\otimes k})$, satisfying that $\mathcal{L}_vC = 0$ for every $v\in \mathfrak{g}$, gives rise to a $\mathfrak{g}$-invariant $k$-linear symmetric map on $\mathfrak{g}$. Recall that $\kappa_{\mathfrak{g}}$ leads to a map $\widetilde{\kappa}_\mathfrak{g}: v\in \mathfrak{g} \mapsto \kappa_{\mathfrak{g}}(v,\cdot)\in \mathfrak{g}^*$ and there exists a natural isomorphism $\mathfrak{g}^{\otimes k} \simeq \left[(\mathfrak{g}^*)^{\otimes k}\right]^{*}$.

\begin{theorem}\label{PolPol}
Every polynomial Casimir element $C$ of order $k$ on a Lie algebra $\mathfrak{g}$ induces a $\mathfrak{g}$-invariant $k$-linear symmetric map on $\mathfrak{g}$ given by $b(v_1,\ldots,v_k):=C(\widetilde{\kappa}_\mathfrak{g}(v_1),\ldots,\widetilde{\kappa}_\mathfrak{g}(v_k))$ for every $v_1, \ldots, v_k \in \mathfrak{g}$.
\end{theorem}
\begin{proof} We have
$
\sum_{j=1}^{k} b(\textrm{ad}_v v_j,v_1, \ldots,\widehat{v_j},\ldots, v_k) = \sum_{j=1}^{k} C(\widetilde \kappa_\mathfrak{g}(\textrm{ad}_v v_j),\widetilde \kappa_\mathfrak{g}(v_1), \ldots, \widehat{\kappa_\mathfrak{g}(v_j)},  \ldots, \widetilde \kappa_\mathfrak{g}(v_k)).
$
 Since $\kappa_\mathfrak{g}$ is $\mathfrak{g}$-invariant, one gets that 
 $$
[{\rm ad}_v^*\circ\widetilde{\kappa}_\mathfrak{g}(v_1)](v_2)\!=\!\widetilde{\kappa}_\mathfrak{g}(v_1)({\rm ad}_vv_2)\!=\!\kappa_{\mathfrak{g}}(v_1,{\rm ad}_vv_2)\!=\!-\kappa_{\mathfrak{g}}({\rm ad}_vv_1,v_2)\!=\!-[\widetilde{\kappa}_\mathfrak{g}\circ {\rm ad}_v(v_1)](v_2),\quad \!\!\!\forall v,v_1,v_2\in\mathfrak{g}.
 $$
 Hence,
 $\widetilde\kappa_\mathfrak{g}\circ {\rm ad}_v=-{\rm ad}_v^*\circ \widetilde\kappa_\mathfrak{g}$ for every $v\in \mathfrak{g}$. As $C$ is a Casimir element, $\mathcal{L}_vC=0$, which along with the above expression and the fact that $\mathcal{L}_v\theta=-{\rm ad}^*_v\theta$ for every $\theta \in \mathfrak{g}^*$ (where $\theta$ can be understood also as a left-invariant one-form on a Lie group $G$ with Lie algebra $\mathfrak{g}$) gives us, for all $ v_1,\ldots,v_k,v\in\mathfrak{g}$, that
 \vskip -0.7cm
\begin{multline*}
\!\sum_{j=1}^{k} b( \textrm{ad}_v v_j,v_1, \ldots,\widehat{v_j}, \ldots, v_k)\! =\!-\!\sum_{j=1}^{k} \!C({\rm ad}^*_v\widetilde\kappa_\mathfrak{g}(v_j),\widetilde\kappa_\mathfrak{g}(v_1), \ldots, \widehat{\widetilde{\kappa}_\mathfrak{g}(v_j)}, \ldots, \widetilde\kappa_\mathfrak{g}(v_k)) 
\!\\=\!(\mathcal{L}_v C)(\widetilde\kappa_\mathfrak{g}(v_1), \ldots, \widetilde\kappa_\mathfrak{g}(v_k)) \!=\! 0.
\end{multline*}
 \vskip -0.7cm
\end{proof}

If $\mathfrak{g}$ is semi-simple, then the proof of Theorem \ref{PolPol} can be reversed and a $\mathfrak{g}$-invariant $k$-linear symmetric amounts to a Casimir element. If $\mathfrak{g}$ is not semi-simple, $\widetilde\kappa_\mathfrak{g}$ is not invertible and $\mathfrak{g}$-invariant multilinear symmetric maps may be more versatile, as they not need to come from Casimir elements.

\section{On the existence of \texorpdfstring{$\mathfrak{g}$}{}-invariant bilinear maps}\label{Invgbil}

It may be difficult to derive $\mathfrak{g}$-invariant maps  when $\mathfrak{g}$ is not a low-dimensional Lie algebra. Next, a series of observations simplify their calculation. Our results will enable us to easily determine $\mathfrak{g}$-invariant metrics three-dimensional Lie algebras in Section \ref{Classification}.
\begin{proposition}\label{prop:sym_form}
Let $b$ be $\mathfrak{g}$-invariant bilinear and symmetric  on a $\mathfrak{g}$-module $V$, then 
$
b(vx,x)=0
$ and $
b({\rm Im}\, \rho_v, {\rm ker}\, \rho_v) = 0$ for every $v\in \mathfrak{g}$ and $x\in V$. 
Let $\omega$ be a $\mathfrak{g}$-invariant bilinear anti-symmetric map on $\mathfrak{g}$ relative to the $\mathfrak{g}$-module $(\mathfrak{g},{\rm ad})$. Then, $\omega({\rm ad}_v(w),w)=0$ for every $v, w\in \mathfrak{g}$.
\end{proposition}
\begin{proof}
Using the $\mathfrak{g}$-invariance and symmetricity of $b$, we get
$b(vx, x) = -b(x, vx) = -b(vx,x)$ for all $v \in \mathfrak{g}$ and $x\in V$. Therefore, $b(vx, x) = 0$ for every $x\in V$ and $v\in \mathfrak{g}$. Meanwhile, every $v_2\in {\rm Im}\rho_{v_1}$ can be written as $v_2 := \rho_{v_1}(v_3)$ for some $v_3\in \mathfrak{g}$. Assume that $v_4 \in {\rm ker}\, \rho_{v_1}$. As $b$ is $\mathfrak{g}$-invariant and symmetric, $b(v_2, v_4) = b(\rho_{v_1}(v_3), v_4) = -b(v_3, \rho_{v_1}(v_4)) = 0$ and $b({\rm Im}\,\rho_v,\ker \rho_v)=0$.

From the $\mathfrak{g}$-invariance and anti-symmetricity of $\omega$, one gets $\omega(\textrm{ad}_{v_1}(v_2), v_2) = -\omega(\textrm{ad}_{v_2}(v_1), v_2) = \omega(v_1, \textrm{ad}_{v_2}(v_2)) = 0$ for every $v_1,v_2 \in \mathfrak{g}$.
\end{proof}

The following proposition is a generalization of the one above.
\begin{proposition}\label{prop:form_cond}
Let $b:V\otimes V\rightarrow \mathbb{R}$ be a $\mathfrak{g}$-invariant bilinear map relative to the $\mathfrak{g}$-module $V$. If $W$ is a two-dimensional linear subspace of $V$ satisfying that $vW\subset W$ for any $v \in \mathfrak{g}$, then for any linearly independent $f_s, f_t \in V$ one has
$$
{\rm Tr}(\rho_v|_W)b(f_s, f_t)f_t\wedge f_s =b(f_s, f_s)(vf_t)\wedge f_t + b(f_t, f_t) f_s\wedge(vf_s).
$$\end{proposition}
\begin{proof} Since $vW\subset W$, there exists constants $\alpha_1,\alpha_2,\beta_1,\beta_2\in \mathbb{K}$ such that
$vf_s=\alpha_1f_s+\alpha_2f_t$ and $vf_t=\beta_1f_s+\beta_2f_t$.
The $\mathfrak{g}$-invariance of $b$ ensures that
\begin{equation*}
\alpha_1 b(f_s, f_t) + \alpha_2 b(f_t, f_t) = b(vf_s, f_t)
= -b(f_s, vf_t) = -\beta_1 b(f_s, f_s) - \beta_2 b(f_s, f_t).
\end{equation*}
After rearranging, the above expression gives the stated formula.
\end{proof}
\begin{example}\label{ex:heisenberg} 
Consider the three-dimensional Heisenberg Lie algebra $\mathfrak{h}$. Take a basis $\{e_1,e_2,e_3\}$ of $\mathfrak{h}$ as in Table \ref{tabela3w}. Since $\mathfrak{h}$ is nilpotent, its Killing form vanishes \cite[pg. 480]{Ha15}. In the bases $\{e_{12}, e_{13}, e_{23}\}$ and $\{e_{123}\}$ of $\Lambda^2 \mathfrak{h}$ and $\Lambda^3 \mathfrak{h}$, respectively, a $\mathfrak{h}$-invariant $b$ and its extensions to $\Lambda^2 \mathfrak{h}$ and $\Lambda^3 \mathfrak{h}$ given by Propositions \ref{extension} and \ref{prop:sym_form} read
$$
	[b]:=\left(\begin{array}{ccc}
	\alpha_1&\alpha_2&0\\
	\alpha_3&\alpha_4&0\\
    0&0&0
	\end{array}\right), \quad
[b_{\Lambda^2\mathfrak{h}}]:=\left(\begin{array}{ccc}
	\alpha_1 \alpha_4 - \alpha_2 \alpha_3&0&0\\
	0&0&0\\
	0&0&0\\
	\end{array}\right),\quad 
	[b_{\Lambda^3\mathfrak{g}}]:=\left(0\right), \quad \alpha_i \in \mathbb{R}.
	$$
Using again Proposition \ref{prop:sym_form}, we can compute the general $\mathfrak{h}$-invariant bilinear map $\widetilde{b}$ on $\Lambda^2\mathfrak{h}$
$$
[\widetilde{b}]:=\left(\begin{array}{ccc}
	\beta_3&\beta_2&\beta_1\\
	\beta_4&0&0\\
	\beta_5&0&0\\
	\end{array}\right), \quad \forall \beta_i \in \mathbb{R}.
$$
In the symmetric case, Proposition \ref{prop:sym_form} yields $\beta_1 = \beta_5 = 0$ and $\beta_2 = \beta_4 = 0$. In the anti-symmetric case, one has the $\mathfrak{h}$-invariant maps given by $\beta_3=\beta_1=\beta_5=0$ and $\beta_2=-\beta_4$. Previous bilinear forms are $\mathfrak{h}$-invariant. Contrary to symmetric forms, antisymetric ones are not exploited to classify Lie bialgebras (cf. \cite{Farinati,Gomez}).
\demo
\end{example}
\begin{example}
Let us consider the Lie algebra $\mathfrak{r}_{3,1} := \langle e_1, e_2, e_3 \rangle$ with commutation relations given in Table \ref{tabela3w}. Proposition \ref{prop:sym_form} yields necessary conditions for a bilinear form $\omega_{\Lambda^2\mathfrak{r}_{3,1}}$ on $\Lambda^2\mathfrak{r}_{3,1}$ to be $\mathfrak{r}_{3,1}$-invariant. In particular,
\vskip-0.5cm$$
[\omega_{\Lambda^2\mathfrak{r}_{3,1}}]:=\left(\begin{array}{ccc}
	0&a&0\\
	b&0&0\\
	0&0&0\\
	\end{array}\right), \quad \forall a,b \in \mathbb{R},
$$
in the basis $\{e_{12}, e_{13}, e_{23}\}$ of $\Lambda^2 \mathfrak{r}_{3,1}$. From Proposition \ref{prop:form_cond} and assuming $v:=e_1$, one gets that $a=b=0$ and there are no $\mathfrak{r}_{3,1}$-invariant forms on $\Lambda^2\mathfrak{r}_{3,1}$.
\demo
\end{example}

\section{Graded Lie algebras, their Grassmann algebras, and CYBEs}\label{sect:gradation}

Let us show that a particular type of graded Lie algebra $\mathfrak{g}$ induces a decomposition in $\Lambda \mathfrak{g}$ compatible with its algebraic Schouten bracket in a specific manner to be detailed next. This will be employed to study the structure and solutions to the mCYBE for a very general class of Lie algebras. 

\begin{definition} We say that $\mathfrak{g}$ admits a $G$-{\it gradation} if
$
\mathfrak{g}= \bigoplus_{\alpha \in G\subset \mathbb{R}^n} \mathfrak{g}^{(\alpha)},$
where $(G\subset \mathbb{R}^n,\star)$ is a commutative group, the $\mathfrak{g}^{(\alpha)}$ are subspaces  of $\mathfrak{g}$, and $[\mathfrak{g}^{(\alpha)},\mathfrak{g}^{(\beta)}]\subset \mathfrak{g}^{(\alpha\star\beta)}$ for all $\alpha, \beta \in G$. The spaces $\mathfrak{g}^{(\alpha)}$, for $\alpha\in G$, are called the {\it homogeneous spaces} of the gradation and $\alpha$ is the degree of the space.

\end{definition}

Although every $\mathfrak{g}$ admits a $\mathbb{Z}$-gradation with $\mathfrak{g}^{(0)}=\mathfrak{g}$ and  $\mathfrak{g}^{(\alpha)}=\langle 0\rangle$ with $\alpha\neq 0$, this gradation will not be useful to our purposes, as follows from  posterior considerations. If $\mathfrak{g}$ admits a $G$-gradation and $G$ is known from context or of minor importance, we will simply say that $\mathfrak{g}$ admits a gradation.  

\begin{example} \label{Gradsu2} It stems from  the commutation relations in Table \ref{tabela3w} that $\mathfrak{su}_2$ is isomorphic to the Lie algebra on $\mathbb{R}^3$ with the Lie bracket given by the vector product $\times$. Consider a basis $\{v, v_1^{\perp}, v_2^{\perp}\}$ of $\mathbb{R}^3$, where $v$ is a unit vector, $v_1^\perp$ is perpendicular to $v$, and $v_2^\perp :=v\times v_1^\perp$. Then, $\mathfrak{su}_2$ admits a $\mathbb{Z}_2$-gradation $\mathfrak{g}^{(0)}=\langle v\rangle$ and $\mathfrak{g}^{(1)}=\langle v_1^\perp,v_2^\perp\rangle$. Since $v$ is far from being established canonically, $\mathfrak{su}(2)$ admits several $\mathbb{Z}_2$-gradations.
\end{example}

Note that a $G$-gradation is a particular type of graded Lie algebra and  $\star$ need not be the addition in $\mathbb{R}^n$. Although some of our results can be extended for $G$ being a semigroup (under eventual mild assumptions), this will not be necessary in examples of this work and this possibility will be  skipped.

Gradations can be understood as a generalisation of root decompositions which can be applied to general Lie algebras. Indeed, if $\mathfrak{g}$ admits a root decomposition, it admits a  $\mathbb{Z}^k$-gradation relative to the group structure $(\mathbb{Z}^k,+)$. Moreover, Table \ref{tabela3w} shows that nontrivial gradations can be found for all three-dimensional Lie algebras. Figures \ref{so22_diags} and \ref{so32_diags} give $\mathbb{Z}^2$-gradations for the special pseudo-orthogonal Lie algebras $\mathfrak{so}(3,2)$ and $\mathfrak{so}(2,2)$ (see \cite{CO13,HWZ92}).  

A particular type of gradation whose properties are very close to standard root decompositions, but applicable to general Lie algebras, is given in the next definition. This will be used to easily obtain $\mathfrak{g}$-invariant subspaces in $\Lambda\mathfrak{g}$ for Lie algebras.

\begin{definition}
	We say that $\mathfrak{g}$ admits a {\it root $\mathbb{Z}^k$-gradation} if  it admits a $\mathbb{Z}^k$-gradation such that $\dim \mathfrak{g}^{(0)}=k$ and there exists an injective group morphism $T:\alpha\in \mathbb{Z}^k \mapsto \talpha\in \mathfrak{g}^{(0)*}$ such that $[e,e^{(\alpha)}]=\talpha(e)e^{(\alpha)}$ for every $e\in \mathfrak{g}^{(0)}$ and $e^{(\alpha)}\in \mathfrak{g}^{(\alpha)}$.
\end{definition}

It is immediate that $\mathfrak{g}^{(0)}$ is an abelian Lie algebra and every root decomposition gives rise to a root $\mathbb{Z}^k$-gradation. For instance, Figure \ref{ex:decomp_so22} shows a root decomposition for $\mathfrak{so}(2,2)$ that gives rise to a root $\mathbb{Z}^2$-gradation for $T:(i,j)\in \mathbb{Z}^2\mapsto ie^0+jh^0\in (\mathfrak{so}(2,2)^{(0)})^*$, where $\{e^0,f^0\}$
form a dual basis to the basis $\{e_0,f_0\}$ of $\mathfrak{so}(2,2)^{(0)}$.

The following theorem shows that a $G$-gradation on $\mathfrak{g}$ gives rise to a decomposition of each $\Lambda^m\mathfrak{g}$ into the so-called {\it homogeneous spaces} in such a way that the Schouten bracket maps homogeneous spaces onto homogeneous spaces in a manner determined by the group structure in $G$. 

\begin{theorem}\label{thm:root}
If $\mathfrak{g}$ admits a $G$-gradation  $
\mathfrak{g}=\bigoplus_{\alpha \in G\subset \mathbb{R}^n} \mathfrak{g}^{(\alpha)},
$ then each $\Lambda^m\mathfrak{g}$ admits a decomposition into so-called {\it homogeneous spaces} of the form
\begin{equation}\label{dec}
\Lambda^m\mathfrak{g}=\bigoplus_{\alpha\in G}(\Lambda^m\mathfrak{g})^{(\alpha)},\qquad (\Lambda^m\mathfrak{g})^{(\alpha)}:=\!\!\!\!\!\!\!\!\bigoplus_{\stackrel{\alpha_1,\ldots,\alpha_m\in G}{\alpha_1\star\ldots\star\alpha_m=\alpha}}\!\!\!\!\!\!\!\!\mathfrak{g}^{(\alpha_1)}\wedge\ldots\wedge \mathfrak{g}^{(\alpha_m)},
\end{equation}
Moreover, $[(\Lambda^p\mathfrak{g})^{(\alpha)},(\Lambda^q\mathfrak{g})^{(\beta)}]_{S} \subset \Lambda^{p+q-1}\mathfrak{g}^{(\alpha\star\beta)}$, for all $p,q\in \mathbb{Z}$, and $\alpha,\beta\in G.$
\end{theorem}
\begin{proof} 
The $m$ exterior products among the elements of a basis  $\mathfrak{g}$ adapted to its $G$-gradation give rise to a basis of $\Lambda^m\mathfrak{g}$. Since $G$ is commutative, the exterior product of every exterior product $e^{(\alpha_1)}\wedge\ldots\wedge e^{(\alpha_m)}$, where $e^{(\alpha_i)}\in \mathfrak{g}^{(\alpha_i)}$ for every $i=1,\ldots,m$, belongs to $(\Lambda^m\mathfrak{g})^{(\alpha)}$ for $\alpha=\alpha_1\star\ldots\star\alpha_m$. Elements $e^{(\alpha_1)}\wedge\ldots\wedge e^{(\alpha_m)}$ with $\alpha=\alpha_1\star\ldots\star\alpha_m$ span a basis of $(\Lambda^m\mathfrak{g})^{(\alpha)}$. Repeating the previous process for every $\alpha\in G$, we obtain a basis of $\Lambda^m\mathfrak{g}$ and the decomposition in (\ref{dec}). 

Let us analyse the behaviour of the homogeneous spaces $(\Lambda^m\mathfrak{g})^{(\alpha)}$, with $\alpha\in G$ and $m\in \mathbb{Z}$, relative to the algebraic Schouten bracket. The algebraic Schouten bracket of elements of $(\Lambda^p\mathfrak{g})^{(\alpha)}$ and $(\Lambda^q\mathfrak{g})^{(\beta)}$ follows from (\ref{multi_sn}). Due to the gradation in $\mathfrak{g}$, each term of the sum (\ref{multi_sn}) belongs to $(\Lambda^{p+q-1}\mathfrak{g})^{(\alpha\star \beta)}$. 
\end{proof}

\begin{example}\label{ex:decomp_so22} 
Consider the basis $\{e_{-}, e_0, e_{+}, f_{-}, f_0, f_{+}\}$ of $\mathfrak{so}(2,2)$, where $\{e_{-}, e_0, e_{+}\}$ and $\{f_{-}, f_0, f_{+}\}$ are bases of each copy of $\mathfrak{sl}_2$ within $\mathfrak{so}(2,2)$. Then, $\mathfrak{so}(2,2)$ admits a root $\mathbb{Z}^2$-gradation, given in the first diagram of Figure \ref{so22_diags}. This gives rise to $\mathbb{Z}^2$-gradations on $\Lambda^2\mathfrak{so}(2,2)$ and $\Lambda^3(\mathfrak{so}(2,2))$, whose non-zero homogeneous spaces are indicated by blue points in  Figure \ref{so22_diags}.  Hence, $\mathfrak{so}(2,2)^{(0)}=\langle f_0,e_0\rangle$. We write $\{e^0,f^0\}$ for the dual basis of $\{e_0,f_0\}$.  
\begin{figure}[h!]
\begin{center}
\begin{minipage}{\textwidth}
\begin{center}
\begin{tikzpicture}[scale=0.7]
{\tiny 
\draw [help lines,dotted] (-2,-2) grid (2,2);
\draw [-] (-2,0)--(2,0) node[right] {$e^0$};
\draw [-] (0,-2)--(0,2) node[above] {$f^0$};
\color{blue}
\filldraw (1,0) circle (2pt) node[above] {$e_{+}$};
\filldraw (0,1) circle (2pt) node[left] {$f_{+}$};
\filldraw (-1,0) circle (2pt) node[below] {$e_{-}$};
\filldraw (0,-1) circle (2pt) node[right] {$f_{-}$};
\filldraw (0,0) circle (2pt) node[below right] {$e_0$} node[above right] {$f_0$};\color{black}}
\end{tikzpicture}$\quad\quad$
\begin{tikzpicture}[scale=0.7]
{\tiny
\draw [help lines,dotted] (-2,-2) grid (2,2);
\draw [-] (-2,0)--(2,0) node[right] {$e^0$};
\draw [-] (0,-2)--(0,2) node[above] {$f^0$};
\color{blue}
\filldraw (1,0) circle (2pt) node[below] {\color{black}$1$\color{blue}};
\filldraw (1,1) circle (2pt);
\filldraw (0,1) circle (2pt) node[above right] {\color{black}$1$\color{blue}};
\filldraw (-1,0) circle (2pt) node[below] {\color{black}$-1$\color{blue}};
\filldraw (-1,-1) circle (2pt);
\filldraw (0,-1) circle (2pt) node[below] {\color{black}$-1$\color{blue}};
\filldraw (1,-1) circle (2pt);
\filldraw (-1,1) circle (2pt);
\filldraw (0,0) circle (2pt) node[below right] {\color{black}$0$\color{blue}};
\draw[red, dashed] (-1,-1)--(-1,1);
\draw[red, dashed] (-1,1)--(1,1);
\draw[red, dashed] (1,1)--(1,-1);
\draw[red, dashed] (1,-1)--(-1,-1);
\color{black}}
\end{tikzpicture}$\quad\quad$
\begin{tikzpicture}[scale=0.7]
{\tiny 
\draw [help lines,dotted] (-2,-2) grid (2,2);
\draw [-] (-2,0)--(2,0) node[right] {$e^0$};
\draw [-] (0,-2)--(0,2) node[above] {$f^0$};
\color{blue}
\filldraw (1,0) circle (2pt) node[below] {\color{black}$1$\color{blue}};
\filldraw (1,1) circle (2pt);
\filldraw (0,1) circle (2pt) node[above right] {\color{black}$1$\color{blue}};
\filldraw (-1,0) circle (2pt) node[below] {\color{black}$-1$\color{blue}};
\filldraw (-1,-1) circle (2pt);
\filldraw (0,-1) circle (2pt) node[below] {\color{black}$-1$\color{blue}};
\filldraw (1,-1) circle (2pt);
\filldraw (-1,1) circle (2pt);
\filldraw (0,0) circle (2pt) node[below right] {\color{black}$0$\color{blue}};

\color{black}}
\end{tikzpicture}
\end{center}
\caption{Root decomposition of $\mathfrak{so}(2,2)$ (left), and the induced decompositions on $\Lambda^2 \mathfrak{so}(2,2)$ (center), $\Lambda^3 \mathfrak{so}(2,2)$ (right) given by Theorem \ref{thm:root}. Non-zero homogeneous spaces are indicated by blue points. A basis for each homogeneous subspace of $\mathfrak{so}(2,2)$ is detailed in the first diagram. Limit homogeneous spaces are represented by blue points over a red dashed line.}\label{so22_diags}
\end{minipage}
\end{center}
\end{figure}
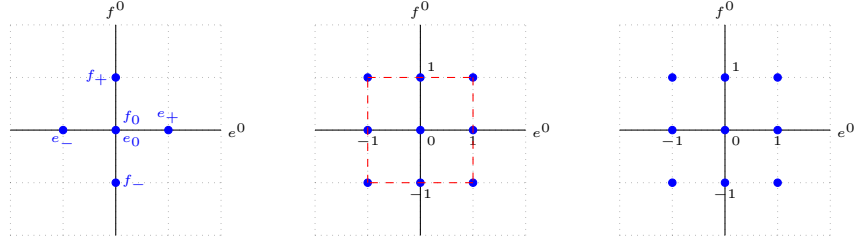

Homogeneous spaces are denoted by their degrees $(i,j)\in \mathbb{Z}^2$. 
The bases for the decompositions (\ref{dec}) of $\Lambda^2 \mathfrak{so}(2,2)$ and $\Lambda^3 \mathfrak{so}(2,2)$ are given in Table \ref{so22_bases}.

\begin{table}[h!]
\centering
\begin{tabular}{|c|c|c|c||c|c|c|}
\hline
$j\backslash i$ & -1 & 0 & 1 & -1 & 0 & 1 \\ \hline
-1 & $e_{-} \wedge f_{-}$ & \parbox[c]{1.3cm}{\vspace{3pt}$e_0 \wedge f_{-},\\ f_{-} \wedge f_0$\vspace{3pt}} & $e_{+} \wedge f_{-}$  & \parbox[c]{2cm}{$e_{-} \!\wedge\! e_0  \!\wedge\!f_{-}\\ e_{-} \!\wedge\! f_{-} \!\wedge\! f_0$} & \parbox[c]{2cm}{\vspace{3pt}$e_0 \!\wedge\! f_{-}  \!\!\wedge\!\! f_0,\\ e_{-} \!\wedge\! e_{+}  \!\!\wedge\!\! f_{-}$\vspace{3pt}} & \parbox[c]{2cm}{$e_0 \!\wedge\! e_{+} \!\wedge\! f_{-}\\ e_{+} \!\wedge\! f_0 \!\wedge\! f_{-}$} \\\hline
0 & \parbox[c]{1.3cm}{$e_{-} \wedge e_0,\\ e_{-} \wedge f_0$} & \parbox[c]{1.3cm}{\vspace{3pt}$ e_0 \wedge f_0,\\ e_{-} \wedge e_{+},\\ f_{-} \wedge f_{+}$\vspace{3pt}} & \parbox[c]{1.3cm}{$e_0 \wedge e_{+},\\ e_{+} \wedge f_0$} & \parbox[c]{2.0cm}{$e_{-} \!\wedge\! e_0 \!\wedge\! f_0,\\ e_{-} \!\wedge\! f_{-} \!\wedge\! f_{+} $} & \parbox[c]{2cm}{\vspace{3pt}$ e_{-} \!\wedge\! e_0 \!\wedge\! e_{+},\\ e_{-} \!\wedge\! e_{+} \!\wedge\! f_0,\\ e_0 \!\wedge\! f_{-} \!\wedge\! f_{+},\\ f_{-} \!\wedge\! f_0 \!\wedge\! f_{+}$\vspace{3pt}} & \parbox[c]{2cm}{$e_0 \!\wedge\! e_{+} \!\wedge\! f_0,\\ e_{+} \!\wedge\! f_{-} \!\wedge\! f_{+}$} \\\hline
1 & $e_{-} \wedge f_{+}$ & \parbox[c]{1.3cm}{$e_0 \wedge f_{+},\\ f_0 \wedge f_{+}$} & $e_{+} \wedge f_{+}$  & \parbox[c]{2cm}{$e_{-} \!\wedge\! e_0 \!\wedge\! f_{+},\\ e_{-} \!\wedge\! f_0 \!\wedge\! f_{+}$} & \parbox[c]{2cm}{\vspace{3pt}$e_0 \!\wedge\! f_0 \!\wedge\! f_{+},\\ e_{-} \!\wedge\! e_{+} \!\wedge\! f_{+}$\vspace{3pt}} & \parbox[c]{2cm}{$e_0 \!\wedge\! e_{+} \!\wedge\! f_+,\\ e_{+} \!\wedge\! f_0 \!\wedge\! f_{+}$} \\[1em]
\hline
\end{tabular}
\caption{Bases for the subspaces $\Lambda^2 \mathfrak{so}(2,2)^{(i,j)}$ (left side) and $\Lambda^3 \mathfrak{so}(2,2)^{(i,j)}$ (right side)}\label{so22_bases}
\end{table}
\end{example}
The following lemma and proposition show that a gradation in $\Lambda \mathfrak{g}$ gives rise to a similar decomposition on $\Lambda_R \mathfrak{g}$.
\begin{lemma}\label{LemgDe} If $\mathfrak{g}$ is a $G$-graded Lie algebra, then $(\Lambda^m\mathfrak{g})^\mathfrak{g}=\bigoplus_{\alpha\in G}(\Lambda^m\mathfrak{g})^{(\alpha)}\cap (\Lambda^m\mathfrak{g})^\mathfrak{g}$
\end{lemma}
\begin{proof}  Due to the gradation of $\mathfrak{g}$, one has that every $w\in \Lambda^m\mathfrak{g}$ can be written in a unique way as $w=\sum_{\alpha\in G}w^{(\alpha)}$ for some $w^{(\alpha)}\in (\Lambda^m\mathfrak{g})^{(\alpha)}$. Since $w$ is $\mathfrak{g}$-invariant, 
$0=\sum_{\alpha\in G}[v^{(\beta)},w^{(\alpha)}]$ for every $v^{(\beta)}\in\mathfrak{g}^{(\beta)}$.
Since $G$ is a group, the elements $\beta+\alpha$, for a fixed $\beta$ and different values of $\alpha$, are different and the $[v^{(\beta)},w^{(\alpha)}]$ belong to different subspaces of the decomposition of $\Lambda^m\mathfrak{g}$. Hence, they must vanish separately and $w^{(\alpha)}\in ((\Lambda^m\mathfrak{g})^{(\alpha)})^{\mathfrak{g}}$. Hence, 
$
(\Lambda^m\mathfrak{g})^\mathfrak{g}\subset \bigoplus_{\alpha\in G}(\Lambda^m\mathfrak{g})^{(\alpha)}\cap (\Lambda^m\mathfrak{g})^{\mathfrak{g}}$. 
 The converse inclusion is immediate.
\end{proof}

We get the following trivial but useful result.

\begin{proposition} 
If $\mathfrak{g}$ admits a $G$-gradation such that $(\Lambda^m\mathfrak{g})^\mathfrak{g}=\bigoplus_{\alpha\in G}(\Lambda \mathfrak{g})^{\mathfrak{g}}\cap \Lambda^m\mathfrak{g}^{(\alpha)}$, then $\Lambda_R\mathfrak{g}$ admits a $G$-gradation such that
$
\Lambda^{m}_R\mathfrak{g}=\bigoplus_{\alpha\in G}(\Lambda^{m}_R\mathfrak{g})^{(\alpha)}$ and $ (\Lambda^m_R\mathfrak{g})^{(\alpha)}= (\Lambda^m\mathfrak{g})^{(\alpha)}/((\Lambda^m\mathfrak{g})^\mathfrak{g}\cap (\Lambda^m\mathfrak{g})^{(\alpha)}).
$
\end{proposition}

Let us now describe a few hints on how gradations and their induced decompositions on Grassmann algebras allow us to obtain certain solutions to CYBEs.

\begin{definition}\label{limitspace}
A {\it limit} homogeneous space of a graded decomposition of $\Lambda^2\mathfrak{g}$ is a homogeneous subspace $\Lambda^2 \mathfrak{g}^{(\alpha)} \subset \Lambda^2\mathfrak{g}$ such that $2\alpha$ is a zero homogeneous space of $\Lambda^3\mathfrak{g}$.
\end{definition}

As a consequence of Definition \ref{limitspace}, elements of limit homogeneous spaces are solutions of the CYBE.
Moreover, if $r_1\in (\Lambda^2\mathfrak{g})^{(\alpha)}$ and $r_2\in (\Lambda^2\mathfrak{g})^{(\beta)}$  are limit $r$-matrices and $(\Lambda^3\mathfrak{g})^{(\alpha+\beta)}=\{0\}$, then any linear combination of elements $r_1,r_2$ is an $r$-matrix. 

\begin{example} Consider the diagram of $\Lambda^2\mathfrak{so}(2,2)$ in Figure \ref{so22_diags}, showing its limit homogeneous subspaces. By direct computation, one can verify that all the elements of these subspaces are solutions of the CYBE.
\demo
\end{example}

A short calculation shows that $f_0\wedge e_0$ is an $r$-matrix for $\mathfrak{so}(2,2)$. This is a particular case of the following more general result.

\begin{proposition} If $\mathfrak{g}$ admits a root gradation, then $\Lambda^2(\mathfrak{g}^{(0)})$ is a subspace of solutions of the CYBE. 
\end{proposition}

When $\dim \mathfrak{g}=5$ or larger, the spaces $\Lambda^2\mathfrak{g},\Lambda^3\mathfrak{g}$ are so large that it is difficult to determine their $\mathfrak{g}$-invariant elements, homogeneous subspaces, and other of their properties related to mCYBEs. To help in analysing these topics, the use of ${\rm Aut}(\mathfrak{g})$ can be useful. If ${\rm Aut}(\mathfrak{g})$ maps homogeneous spaces into homogeneous spaces of a gradation of $\mathfrak{g}$, e.g. when ${\rm Aut}(\mathfrak{g})$ preserves the Cartan subalgebra of a root decomposition of $\mathfrak{g}$, then ${\rm Aut}(\mathfrak{g})$ is useful to obtain decompositions in $\Lambda\mathfrak{g}$ induced by gradations in $\mathfrak{g}$ when $\dim \mathfrak{g}$ is large. To illustrate our claims, we will now study $\mathfrak{so}(3,2)$.

Let $\{J_{\pm},J_3,K_{\pm},K_3,S_{\pm},R_{\pm}\}$ be a basis of $\mathfrak{so}(3,2)$ satisfying the commutation relations \cite{CO13}
\begin{align*}
&[J_{\pm},K_{\pm}]=\pm R_{\pm},&[J_{\mp},K_{\pm}]&=\pm S_{\pm},&[J_3,K_{\pm}]&=0,&[J_3,K_3]&=0,\\
&[J_{\pm},R_{\pm}]=0,&[J_{\mp},R_{\pm}]&=\pm 2K_{\pm},&[J_3,R_{\pm}]&=\pm R_{\pm},&[R_{\pm},S_{\pm}]&=0,\\
&[J_{\pm},S_{\pm}]=\pm 2K_{\pm},&[J_{\mp},S_{\pm}]&=0,&[J_3,S_{\pm}]&=\mp S_{\pm},&[R_{\mp},S_{\pm}]&=0,\\
&[K_{\pm},R_{\pm}]=0,&[K_{\mp},R_{\pm}]&=\pm 2J_{\pm},& [K_3,R_{\pm}]&=\pm R_{\pm},&[S_+,S_-]&=-4(K_3-J_3),\\
&[K_{\pm},S_{\pm}]=0,&[K_{\mp},S_{\pm}]&=\pm 2J_{\mp},&[K_3,S_{\pm}]&=\pm S_{\pm},&[R_+,R_-]&=-4(K_3+J_3),\\
&&[K_-,K_+]&=2K_3,&[J_-,J_+]&=-2J_3.
\end{align*}

The diagrams for the induced homogeneous spaces of $\Lambda \mathfrak{so}(3,2)$ appear in Figure \ref{so32_diags}. Consider  $T_1,T_2\in {\rm Aut}(\mathfrak{so}(3,2))$ that act on the diagram for $\mathfrak{so}(3)$ as reflections on the $OY$ and $OX$ axis in such a way that
$
T_1(J_3)=-J_3, T_1(K_3)=K_3, T_2(J_3)=J_3, T_2(K_3)=-K_3,
$
and they act as a permutation in the rest of elements of the chosen basis of $\mathfrak{so}(3,2)$. Evidently, $T_1,T_2$ do not preserve the subspaces of the root $\mathbb{Z}^2$-gradation, but they map homogeneous subspaces of $\mathfrak{so}(3,2)$ onto homogeneous subspaces. Then, they map homogeneous spaces of the induced decomposition in $\Lambda\mathfrak{so}(3,2)$ into homogeneous spaces of the decomposition. This allows us to obtain a basis adapted to the decomposition of any $\Lambda^m\mathfrak{so}(3,2)$ from the bases of homogeneous spaces with $(i\geq 0,j\geq 0)$ (see Table \ref{so32_bases_3}). This basis of $\Lambda^2\mathfrak{so}(3,2)$ can be used to simplify mCYBEs, for instance, by searching for $r$-matrices belonging to certain subfamilies of homogeneous subspaces.
\noindent
\begin{figure}[h!]
\centering
\begin{minipage}{0.3\textwidth}
\begin{center}
\begin{tikzpicture}[scale=0.5]
{\tiny 
\draw [help lines,dotted] (-1,-1) grid (1,1);
\draw [-] (-2,0)--(2,0) node[right] {$K_3$};
\draw [-] (0,-2)--(0,2) node[above] {$J_3$};
\color{blue}
\filldraw (1,0) circle (1pt) node[right] {$K_{+}$};
\filldraw (1,1) circle (1pt) node[above] {$R_{+}$};
\filldraw (0,1) circle (1pt) node[above] {$J_{+}$};
\filldraw (-1,0) circle (1pt) node[below] {$K_{-}$};
\filldraw (-1,-1) circle (1pt) node[below] {$R_{-}$};
\filldraw (0,-1) circle (1pt) node[below] {$J_{-}$};
\filldraw (-1,1) circle (1pt) node[above] {$S_{-}$};
\filldraw (1,-1) circle (1pt) node[below] {$S_{+}$};
\filldraw (0,0) circle (2pt) node[below right] {$K_3$} node[above right] {$J_3$};\color{black}}
\end{tikzpicture}
\end{center}
\end{minipage}
\begin{minipage}{0.3\textwidth}
\begin{center}
\begin{tikzpicture}[scale=0.5]
{\tiny
\draw [help lines,dotted] (-2,-2) grid (2,2);
\draw [-] (-3,0)--(3,0) node[right] {$K_3$};
\draw [-] (0,-3)--(0,3) node[above] {$J_3$};
\color{blue}
\draw[red, dashed] (-1,2)--(1,2);
\draw[red, dashed] (1,2)--(1,1);
\draw[red, dashed] (1,1)--(2,1);
\draw[red, dashed] (2,1)--(2,-1);
\draw[red, dashed] (2,-1)--(1,-1);
\draw[red, dashed] (1,-1)--(1,-2);
\draw[red, dashed] (1,-2)--(-1,-2);
\draw[red, dashed] (-1,-2)--(-1,-1);
\draw[red, dashed] (-1,-1)--(-2,-1);
\draw[red, dashed] (-2,-1)--(-2,1);
\draw[red, dashed] (-2,1)--(-1,1);
\draw[red, dashed] (-1,1)--(-1,2);
\filldraw (1,0) circle (2pt) node[below] {$1$};
\filldraw (2,0) circle (2pt) node[below] {$2$};
\filldraw (2,1) circle (2pt);
\filldraw (1,1) circle (2pt);
\filldraw (1,2) circle (2pt);
\filldraw (0,1) circle (2pt) node[right] {$1$};
\filldraw (0,2) circle (2pt) node[right] {$2$};
\filldraw (-1,1) circle (2pt);
\filldraw (-1,2) circle (2pt);
\filldraw (-1,0) circle (2pt) node[below] {$-1$};
\filldraw (-2,0) circle (2pt) node[below] {$-2$};
\filldraw (-1,-2) circle (2pt);
\filldraw (-1,-1) circle (2pt);
\filldraw (-2,-1) circle (2pt);
\filldraw (0,-1) circle (2pt) node[right] {$-1$};
\filldraw (0,-2) circle (2pt) node[right] {$-2$};
\filldraw (-2,1) circle (2pt);
\filldraw (1,0-2) circle (2pt);
\filldraw (2,-1) circle (2pt);
\filldraw (1,-1) circle (2pt);
\filldraw (0,0) circle (2pt);\color{black}}
\end{tikzpicture}
\end{center}
\end{minipage}
\begin{minipage}{0.3\textwidth}
\begin{center}
\begin{tikzpicture}[scale=0.5]
{\tiny 
\draw [help lines,dotted] (-3,-3) grid (3,3);
\draw [-] (-4,0)--(4,0) node[right] {$K_3$};
\draw [-] (0,-4)--(0,4) node[above] {$J_3$};
\color{blue}
\filldraw (1,0) circle (2pt) node[below] {$1$};
\filldraw (2,0) circle (2pt) node[below] {$2$};
\filldraw (3,0) circle (2pt) node[below] {$3$};
\filldraw (1,1) circle (2pt);
\filldraw (2,2) circle (2pt);
\filldraw (2,1) circle (2pt);
\filldraw (1,2) circle (2pt);
\filldraw (0,1) circle (2pt) node[right] {$1$};
\filldraw (0,2) circle (2pt) node[right] {$2$};
\filldraw (0,3) circle (2pt) node[right] {$3$};
\filldraw (-1,1) circle (2pt);
\filldraw (-1,2) circle (2pt);
\filldraw (-2,1) circle (2pt);
\filldraw (-2,2) circle (2pt);
\filldraw (-1,0) circle (2pt) node[below] {$-1$};
\filldraw (-2,0) circle (2pt) node[below] {$-2$};
\filldraw (-3,0) circle (2pt) node[below] {$-3$};
\filldraw (-1,-1) circle (2pt);
\filldraw (-2,-1) circle (2pt);
\filldraw (-1,-2) circle (2pt);
\filldraw (-2,-2) circle (2pt);
\filldraw (0,-1) circle (2pt) node[right] {$-1$};
\filldraw (0,-2) circle (2pt) node[right] {$-2$};
\filldraw (0,-3) circle (2pt) node[right] {$-3$};
\filldraw (1,-1) circle (2pt);
\filldraw (2,-1) circle (2pt);
\filldraw (1,-2) circle (2pt);
\filldraw (2,-2) circle (2pt);
\filldraw (0,0) circle (2pt);\color{black}}
\end{tikzpicture}
\end{center}
\end{minipage}
\caption{Graded decompositions of $\mathfrak{so}(3,2)$ (left), $\Lambda^2\mathfrak{so}(3,2)$ (center), $\Lambda^3\mathfrak{so}(3,2)$ (right) and limit homogeneous spaces (subspaces over the red dashed line).}\label{so32_diags}
\end{figure}
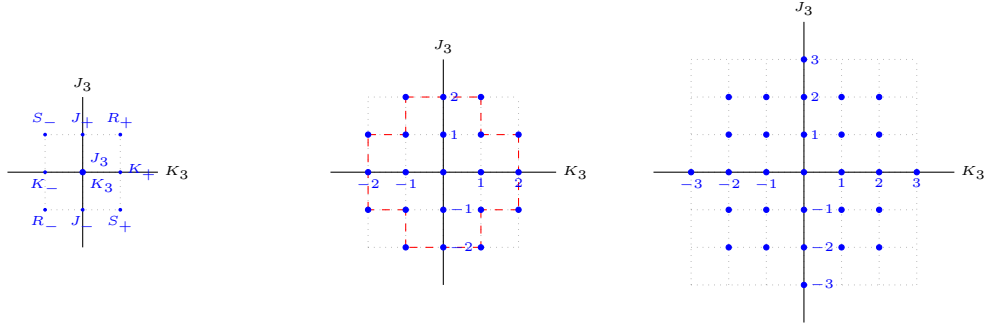
\begin{table}[!h]
\centering
\begin{tabular}{|c|c|c|c|c|c|}
\hline
$j\backslash i$ & -2 & -1 & 0 & 1 & 2 \\
\hline 
-2 &
&
$J_{-} \wg R_{-}$&
$R_{-} \wg S_{+}$&
$J_{-} \wg S_{+}$&
\\\hline
-1 &
$K_{-} \wg R_{-}$& 
\parbox[c]{2.9cm}{\vspace{3pt}$J_3 \wg R_{-}$, $K_3 \wg R_{-}$,\\ $J_{-} \wg K_{-}$\vspace{3pt}} &
\parbox[c]{2.9cm}{\vspace{3pt}$J_3 \wg J_{-}$, $K_3 \wg J_{-}$,\\ $K_{+} \wg R_{-}$, $K_{-} \wg S_{+}$\vspace{3pt}} &
\parbox[c]{2.9cm}{\vspace{3pt}$J_3 \wg S_{+}$, $K_3 \wg S_{+}$,\\ $J_{-} \wg K_{+}$\vspace{3pt}} &
$K_{+} \wg S_{+}$\\\hline
0 &
$R_{-} \wg S_{-}$&
\parbox[c]{2.9cm}{\vspace{3pt}$J_3 \wg K_{-}$, $K_3 \wg K_{-}$,\\ $J_{+} \wg R_{-}$, $J_{-} \wg S_{-}$\vspace{3pt}} &
\parbox[c]{2.9cm}{\vspace{3pt}$J_3 \wg K_3$, $J_{+} \wg J_{-}$,\\ $K_{+} \wg K_{-}$, $R_{+} \wg R_{-}$,\\ $S_{+} \wg S_{-}$\vspace{3pt}} &
\parbox[c]{2.9cm}{\vspace{3pt}$J_3 \wg K_{+}$, $K_3 \wg K_{+}$,\\ $J_{+} \wg S_{+}$, $J_{-} \wg R_{+}$\vspace{3pt}} &
$R_{+} \wg S_{+}$\\\hline
1 &
$K_{-} \wg S_{-}$&
\parbox[c]{2.9cm}{\vspace{3pt}$J_3 \wg S_{-}$, $K_3 \wg S_{-}$,\\ $J_{+} \wg K_{-}$\vspace{3pt}} &
\parbox[c]{2.9cm}{\vspace{3pt}$J_3 \wg J_{+}$, $K_3 \wg J_{+}$,\\ $K_{+} \wg S_{-}$, $K_{-} \wg R_{+}$\vspace{3pt}} &
\parbox[c]{2.9cm}{\vspace{3pt}$J_3 \wg R_{+}$, $K_3 \wg R_{+}$,\\ $J_{+} \wg K_{+}$\vspace{3pt}} &
$K_{+} \wg R_{+}$\\\hline
2 &
&
$J_{+} \wg S_{-}$&
$R_{+} \wg S_{-}$&
$J_{+} \wg R_{+}$&
\\
\hline
\end{tabular}
\caption{Bases for the homogeneous subspaces of $\Lambda^2 \mathfrak{so}(3,2)^{(i,j)}$.}\label{so32_bases_2}
\end{table}

\begin{table}[!hp]
\centering
\small
\rotatebox{90}{
\begin{tabular}{|c|c|c|c|c|c|c|c|}
\hline
$j\backslash i$ & -3 & -2 & -1 & 0 & 1 & 2 & 3 \\
\hline &&&&&&\\[-1em]
-3 &
&
&
&
$J_{-} \wg R_{-} \wg S_{+}$&
&
&
\\\hline
-2 &
&
$J_{-} \wg K_{-} \wg R_{-}$&
\parbox[c]{2.2cm}{$J_3 \wg J_{-} \wg R_{-}$,\\ $K_3 \wg J_{-} \wg R_{-}$,\\ $K_{-} \wg R_{-} \wg S_{+}$}&
\parbox[c]{2.2cm}{\vspace{3pt}$J_3 \wg R_{-} \wg S_{+}$,\\ $K_3 \wg R_{-} \wg S_{+}$,\\ $J_{-} \wg K_{+} \wg R_{-}$,\\ $J_{-} \wg K_{-} \wg S_{+}$\vspace{3pt}}&
\parbox[c]{2.2cm}{$J_3 \wg J_{-} \wg S_{+}$,\\ $K_3 \wg J_{-} \wg S_{+}$,\\ $K_{+} \wg R_{-} \wg S_{+}$}&
$J_{-} \wg K_{+} \wg S_{+}$&
\\\hline
-1 &
&
\parbox[c]{2.2cm}{$J_3 \wg K_{-} \wg R_{-}$,\\ $K_3 \wg K_{-} \wg R_{-}$,\\ $J_{-} \wg R_{-} \wg S_{-}$}&
\parbox[c]{2.2cm}{$J_3 \wg K_3 \wg R_{-}$,\\ $J_3 \wg J_{-} \wg K_{-}$,\\ $K_3 \wg J_{-} \wg K_{-}$,\\ $J_{+} \wg J_{-} \wg R_{-}$,\\ $K_{+} \wg K_{-} \wg R_{-}$,\\ $R_{-} \wg S_{+} \wg S_{-}$}&
\parbox[c]{2.2cm}{\vspace{3pt}$J_3 \wg K_3 \wg J_{-}$,\\ $J_3 \wg K_{+} \wg R_{-}$,\\ $J_3 \wg K_{-} \wg S_{+}$,\\ $K_3 \wg K_{+} \wg R_{-}$,\\ $K_3 \wg K_{-} \wg S_{+}$,\\ $J_{+} \wg R_{-} \wg S_{+}$,\\ $J_{-} \wg K_{+} \wg K_{-}$,\\ $J_{-} \wg R_{+} \wg R_{-}$,\\ $J_{-} \wg S_{+} \wg S_{-}$\vspace{3pt}}&
\parbox[c]{2.2cm}{$J_3 \wg K_3 \wg S_{+}$,\\ $J_3 \wg J_{-} \wg K_{+}$,\\ $K_3 \wg J_{-} \wg K_{+}$,\\ $J_{+} \wg J_{-} \wg S_{+}$,\\ $K_{+} \wg K_{-} \wg S_{+}$,\\ $R_{+} \wg R_{-} \wg S_{+}$}&
\parbox[c]{2.2cm}{$J_3 \wg K_{+} \wg S_{+}$,\\ $K_3 \wg K_{+} \wg S_{+}$,\\ $J_{-} \wg R_{+} \wg S_{+}$}&
\\\hline
0 &
$K_{-} \wg R_{-} \wg S_{-}$&
\parbox[c]{2.2cm}{$J_3 \wg R_{-} \wg S_{-}$,\\ $K_3 \wg R_{-} \wg S_{-}$,\\ $J_{+} \wg K_{-} \wg R_{-}$,\\ $J_{-} \wg K_{-} \wg S_{-}$}&
\parbox[c]{2.2cm}{$J_3 \wg K_3 \wg K_{-}$,\\ $J_3 \wg J_{+} \wg R_{-}$,\\ $J_3 \wg J_{-} \wg S_{-}$,\\ $K_3 \wg J_{+} \wg R_{-}$,\\ $K_3 \wg J_{-} \wg S_{-}$,\\ $J_{+} \wg J_{-} \wg K_{-}$,\\ $K_{+} \wg R_{-} \wg S_{-}$,\\ $K_{-} \wg R_{+} \wg R_{-}$,\\ $K_{-} \wg S_{+} \wg S_{-}$}&
\parbox[c]{2.2cm}{\vspace{3pt}$J_3 \wg J_{+} \wg J_{-}$,\\ $J_3 \wg K_{+} \wg K_{-}$,\\ $J_3 \wg R_{+} \wg R_{-}$,\\ $J_3 \wg S_{+} \wg S_{-}$,\\ $K_3 \wg J_{+} \wg J_{-}$,\\ $K_3 \wg K_{+} \wg K_{-}$,\\ $K_3 \wg R_{+} \wg R_{-}$,\\ $K_3 \wg S_{+} \wg S_{-}$,\\ $J_{+} \wg K_{+} \wg R_{-}$,\\ $J_{+} \wg K_{-} \wg S_{+}$,\\ $J_{-} \wg K_{+} \wg S_{-}$,\\ $J_{-} \wg K_{-} \wg R_{+}$\vspace{3pt}}&
\parbox[c]{2.2cm}{$J_3 \wg K_3 \wg K_{+}$,\\ $J_3 \wg J_{+} \wg S_{+}$,\\ $J_3 \wg J_{-} \wg R_{+}$,\\ $K_3 \wg J_{+} \wg S_{+}$,\\ $K_3 \wg J_{-} \wg R_{+}$,\\ $J_{+} \wg J_{-} \wg K_{+}$,\\ $K_{+} \wg R_{+} \wg R_{-}$,\\ $K_{+} \wg S_{+} \wg S_{-}$,\\ $K_{-} \wg R_{+} \wg S_{+}$}&
\parbox[c]{2.2cm}{$J_3 \wg R_{+} \wg S_{+}$,\\ $K_3 \wg R_{+} \wg S_{+}$,\\ $J_{+} \wg K_{+} \wg S_{+}$,\\ $J_{-} \wg K_{+} \wg R_{+}$}&
$K_{+} \wg R_{+} \wg S_{+}$\\\hline
1 &
&
\parbox[c]{2.2cm}{$J_3 \wg K_{-} \wg S_{-}$,\\ $K_3 \wg K_{-} \wg S_{-}$,\\ $J_{+} \wg R_{-} \wg S_{-}$}&
\parbox[c]{2.2cm}{$J_3 \wg K_3 \wg S_{-}$,\\ $J_3 \wg J_{+} \wg K_{-}$,\\ $K_3 \wg J_{+} \wg K_{-}$,\\ $J_{+} \wg J_{-} \wg S_{-}$,\\ $K_{+} \wg K_{-} \wg S_{-}$,\\ $R_{+} \wg R_{-} \wg S_{-}$}&
\parbox[c]{2.2cm}{\vspace{3pt}$J_3 \wg K_3 \wg J_{+}$,\\ $J_3 \wg K_{+} \wg S_{-}$,\\ $J_3 \wg K_{-} \wg R_{+}$,\\ $K_3 \wg K_{+} \wg S_{-}$,\\ $K_3 \wg K_{-} \wg R_{+}$,\\ $J_{+} \wg K_{+} \wg K_{-}$,\\ $J_{+} \wg R_{+} \wg R_{-}$,\\ $J_{+} \wg S_{+} \wg S_{-}$,\\ $J_{-} \wg R_{+} \wg S_{-}$\vspace{3pt}}&
\parbox[c]{2.2cm}{$J_3 \wg K_3 \wg R_{+}$,\\ $J_3 \wg J_{+} \wg K_{+}$,\\ $K_3 \wg J_{+} \wg K_{+}$,\\ $J_{+} \wg J_{-} \wg R_{+}$,\\ $K_{+} \wg K_{-} \wg R_{+}$,\\ $R_{+} \wg S_{+} \wg S_{-}$}&
\parbox[c]{2.2cm}{$J_3 \wg K_{+} \wg R_{+}$,\\ $K_3 \wg K_{+} \wg R_{+}$,\\ $J_{+} \wg R_{+} \wg S_{+}$}&
\\\hline
2 &
&
$J_{+} \wg K_{-} \wg S_{-}$&
\parbox[c]{2.2cm}{$J_3 \wg J_{+} \wg S_{-}$,\\ $K_3 \wg J_{+} \wg S_{-}$,\\ $K_{-} \wg R_{+} \wg S_{-}$}&
\parbox[c]{2.2cm}{\vspace{3pt}$J_3 \wg R_{+} \wg S_{-}$,\\ $K_3 \wg R_{+} \wg S_{-}$,\\ $J_{+} \wg K_{+} \wg S_{-}$,\\ $J_{+} \wg K_{-} \wg R_{+}$\vspace{3pt}}&
\parbox[c]{2.2cm}{$J_3 \wg J_{+} \wg R_{+}$,\\ $K_3 \wg J_{+} \wg R_{+}$,\\ $K_{+} \wg R_{+} \wg S_{-}$}&
$J_{+} \wg K_{+} \wg R_{+}$&
\\\hline
3 &
&
&
&
$J_{+} \wg R_{+} \wg S_{-}$&
&
&
\\\hline
\end{tabular}
\normalsize}
\caption{Bases for the homogeneous subspaces $\Lambda^3 \mathfrak{so}(3,2)^{(i,j)}$.}\label{so32_bases_3}
\end{table}

\section{Geometry of \texorpdfstring{$\mathfrak{g}$}{}-invariant elements\label{gInvar}}

This section addresses the study and characterisation of the spaces $(\Lambda^m \mathfrak{g})^\mathfrak{g}$ of $\mathfrak{g}$-invariant $m$-vectors for Lie algebras with root gradations and nilpotent Lie algebras. The interest of the spaces $(\Lambda^m\mathfrak{g})^{\mathfrak{g}}$ is due to their occurrence in the analysis of Lie bialgebras, mCYBEs, and ${\rm Aut}(\mathfrak{g})$ \cite{Pressley}.  The first part of this section concerns the study of ($\Lambda \mathfrak{g})^{\mathfrak{g}}$ for  general Lie algebras and, then, Lie algebras with a root gradation. The second part of this section is focused on $(\Lambda \mathfrak{g})^{\mathfrak{g}}$ for nilpotent Lie algebras $\mathfrak{g}$. Although our results do not characterise completely $(\Lambda^m\mathfrak{g})^{\mathfrak{g}}$, they are general enough to obtain many of its elements and, in several cases, to determine the whole $(\Lambda^m\mathfrak{g})^{\mathfrak{g}}$.
 Let us begin with a simple interesting fact.

\begin{proposition}\label{prop:lgg}
The space $(\Lambda \mathfrak{g})^\mathfrak{g}$ is an $\mathbb{R}$-algebra relative to the exterior product. Moreover, each space $(\Lambda^m \mathfrak{g})^\mathfrak{g}$ is $\widehat{{\rm ad}}$-invariant.
\end{proposition}
\begin{proof}
We can write $(\Lambda \mathfrak{g})^\mathfrak{g}=\cap_{v\in \mathfrak{g}}\ker [v,\cdot]_S$. Then, $(\Lambda \mathfrak{g})^{\mathfrak{g}}$ is a linear space and, since $[v,\cdot]_S$ is a derivation relative to the exterior product, the exterior product of elements of $(\Lambda \mathfrak{g})^{\mathfrak{g}}$ belongs to it. Hence, $(\Lambda \mathfrak{g})^{\mathfrak{g}}$ is a subalgebra of $\Lambda \mathfrak{g}$ relative to the exterior product.

By Proposition \ref{prop:glrho_ginv}, the second part of our proposition amounts to the fact that $(\Lambda^m\mathfrak{g})^\mathfrak{g}$ is invariant under relative to $\Lambda^mT$ for every  $T\in {\rm Aut}_c(\mathfrak{g})$. But if $w\in (\Lambda^m\mathfrak{g})^{\mathfrak{g}}$ and $T\in {\rm Aut}(\mathfrak{g})$, then
$$
[v,\Lambda^mTw]_S=[TT^{-1}v,\Lambda^mTw]_S=\Lambda^mT[T^{-1}v,w]_S=0, \quad \forall v \in \mathfrak{g} \,\, \xRightarrow{} \,\, \Lambda^m Tw \in (\Lambda^m\mathfrak{g})^{\mathfrak{g}}.
$$
Hence, $(\Lambda^m\mathfrak{g})^{\mathfrak{g}}$ is invariant relative to the action of ${\rm Aut}_c(\mathfrak{g})$ on it.
\end{proof}

In order to prove our following results, it is appropriate to introduce the next notion.

\begin{definition}A {\it traceless ideal} of $\mathfrak{g}$ is an ideal $\mathfrak{h}\subset \mathfrak{g}$ satisfying that the restriction of each ${\rm ad}_v$, with $v\in\mathfrak{g}$, to $\mathfrak{h}$, say ${\rm ad}_v|_{\mathfrak{h}}$, is traceless. If the elements of ${\rm ad}(\mathfrak{g})$ are traceless, then $\mathfrak{g}$ is called {\it unimodular} .
\end{definition}

Unimodular Lie algebras have been applied to several different mathematical and physical problems \cite{ACK99,Ov06}, which  motivate their study. There exist several conditions ensuring that a Lie algebra is unimodular \cite{ACK99,GP10,Ha15}, e.g. the Lie algebras of abelian, compact, semi-simple, or nilpotent groups are unimodular \cite{BR86,Ha15}. One of our reasons to study unimodular Lie algebras and their traceless ideals is given by the following proposition.

\begin{proposition}\label{prop:inv_trlsub}
Every traceless ideal $\mathfrak{h}$ of a solvable Lie algebra $\mathfrak{g}$ is such that $\Lambda^{\dim\mathfrak{h}}\mathfrak{h}\subset(\Lambda^{\dim\mathfrak{h}}\mathfrak{g})^{\mathfrak{g}}$.
\end{proposition}
\begin{proof} The ideal $\mathfrak{h}$ induces a one-dimensional space $\Lambda^{\dim \mathfrak{h}}\mathfrak{h}$. Let $r\in \Lambda^{\dim \mathfrak{h}}\mathfrak{h}$. Since $\mathfrak{h}$ is a traceless ideal by assumption, $[v,r]_S={\rm Tr} ({\rm ad}_v|_{\mathfrak{h}})r=0$ for every $v\in \mathfrak{g}$. Hence, $r\in (\Lambda^{\dim \mathfrak{h}}\mathfrak{g})^\mathfrak{g}$. 
\end{proof}

Proposition \ref{prop:inv_trlsub} says that aforesaid ideals give rise to decomposable elements of $(\Lambda^{\dim\mathfrak{h}}\mathfrak{g})^\mathfrak{g}$, which in turn can be obtained by a family of equations involving the algebraic Schouten bracket. This approach gives a manner to determine traceless ideals in $\mathfrak{g}$ via $(\Lambda \mathfrak{g})^\mathfrak{g}$. The following theorem gives a method to determine the $\mathfrak{g}$-invariant elements in $\Lambda\mathfrak{g}$.

\begin{theorem}\label{Theo:InvLk} Every decomposable $ \Omega\in (\Lambda\mathfrak{g})^\mathfrak{g} \backslash \{0\}$ amounts to a unique ideal $\mathfrak{h}\subset \mathfrak{g}$ such that $\langle\Omega\rangle=\Lambda^{\dim\mathfrak{h}}\mathfrak{h}$ and ${\rm ad}_v \vert_{\mathfrak{h}}$ for every $v\in \mathfrak{g}$ is traceless. In turn, if $\mathfrak{h} \neq \{0\}$, then $\Lambda^{\dim\mathfrak{h}}\mathfrak{h}\subset (\Lambda \mathfrak{g})^\mathfrak{g}$. 
\end{theorem} 
\begin{proof}
A non-zero decomposable $\Omega \in (\Lambda\mathfrak{g})^{\mathfrak{g}}$ takes the form $\Omega = v_1 \wedge \ldots \wedge v_m$ for some linearly independent $v_1,\ldots,v_m\in \mathfrak{g}$. This defines a unique  $\mathfrak{h} := \langle v_1, \ldots, v_m\rangle\subset\mathfrak{g}$ that is independent of the $v_1,\ldots,v_m$, and $\langle \Omega\rangle= \Lambda^{\dim\mathfrak{h}}\mathfrak{h}$. Since $[v, \Omega]_S = 0$ for every $v \in \mathfrak{g}$, one has that $\mathfrak{h}$ is an ideal of $\mathfrak{g}$. Let us prove this fact. If $\widehat \Omega: \beta \in \mathfrak{g^*} \mapsto \iota_{\beta} \Omega \in \Lambda \mathfrak{g}$, then $\ker \widehat\Omega=\mathfrak{h}^\circ$. Assuming $v\in \mathfrak{g}$, $\theta\in \mathfrak{h}^\circ$, we obtain
\begin{multline*}
\iota_{{\rm ad}_v^*\theta}\Omega
={\rm ad}_v^*\theta(v_1)\wedge\ldots\wedge v_m+\ldots+(-1)^mv_1\wedge\ldots\wedge {\rm ad}_v^*\theta(v_m)\\
=\iota_\theta([v,v_1]\wedge  \ldots \wedge v_m+\ldots + (-1)^mv_1\wedge \ldots\wedge [v,v_m])=\iota_\theta[v,\Omega]_S.
\end{multline*}
Then, ${\rm ad}^*_v\theta\in \mathfrak{h}^\circ$ for every $\theta\in \mathfrak{h}^\circ$. Consequently, ${\rm ad}_v\mathfrak{h}\subset\mathfrak{h}$ for every $v\in \mathfrak{g}$ and $\mathfrak{h}$ is an ideal of $\mathfrak{g}$. Since $[v,\Omega]_S=({\rm Tr} {\rm ad}_v|_\mathfrak{h})\Omega=0$, one gets that ${\rm ad}_v|_{\mathfrak{h}}$ is traceless and, by Proposition \ref{prop:inv_trlsub}, it follows that $\Lambda^{\dim \mathfrak{h}}\mathfrak{h}\subset (\Lambda\mathfrak{g})^\mathfrak{g}$.

Conversely, if $\mathfrak{h}$ is a non-zero ideal, then  $ \Lambda^{\dim \mathfrak{h}}\mathfrak{h}$ is one-dimensional and it admits a basis, $\Omega$, given by the exterior product of the elements of a basis of $\mathfrak{h}$. Since every ${\rm ad}_v$, with $v\in\mathfrak{g}$, acts on $\mathfrak{h}$ tracelessly by assumption, $[v, \Omega]_S = (\textrm{Tr ad}_v|_\mathfrak{h})\Omega = 0$ and  $\Omega \in (\Lambda^{\dim \mathfrak{h}} \mathfrak{g})^{\mathfrak{g}}$.\vskip-0.4cm
\end{proof}

\begin{corollary} There exists a one-to-one correspondence between one-dimensional subspaces of decomposable $\mathfrak{g}$-invariant elements of $\Lambda\mathfrak{g}$ and ideals of $\mathfrak{g}$ where $\mathfrak{g}$ acts traceless.
\end{corollary}
\begin{proof} Let $D_\mathfrak{g}$ be the set of one-dimensional subspaces of decomposable elements of $\Lambda \mathfrak{g}$ and let ${\rm Tr}(\mathfrak{g})$ be the space of non-zero ideals of $\mathfrak{g}$ where $\mathfrak{g}$ acts tracelessly. Define the mapping $\phi: \langle\Omega\rangle \in D_\mathfrak{g} \mapsto \mathfrak{h}_\Omega \in {\rm Tr}(\mathfrak{g})$,
where $\Omega$ is a $\mathfrak{g}$-invariant decomposable element of $\Lambda\mathfrak{g}$ and $\mathfrak{h}_\Omega$ is the unique element of Tr$(\mathfrak{g})$ induced by $\langle \Omega\rangle $ in virtue of
Theorem \ref{Theo:InvLk}. We want to prove that $\phi$  is a well-defined bijection. The map $\phi$ is well defined as it does not depend on the element $\Omega$ spanning the space $\langle \Omega\rangle$. In turn, $\Lambda^{\dim\mathfrak{h}_\Omega}\mathfrak{h}_\Omega=\langle \Omega\rangle$. Hence, $\phi$ has a right inverse. Additionally, an ideal $\mathfrak{h}$ gives rise in view of Theorem \ref{Theo:InvLk} to an element of $\Lambda^{\dim \mathfrak{h}}\mathfrak{h}$ that is $\mathfrak{g}$-invariant. In turn this element is related to $\mathfrak{h}$. Therefore $\phi$ has a left inverse. This gives the searched bijection.
\end{proof}

The aforesaid ideal $\mathfrak{h}$ in Theorem \ref{Theo:InvLk} induced by a non-zero $\mathfrak{g}$-invariant decomposable multivector acts on itself by the adjoint action tracelessly, i.e. it is  {\it unimodular}.

\subsection{Lie algebras with a root gradation}
 
 Let us study the relation between induced decompositions on the spaces $\Lambda^m\mathfrak{g}$ by root gradations on Lie algebras and their $\mathfrak{g}$-invariant metrics.
\begin{proposition}\label{gin0}
If $\mathfrak{g}$ admits a root gradation, then  $(\Lambda^m\mathfrak{g})^\mathfrak{g}\subset (\Lambda^m\mathfrak{g})^{(0)}$.
\end{proposition}
\begin{proof} If $w\in (\Lambda^m\mathfrak{g})^{\mathfrak{g}}$, then $w=\sum_{\alpha \in G}w^{(\alpha)}$ for a uniquely determined family of elements $w^{(\alpha)}\in (\Lambda^m\mathfrak{g})^{(\alpha)}$ for every $\alpha\in G$. Lemma \ref{LemgDe} yields that each $w^{(\alpha)}$ belongs to $(\Lambda^m\mathfrak{g})^{\mathfrak{g}}$. For every $e\in \mathfrak{g}^{(0)}$, one has that $[e,w^{(\alpha)}]=\talpha(e)w^{(\alpha)}=0$. Hence, $\talpha=0$. Since the mapping $T$ of the root gradation is an injection, one gets $\alpha=0$. Hence $w\subset (\Lambda^m\mathfrak{g})^{(0)}$.
	\end{proof}
Proposition \ref{gin0} allows us to restrict the search for elements of $(\Lambda^m\mathfrak{g})^\mathfrak{g}$ to $(\Lambda^m\mathfrak{g})^{(0)}$, which restricts the form of the elements of $(\Lambda^m\mathfrak{g})^\mathfrak{g}$. Moreover, $(\Lambda^m\mathfrak{g})^{(0)}$ can also be obtained via the root gradation of $\mathfrak{g}$ as shown in Theorem \ref{thm:root} and other parts of Section \ref{sect:gradation}.

Let us now analyse the behaviour of the $(\Lambda^m\mathfrak{g})^{(\alpha)}$ relative to the $\mathfrak{g}$-invariant metrics on $\Lambda\mathfrak{g}$. This illustrates the structure of each $\Lambda^m\mathfrak{g}$ and facilitates finding the elements of $(\Lambda^m\mathfrak{g})^\mathfrak{g}$. To show these points, we start by proving the following result.

\begin{theorem}\label{PerpL} 
If $b$ is a $\mathfrak{g}$-invariant bilinear symmetric map on $\mathfrak{g}$, then $b_{\Lambda^m\mathfrak{g}}(v_{(\alpha)}, v_{(\beta)})=0$ for every $v_{(\alpha)} \in \Lambda^m\mathfrak{g}^{(\alpha)}$ and $v_{(\beta)} \in \Lambda^m\mathfrak{g}^{(\beta)}$ with $\alpha+\beta\neq 0$. 
\end{theorem}
\begin{proof} Since $b_{\Lambda \mathfrak{g}}$ is $\mathfrak{g}$-invariant, one gets  that
$
b_{\Lambda^m\mathfrak{g}}([h,v_{(\alpha)}]_{S}, v_{(\beta)})=-b_{\Lambda^m\mathfrak{g}}(v_{(\alpha)},[h,v_{(\beta)}]_{S})$, for every $h\in \mathfrak{g}^{(0)}$,  $v_{(\alpha)} \in \Lambda^p\mathfrak{g}^{(\alpha)}$ and $\forall v_{(\beta)} \in \Lambda^p\mathfrak{g}^{(\beta)}.
$
Hence,
$
(\talpha+\tbeta)(h)b_{\Lambda^m\mathfrak{g}}(v_{(\alpha)},v_{(\beta)})=0.
$
Since $\alpha+\beta\neq 0$ by assumption, then the injectivity of the map  $T$ of the root gradation gives that $\talpha+\tbeta\neq 0$. Hence, there exists $h\in \mathfrak{g}^{(0)}$ such that $(\talpha+\tbeta)(h)\neq 0$ and the theorem follows.
\end{proof}

\begin{example} Let us illustrate Theorem \ref{PerpL} for $\Lambda^2\mathfrak{so}(2,2)$. Using the basis $\{e_-,e_0,e_+,f_-,f_0,f_+\}$ of $\mathfrak{so}(2,2)$ used in Example \ref{ex:decomp_so22}, we obtain that
$$
\kappa_{\mathfrak{so}(2,2)}(e_0, e_0) = \kappa_{\mathfrak{so}(2,2)}(f_0, f_0) = 2, \quad \kappa_{\mathfrak{so}(2,2)}(e_{-}, e_{+}) = \kappa_{\mathfrak{so}(2,2)}(f_{-}, f_{+}) = 2.
$$
The previous calculation enables us to determine an orthogonal basis of $\Lambda^2\mathfrak{so}(2,2)$ relative to the induced form $\kappa_{\Lambda^2\mathfrak{so}(2,2)}$:
\begin{gather*}
\{e_{-} \!\wedge\! f_{-} \pm e_{+} \!\wedge\! f_{+},\quad
e_{+} \!\wedge\! f_{-} \pm e_{-} \!\wedge\! f_{+},\quad
e_{0} \!\wedge\! f_{-} \pm e_{0} \!\wedge\! f_{+},\quad
f_{-} \!\wedge\! f_{0} \pm f_{0} \!\wedge\! f_{+},\\
 e_{-} \!\wedge\! e_{0} \pm e_{0} \!\wedge\! e_{+},\quad
 e_{-} \!\wedge\! f_{0} \pm e_{+} \!\wedge\! f_{0}, \quad
 e_{0} \!\wedge\! f_{0},\quad
 e_{-} \!\wedge\! e_{+}, \quad
 f_{-} \!\wedge\! f_{+}\}.
\end{gather*}
One easily sees that this basis satisfies the orthogonality relations relative to $b_{\Lambda^2\mathfrak{so}(2,2)}$ determined by Theorem \ref{PerpL}.
\demo
\end{example}

\begin{corollary}\label{cor:nondeg_subsp} 
If $\mathfrak{g}$ is semi-simple and $b$ is a non-degenerate bilinear symmetric form on $\mathfrak{g}$, then the restrictions of $b_{\Lambda^m\mathfrak{g}}$ to $(\Lambda^m\mathfrak{g})^{(0)}$ and $\Lambda^m \mathfrak{g}^{(\alpha)} \oplus \Lambda^m \mathfrak{g}^{(-\alpha)}$ are non-degenerate. If $b$ is positive-definite or negative-definite, then the restriction of $b_{\Lambda^m\mathfrak{g}}$ to $(\Lambda^m\mathfrak{g})^{\mathfrak{g}}$ is also non-degenerate.
\end{corollary}
\begin{proof} Let us prove both results by reduction to contradiction. 

If the restriction of $b_{\Lambda^m\mathfrak{g}}$ to $(\Lambda^m\mathfrak{g})^{(0)}$ is degenerate, then there exists an element of $(\Lambda^m\mathfrak{g})^{(0)}$  perpendicular (with respect to $b_{\Lambda^m\mathfrak{g}}$) to every element of this space. Theorem \ref{PerpL} yields that this element is perpendicular to $\Lambda^m\mathfrak{g}$ and $b_{\Lambda^m\mathfrak{g}}$ is degenerate, which goes against our initial assumption. Hence, $b_{\Lambda^m\mathfrak{g}}$ is non-degenerate on $(\Lambda^m\mathfrak{g})^{(0)}$.

Similarly, if $w \in \Lambda^m\mathfrak{g}^{(\alpha)}\oplus \Lambda^m\mathfrak{g}^{(-\alpha)}$ is orthogonal to $\Lambda^m\mathfrak{g}^{(\alpha)}\oplus \Lambda^m\mathfrak{g}^{(-\alpha)}$, then it stems from Theorem \ref{PerpL} that $w$ is orthogonal to $\Lambda^m\mathfrak{g}$, which is a contradiction to our initial hypothesis concerning the non-degeneracy of $b_{\Lambda^m\mathfrak{g}}$. 

Finally, suppose that the restriction of $b_{\Lambda^m\mathfrak{g}}$ to $(\Lambda^m\mathfrak{g})^{\mathfrak{g}}$ is degenerate.  In view of previous paragraphs,  $b_{\Lambda^m\mathfrak{g}}$ is non-degenerate on $(\Lambda^m\mathfrak{g})^{(0)}$. Since $b$ is definite, then $b_{\Lambda^m\mathfrak{g}}$ is definite (see Proposition \ref{Cor:Diatwoform}) on $(\Lambda^m\mathfrak{g})^{(0)}$ and the orthogonal to $(\Lambda^m\mathfrak{g})^{\mathfrak{g}}$ within this space is also a complementary subspace. Hence, if an element is perpendicular to the whole $(\Lambda^m\mathfrak{g})^{\mathfrak{g}}$, it will also be perpendicular to the whole  $(\Lambda^m\mathfrak{g})^{(0)}$ and hence to $\Lambda^m \mathfrak{g}$. This is a contradiction and $b_{\Lambda^m\mathfrak{g}}$ must be non-degenerate on $(\Lambda^m\mathfrak{g})^{\mathfrak{g}}$.
\end{proof}

\begin{example} Corollary \ref{cor:nondeg_subsp} can be illustrated by varifying, after a short calculation, that the basis of $(\Lambda^3\mathfrak{so}(2,2))^{(0)}\mathfrak{so}(2,2)$ given by (see Table \ref{so22_bases})
$
 e_{-} \wedge e_0 \wedge e_{+},\,\, e_{-} \wedge e_{+} \wedge f_0,\,\, e_0 \wedge f_{-} \wedge f_{+},\,\, f_{-} \wedge f_0 \wedge f_{+}.
$
is orthogonal relative to $\kappa_{\Lambda^3\mathfrak{so}(2,2)}$ as claimed in Corollary \ref{cor:nondeg_subsp}.
\end{example}

Theorem \ref{PerpL} and Corollary \ref{cor:nondeg_subsp} simplify the determination of elements of $(\Lambda^3\mathfrak{g})^\mathfrak{g}$. Assume for instance the case of $b$ being positive (or negative) definite. If $w\in (\Lambda^3\mathfrak{g})^\mathfrak{g}$, then the orthogonal $W$ to $w$ within $(\Lambda^3\mathfrak{g})^{(0)}$ contains a set of elements that along with $w$ span the whole $ (\Lambda^3\mathfrak{g})^\mathfrak{g}$. In other words, it is enough to look for the remaining elements of $ (\Lambda^3\mathfrak{g})^\mathfrak{g}$ within $W$.

\subsection{Nilpotent Lie algebras}\label{NilMeth}

Every nilpotent Lie algebra possesses a flag of ideals, called the \textit{lower central series} of $\mathfrak{g}$, defined recurrently as $\mathfrak{g}_{s)} :=[\mathfrak{g},\mathfrak{g}_{s-1)}]$ for $s\in \mathbb{N}$ with $\mathfrak{g}_{0)}:=\mathfrak{g}$. Then, $\mathfrak{g}\supset \mathfrak{g}_{1)} \supset \ldots\supset \mathfrak{g}_{p-1)}\supset\mathfrak{g}_{p)}=\{0\}$. Let us use this fact to study $(\Lambda \mathfrak{g})^{\mathfrak{g}}$. First, the nilpotency allows for the characterisation of certain decomposable elements of $(\Lambda^m\mathfrak{g})^{\mathfrak{g}}$ via Lie subalgebras of $\mathfrak{g}$. This is done in the next proposition.
\begin{proposition} 
If $\mathfrak{g}$ is nilpotent, then every non-zero decomposable element of $(\Lambda\mathfrak{g})^\mathfrak{g}$ expands the space $\Lambda^{{\rm dim }\mathfrak{h}}\mathfrak{h}$ of a non-zero nilpotent ideal $\mathfrak{h}$ of $\mathfrak{g}$.
\end{proposition}
\begin{proof}
Theorem \ref{Theo:InvLk} shows that every non-zero decomposable element of $(\Lambda\mathfrak{g})^\mathfrak{g}$ gives rise to a non-zero ideal $\mathfrak{h}\subset \mathfrak{g}$.
Since $\mathfrak{g}$ is nilpotent, $\mathfrak{h}$ is nilpotent as well. Conversely, if $\mathfrak{g}$ is nilpotent, then ${\rm ad}_v$, for every $v \in \mathfrak{g}$, is a nilpotent map on $\mathfrak{g}$. If $\mathfrak{h}$ is an ideal, then ${\rm ad}_v \vert_{\mathfrak{h}}$ is also nilpotent. Then, Theorem \ref{Theo:InvLk} yields that $\Lambda^{\dim\mathfrak{h}}\mathfrak{h}\subset (\Lambda\mathfrak{g})^\mathfrak{g}$ is generated by a decomposable element.
\end{proof}

\begin{proposition}\label{lambda2nil}
If $\dim\mathfrak{z}(\mathfrak{g})=1$, then $\mathfrak{z}(\mathfrak{g})\wedge \mathfrak{g}_{p-2)} \subset (\Lambda^2 \mathfrak{g})^{\mathfrak{g}}$.
\end{proposition}
\begin{proof}
Since $\mathfrak{g}_{p-1)} \subset \mathfrak{z}(\mathfrak{g})$ and using the properties of the algebraic Schouten bracket,
$$
[v, \mathfrak{z}(\mathfrak{g}) \wedge \mathfrak{g}_{p-2)}] = \mathfrak{z}(\mathfrak{g}) \wedge [v, \mathfrak{g}_{p-2)}] \subset \mathfrak{z}(\mathfrak{g}) \wedge \mathfrak{g}_{p-1)} \subset \mathfrak{z}(\mathfrak{g}) \wedge\mathfrak{z}(\mathfrak{g}) = \{0\}, \quad \forall v \in \mathfrak{g}.
$$
\vskip-0.6cm
\end{proof}

\begin{proposition}\label{lambda3nil}
For a nilpotent Lie algebra $\mathfrak{g}$ such that $\mathfrak{g}_{p)}=0$ and $\mathfrak{g}_{p-1)}\neq 0$, one has:
\begin{enumerate}
\item If $\dim\mathfrak{z}(\mathfrak{g})=2$, then $\Lambda^2\mathfrak{z}(\mathfrak{g})\wedge \mathfrak{g}_{p-2)} \subset (\Lambda^3 \mathfrak{g})^{\mathfrak{g}}$. 
\item If $\dim\mathfrak{z}(\mathfrak{g})=1$ and $\dim\mathfrak{g}_{p-2)}>1$, then $\mathfrak{z}(\mathfrak{g})\wedge\Lambda^2\mathfrak{g}_{p-2)} \subset (\Lambda^3 \mathfrak{g})^{\mathfrak{g}}$.
\item If $\dim\mathfrak{z}(\mathfrak{g})=1$ and $\dim\mathfrak{g}_{p-2)}=1$, then $\mathfrak{z}(\mathfrak{g})\wedge\mathfrak{g}_{p-2)}\wedge \mathfrak{g}_{p-3)} \subset (\Lambda^3 \mathfrak{g})^{\mathfrak{g}}$.\end{enumerate}
\end{proposition}
\begin{proof}
Let us prove 1), 2), and 3) by verifying the given inclusions on decomposable elements. The general case follows from it. 
To prove the first case consider $a,b\in \mathfrak{z}(\mathfrak{g})$ and $c\in \mathfrak{g}_{p-2)}$. Then $[v,a\wedge b\wedge c]_S=a\wedge b\wedge [v,c]$. Since $[v,\mathfrak{g}_{p-2)}]\subset \mathfrak{g}_{p-1)}\subset \mathfrak{z}(\mathfrak{g})$ and $\dim \mathfrak{z}(\mathfrak{g})=2$, it follows that $a\wedge b\wedge[v,c]\in \Lambda^3\mathfrak{z}(\mathfrak{g})=\{0\}$ and $a\wedge b\wedge c \in (\Lambda^3\mathfrak{g})^\mathfrak{g}$.

To prove the second formula, assume $a\in \mathfrak{z}(\mathfrak{g})$ and $b,c\in \mathfrak{g}_{p-2)}$. Then, $[v,a\wedge b\wedge c]_S=a\wedge([v,b]\wedge c+b\wedge [v,c])$. Obviously $[v,b],[v,c]\in \mathfrak{z}(\mathfrak{g})$ and $[v,a\wedge b\wedge c]_S=a\wedge([v,b]\wedge c+b\wedge [v,c])\in \Lambda^2\mathfrak{z}(\mathfrak{g})\wedge \mathfrak{g}_{p-2)}$. Since $\dim\mathfrak{z}(\mathfrak{g})=1$,  one has that $\Lambda^2 \mathfrak{z}(\mathfrak{g}) = \{0\}$ and $a\wedge b\wedge c \in (\Lambda^3\mathfrak{g})^\mathfrak{g}$. 

The third case is similar to previous ones. Assume $a\in \mathfrak{g}_{p-1)}=\mathfrak{z}(\mathfrak{g}), b\in\mathfrak{g}_{p-2)},c\in \mathfrak{g}_{p-3)}$. Then, using the assumptions on the dimensions of $\mathfrak{z}(\mathfrak{g}),\mathfrak{g}_{p-2)},\mathfrak{g}_{p-3)}$, we obtain
\vskip -0.5cm
$$
[v,a]\wedge b\wedge c=0,\,a\wedge [v,b]\wedge c \in \Lambda^2\mathfrak{z}(\mathfrak{g})\wedge \mathfrak{g}_{p-3)}=\{0\},\,a\wedge b\wedge [v,c]\in \mathfrak{z}(\mathfrak{g})\wedge \Lambda^2\mathfrak{g}_{p-2)}=\{0\}.
$$
\vskip -0.25cm
In view of the expression for $[v,a\wedge b \wedge c]_S$ and previous relations, it follows that $a\wedge b\wedge c\in (\Lambda^3\mathfrak{g})^\mathfrak{g}$. 
\end{proof}

\section{Reduced mCYBEs}
If $\dim \mathfrak{g}\geq 4$, mCYBEs are frequently very complicated to solve. This section shows a simplification of the mCYBE concerning, mainly, not semi-simple $\mathfrak{g}$ obtained by mapping structures of  $\Lambda^m\mathfrak{g}$ onto $\Lambda^m_R\mathfrak{g}:=\Lambda^m\mathfrak{g}/(\Lambda^m\mathfrak{g})^\mathfrak{g}$.

\begin{definition}
The elements of $\Lambda^m_R\mathfrak{g}$, for $m \in \mathbb{Z}$, are called {\it reduced $m$-vectors}.
\end{definition}

\begin{proposition}\label{prop:reduced}
Let $\pi_p:w_p\in \Lambda^p\mathfrak{g}\mapsto [w_p]\in\Lambda^p_R\mathfrak{g}$, with $p \in \mathbb{Z}$. The algebraic Schouten bracket induces a new bracket, called the reduced Schouten bracket, on $
\Lambda_R\mathfrak{g}:=\bigoplus_{p\in \mathbb{Z}}\Lambda^p_R\mathfrak{g},$ 
of the form
\begin{equation}\label{RSN}
{[[w_p],[w_q]]_R}:=[[w_p,w_q]_{S}],\qquad \forall w_p\in \Lambda^p\mathfrak{g},\quad \forall w_q\in\Lambda^q\mathfrak{g}.
\end{equation}
This new bracket induces a decomposition on $\Lambda_R\mathfrak{g}$ compatible with $[\cdot, \cdot]_R$ in such a way that $\pi=\bigoplus_{p\in \mathbb{Z}}\pi_p$ satisfies that $[\pi(a), \pi(b)]_R = \pi([a,b]_S)$ for arbitrary $a,b \in \Lambda\mathfrak{g}$ and $\pi(\Lambda^p\mathfrak{g})\subseteq \Lambda^p_R\mathfrak{g}$ for any $p \in \mathbb{Z}$.
Then, $r\in \Lambda^2\mathfrak{g}$ is an $r$-matrix if and only if $[\pi(r),\pi(r)]_{R}=0$.
\end{proposition}
\begin{proof} 
Let us show that (\ref{RSN}) is well defined. If $[w_p] = [\bar{w}_p]$ and $[w_q] = [\bar{w}_q]$ for $w_p, \bar{w}_p\in \Lambda^p\mathfrak{g}$ and $w_q, \bar{w}_q\in \Lambda^q\mathfrak{g}$, then $w_p - \bar{w}_p,w_q-\bar w_q \in (\Lambda\mathfrak{g})^{\mathfrak{g}}$ and $[(\Lambda\mathfrak{g})^\mathfrak{g},\Lambda \mathfrak{g}]_{S}=0$. Hence,
$$
[[w_p],[w_q]]_R :=[[w_p,w_q]_{S}]=[[w_p-\bar{w}_p+\bar{w}_p,w_q-\bar w_q+\bar w_q]_{S}]=[[\bar w_p],[\bar w_q]]_R.
$$

To prove that $\Lambda_R\mathfrak{g}$ is a graded algebra relative to the reduced  bracket, it is enough and immediate to see that the reduced bracket (\ref{RSN}) satisfies that $[\Lambda^p_R\mathfrak{g},\Lambda^q_R\mathfrak{g}]_R\subset \Lambda^{p+q-1}_R\mathfrak{g}$ and
\begin{enumerate}
\item $[[w_p], [w_q]]_R= -(-1)^{(p-1)(q-1)} [[w_q], [w_p]]_R$,
\item $(-1)^{(p-1)(s-1)} [[w_p], [[w_q], [w_s]]_{R}] _R+ (-1)^{(q-1)(p-1)} [[w_q], [[w_s], [w_p]]_{R}]_R \\+ (-1)^{(s-1)(q-1)} [[w_s], [[w_p], [w_q]]_{R}]_R = 0,$
\end{enumerate}
for all $w_p \in \Lambda^p \mathfrak{g},w_q \in \Lambda^q \mathfrak{g}, w_s \in \Lambda^s \mathfrak{g}$.

The relation $[\pi(a),\pi(b)]_R=\pi([a,b]_S)$ is an immediate consequence of (\ref{RSN}). Meanwhile, if $r$ is an $r$-matrix, then  $[\pi(r),\pi(r)]_R=\pi([r,r]_S)\in \pi((\Lambda^3\mathfrak{g})^{\mathfrak{g}})=0$. The converse is trivial.
\end{proof}

Our aim now is to prove Proposition \ref{Coho}, which shows that there exists a cohomology on $\Lambda \mathfrak{g}^*\otimes \Lambda_R^2\mathfrak{g}$ which characterises cocommutators on $\mathfrak{g}$. In fact, $\Lambda^2_R \mathfrak{g}$ is a $\mathfrak{g}$-module as shown in the following lemma.

\begin{lemma}\label{reduced_action}
The pair  $(\Lambda^m_R\mathfrak{g},\, \sigma: v\in \mathfrak{g} \mapsto [[v], \cdot]_R \in \mathfrak{gl}(\Lambda^m_R \mathfrak{g}))$ is a $\mathfrak{g}$-module and $\Psi:T\in {\rm Aut}(\mathfrak{g})\mapsto [\Lambda^mT]\in { GL}(\Lambda^m_R\mathfrak{g})$, with $[\Lambda^mT]([w]):=[\Lambda^mT(w)]$ for every $w\in \Lambda^m\mathfrak{g}$, is a Lie group action.
\end{lemma}
\begin{proof}
Let us show that $(\Lambda^m_R\mathfrak{g},\sigma)$ is a $\mathfrak{g}$-module. The mapping $\sigma$ is well defined since the reduced bracket is well defined. Moreover, $\sigma$ is a Lie algebra homomorphism because the property 2) of the reduced bracket in the proof of Proposition \ref{prop:reduced} for $p,q=1$, and arbitrary $s$ reduces to
$$
[[w_1],[[w_1'],[w_s]]_R]_R+[[w_1'],[[w_s],[w_1]]_R]_R+[[w_s],[[w_1],[w_1']]_R]_R=0
$$
and therefore
$$
\sigma(w_1)\sigma(w_1')([w_s])-\sigma(w_1')\sigma(w_1)([w_s])=\sigma([w_1,w_1'])([w_s]).
$$

Let us now prove that ${\rm Aut}(\mathfrak{g})$ acts on $\Lambda^m_R\mathfrak{g}$. In this respect, it is only necessary to verify that $\Psi(T)$ must be unambiguous. This amounts to proving that if $w,w'\in [w]$, then $[\Lambda^mT]([w])=[\Lambda^mT]([w'])$. Note that $\Lambda^mT(w)=\Lambda^mT(w-w'+w')=\Lambda^mT(w-w')+\Lambda^mT(w')$. Using that $\Lambda^mT(\Lambda^m\mathfrak{g})^\mathfrak{g}\subset (\Lambda^m\mathfrak{g})^\mathfrak{g}$
in virtue of Proposition \ref{prop:lgg}, we obtain that $[\Lambda^mT(w)]=[\Lambda^mT(w')]$. Therefore, $[\Lambda^mT]([w])=[\Lambda^mT(w)]=[\Lambda^mT(w')]=[\Lambda^mT]([w'])$ and $[\Lambda^mT]$ is well-defined.
\end{proof}

Using Lemma \ref{reduced_action}, we now obtain $\mathfrak{g}$-invariant maps on $\Lambda^m_R\mathfrak{g}$ out of certain $\mathfrak{g}$-invariant maps on $\Lambda^m\mathfrak{g}$. 

\begin{proposition}\label{RedgInv} If $b:\Lambda^m\mathfrak{g}\rightarrow \mathbb{R}$ is a symmetric or anti-symmetric  $\mathfrak{g}$-invariant $k$-linear map and its kernel contains $(\Lambda^m\mathfrak{g})^\mathfrak{g}$, then we obtain a $\mathfrak{g}$-invariant $k$-linear map $b_R$ on $\Lambda_R^m\mathfrak{g}$ given by
\begin{equation}\label{gInvRed}
b_R([w_1],\ldots,[w_k]):=b(w_1,\ldots,w_k),\qquad \forall w_1,\ldots,w_k\in \Lambda^m\mathfrak{g}.
\end{equation}
\end{proposition}
\begin{proof} The map (\ref{gInvRed}) is well defined because if $w'_i\in [w_i]$ for $i\in \overline{1,k}$, then $w_i-w'_i\in (\Lambda^m\mathfrak{g})^{\mathfrak{g}}$.  Since the kernel of $b$ contains $(\Lambda^m\mathfrak{g})^{\mathfrak{g}}$, one has 
$$
b_R([w_1],\ldots,[w_k])\!=\!b(w_1',\ldots,w_{k-1}',w_k)\!=\!b(w'_1, \ldots,w_k-w_k'+w_k')\!=\!b(w_1',\ldots,w'_k)\!=\!b_R([w_1'],\ldots,[w_k'])
$$
and the value of $b_R$ does not depend on the representative of each particular equivalence class of $\Lambda^m_R\mathfrak{g}$.

The $\mathfrak{g}$-invariance of $b_R$ stems from the following relations
$$
b_R(\sigma(v)[w_1],\ldots,[w_k])+\ldots+b_R([w_1],\ldots,\sigma(v)[w_k])=b([v,w_1]_S,\ldots,w_k)+\ldots+b(w_1,\ldots,[v,w_k]_S)=0.
$$

\end{proof}
\begin{proposition} \label{Coho}
	There is a natural cohomology complex on the spaces $\Lambda^q\mathfrak{g}^*\otimes \Lambda^2_R\mathfrak{g}$ making the following diagram commutative:
	\begin{displaymath}
	\xymatrix{
		\mathbb{R}\otimes \Lambda^2\mathfrak{g}\ar[d]^{{\rm id}\otimes \pi_2}\ar[r]^{{\rm d}} & \mathfrak{g}^* \otimes \Lambda^2\mathfrak{g}\ar[r]^{{\rm d}}\ar[d]^{{\rm id}\otimes\pi_2} & \Lambda^2\mathfrak{g}^*\otimes \Lambda^2\mathfrak{g} \ar[r]^{{\rm d}} \ar[d]^{{\rm id}\otimes\pi_2} & \Lambda^3\mathfrak{g}^*\otimes \Lambda^2\mathfrak{g} \ar[r]^-{{\rm d}} \ar[d]^{{\rm id}\otimes\pi_2} & \ldots \\
		\mathbb{R}\otimes \Lambda^2_R\mathfrak{g}\ar[r]^{{\rm d}_R} &\mathfrak{g}^* \otimes \Lambda^2_R\mathfrak{g}\ar[r]^{{\rm d}_R} & \Lambda^2\mathfrak{g}^*\otimes \Lambda^2_R\mathfrak{g} \ar[r]^{{\rm d}_R}& \Lambda^3\mathfrak{g}^*\otimes \Lambda_R^2\mathfrak{g} \ar[r]^-{{\rm d}_R} & \ldots \\
	}
	\end{displaymath}
\end{proposition}
\begin{proof}
	It is enough to see that if $w\in \Lambda^2\mathfrak{g}$ and $\theta\in \Lambda^m\mathfrak{g}^*$, then
	\vskip -0.7cm
	\begin{multline*}
	\textrm{d}_R(\theta \otimes [w])(v_1,\ldots,v_{k+1}) := \sum_{i=1}^{k+1} (-1)^{i+1} \theta(v_1, \ldots , \hat{v_i},\ldots,v_{k+1}) \otimes [[v_i], [w]]_{R} \\+ \sum_{\substack{p,q=1\\p<q}}^{k+1} (-1)^{p+q} \theta([v_p, v_q] ,v_1 ,\ldots,\hat v_p,\ldots,\hat v_q,\ldots,v_{k+1}) \otimes [w],
	\end{multline*}
	\vskip -0.4cm \noindent where the hatted elements are dropped, is $(m+1)$-linear and anti-symmetric. Moreover, ${\rm d}_R(\theta \otimes [r]) = [{\rm d}(\theta \otimes r)]$, the commutativity of the diagram is straightforward and ${\rm d}_R^2=0$ stems immediately from the fact that ${\rm d}^2=0$.
\end{proof}

\section{On the automorphisms of a Lie algebra}\label{sect:auto_alg}

This section investigates $\textrm{Aut}(\mathfrak{g})$ and its relations to ${\rm Aut}(\mathfrak{g})$-invariant metrics on $\Lambda\mathfrak{g}$. This will be used to classify solvable Lie bialgebras in Section \ref{Classification}.

\begin{proposition}\label{aut} 
Let $\mathfrak{der}(\mathfrak{g})$ be a Lie algebra of derivations of $\mathfrak{g}$. Then, $\mathfrak{der}(\mathfrak{g}) \simeq \mathfrak{aut}(\mathfrak{g})$.
	\end{proposition}
\begin{proof} 
The $\mathfrak{aut}(\mathfrak{g})$ is spanned by the tangent vectors to curves $\gamma:t\in \mathbb{R}\mapsto T_t \in {\rm Aut}(\mathfrak{g})$ such that $\gamma(0) = \textrm{Id}$. Since $\{T_t\}_{t\in\mathbb{R}} \subset{\rm Aut}(\mathfrak{g}) $ and defining $D(v):=\frac{{\rm d}}{{\rm d}t}\big|_{t=0}T_t(v)$ for every $v\in \mathfrak{g}$, one has $[T_t(v_1),T_t(v_2)]-T_t[v_1,v_2]=0$ and therefore
	\begin{equation}\label{Det:D}
\frac{{\rm d}}{{\rm d}t}\bigg|_{t=0}([T_t(v_1),T_t(v_2)]-T_t[v_1,v_2])=0,\,\, \forall v_1,v_2\in \mathfrak{g}
\Rightarrow 	D([v_1,v_2])=[D(v_1),v_2]+[v_1,D(v_2)],\,\, \forall v_1,v_2\in \mathfrak{g}.
	\end{equation}
	In other words, $D$ is a derivation of $\mathfrak{g}$. 
	Conversely, every $D \in \mathfrak{der}(\mathfrak{g})$ gives rise to a curve $\gamma:t\in\mathbb{R}\mapsto T_t:=\exp(tD)\in GL(\mathfrak{g})$. Since $T_0={\rm Id}$ and $D \in \mathfrak{der}(\mathfrak{g})$, one has that $T_t\in {\rm Aut}(\mathfrak{g})$ for every $t\in \mathbb{R}$.
	\end{proof}

Derivations of $\mathfrak{g}$ can be obtained by determining those  $T\in \mathfrak{gl}(\mathfrak{g})$ satisfying the right-hand side of (\ref{Det:D}), which can be solved via computer programs even for relatively high-dimensional Lie algebras.

Proposition \ref{aut} also provides information about the connected part of the neutral element of ${\rm Aut}(\mathfrak{g})$, namely ${\rm Aut}_c(\mathfrak{g})$. Unfortunately, ${\rm Aut}(\mathfrak{g})$ need not be connected and the determination of its different connected parts can be tricky. 

To illustrate our above claim, consider $\mathfrak{sl}_2$. The Killing metric on $\mathfrak{sl}_2$, given by (\ref{sl2A}) in the basis $\{e_1,e_2,e_3\}$ indicated in Table \ref{tabela3w}, is indefinite and non-degenerate with signature $(2,1)$. The quadratic function on $\mathfrak{sl}_2$ induced by this Killing metric is given by $f(xe_1+ye_2+ze_3)=2x^2+4yz$. The surfaces, $S_k$, consists of points $(x, y, z)$, where $2x^2+4yz=k$. If $k<0$ such surfaces are two-sheeted hyperboloids contained in the region of $\mathfrak{sl}_2$ with $z>0$ or in the region of $\mathfrak{sl}_2$ with $z<0$.
The space ${\rm Aut}(\mathfrak{sl}_2)$ consists of isometries of the Killing metric. The connected part of $\textrm{Id}$ in ${\rm Aut}(\mathfrak{sl}_2)$ leaves invariant the elements of each component of a two-sheeted hyperboloid. The element $T \in {\rm Aut}(\mathfrak{sl}_2)$ such that $T(e_1)=-1$, $T(e_2)=-e_3$, and $T(e_3)=-2$ does not preserve the sign of the coordinate $z$. Consequently, $T\notin  {\rm Aut}_c(\mathfrak{sl}_2)$ and ${\rm Aut}(\mathfrak{sl}_2)$ is not connected. Since ${\rm Inn}(\mathfrak{g})$ is connected, ${\rm Inn}(\mathfrak{sl}_2)\neq {\rm Aut}(\mathfrak{sl}_2)$ and the assumption ${\rm Inn}(\mathfrak{sl}_2)={\rm Aut}(\mathfrak{sl}_2)$, made in \cite{Farinati}, is incorrect.

The ${\rm Aut}(\mathfrak{g})$-invariant metrics are easier to obtain than ${\rm Aut}(\mathfrak{g})$, e.g. by Proposition \ref{prop:glrho_ginv} they are a subclass of $\mathfrak{der}(\mathfrak{g})$-invariant metrics . It will be shown in Section \ref{Classification} that this will frequently be enough to characterize coboundary real three-dimensional Lie bialgebras.

To use the above fact in practical applications is convenient to enunciate the following result.

\begin{theorem} If $b$ is a $k$-linear map on $\mathfrak{g}$ invariant under ${\rm Aut}(\mathfrak{g})$, then its extension $b_{\Lambda^m\mathfrak{g}}$ is invariant under the action of ${\rm Aut}(\mathfrak{g})$.  
	\end{theorem}

The crux now is that if $b$ is a $k$-linear symmetric metric on $\mathfrak{g}$ invariant relative to ${\rm Aut}(\mathfrak{g})$, then the spaces $S_k$ where the polynomial on $\Lambda^m\mathfrak{g}$ of the form
$
p(v):=b_{\Lambda^m\mathfrak{g}}(v,\ldots, v)$, for all $v\in \Lambda^m\mathfrak{g},
$
takes a constant value $k$ are invariant under the action of ${\rm Aut}(\mathfrak{g})$ on $\Lambda^m\mathfrak{g}$. The orbits of ${\rm Aut}(\mathfrak{g})$ on $\Lambda^m\mathfrak{g}$ need not be connected, but they must be contained in a single $S_k$. Using that Inn$(\mathfrak{g})$ can be relatively easily obtained and it gives information on the connected components of ${\rm Aut}(\mathfrak{g})$, we can investigate the action of the whole ${\rm Aut}(\mathfrak{g})$ by searching elements connecting the different orbits of {\rm Inn}$(\mathfrak{g})$ within the same $S_k$. This process will be illustrated in Section \ref{Classification}.

Let us now provide hints to characterize automorphisms for Lie algebras. More specifically, let us analyse properties of $\Lambda^2T$ for every $T\in {\rm Aut}(\mathfrak{g})$.

In the case of complex simple or semi-simple Lie algebras, the space of Lie algebra automorphisms is determined by the inner automorphisms of the Lie algebra, which already had a characterization in this work, and the Dynkin diagram \cite{Ha15,Jacobson}. Meanwhile, automorphisms of general Lie algebras cannot be determined so easily. In particular, we focus upon automorphisms of solvable and nilpotent Lie algebras.

Consider for instance a solvable or nilpotent Lie algebra $\mathfrak{g}$. The \textit{derived} and \textit{lower central series} are defined recurrently as a the sequence of ideals given by \cite{Ha15}:
$$
\mathfrak{g}^{p)}:=[\mathfrak{g}^{p-1)},\mathfrak{g}^{p-1)}],\qquad\mathfrak{g}_{p)}:=[\mathfrak{g},\mathfrak{g}_{p-1)}],\qquad \mathfrak{g}^{0)}:=\mathfrak{g}_{0)}:=\mathfrak{g},\qquad \forall p\in \mathbb{N}.
$$ 
Moreover, if $T\in {\rm Aut}(\mathfrak{g})$, then $T\mathfrak{g}=\mathfrak{g}$ and $T\mathfrak{g}_{p)}=[T\mathfrak{g},T\mathfrak{g}_{p-1)}]=[\mathfrak{g},T\mathfrak{g}_{p-1)}]$. By induction, $T\mathfrak{g}_{p)}=\mathfrak{g}_{p)}$ for $p\in \mathbb{N} \cup \{0\}$. A similar result applies to  derived series. 

Given a solvable Lie algebra $\mathfrak{g}$, the elements of ${\rm Aut}(\mathfrak{g})$ leave invariant the elementary sequence
$$
\mathfrak{s}_{pq}:=\mathfrak{g}_{p)}\wedge \mathfrak{g}_{q)}, \quad p\leq q,\quad p,q\in \mathbb{N}\cup \{0\}.
$$
If $\mathfrak{s}_{pq}\neq 0$, then $\mathfrak{s}_{pq} \supset \mathfrak{s}_{lm}$ if and only if $p\leq l$ and $q\leq m$. Similar results apply to the spaces 
$$
\mathfrak{s}^{pq}:=\mathfrak{g}^{p)}\wedge \mathfrak{g}^{q)}, \quad p\leq q,
$$
for $p,q\in \mathbb{N}\cup \{0\}$. Above relations allow to estimate the form of $\Lambda^2T$.

\section{Study of real three-dimensional coboundary Lie bialgebras}\label{Classification}

This section exploits previous techniques to analyse and to classify, up to Lie algebra automorphisms, coboundary real three-dimensional Lie bialgebras. The use of gradations allows us to obtain $\mathfrak{g}$-invariant elements of Lie bialgebras and to obtain, relatively easily, solutions to CYBEs. Instead of using all automorphisms in the classification problem of Lie bialgebras, which is complicated (cf. \cite{Farinati}), we focus on the classification up to inner Lie algebra automorphisms, which is easier. Next, the derivation of a few not inner automorphisms leads to the final classification. Our results retrieve geometrically findings in \cite{Farinati, Gomez}, solve minor gaps in these works, and provide a new approach. 

\begin{landscape}
\begin{table}[ht]
\centering
\begin{tabular}{|c|c|c|c|c|c|c|c|c|c|}
\hline
 & $[e_1,e_2]$ &$[e_1,e_3]$&$[e_3,e_2]$ & $\mathcal{S}_k$& $G$ & $\mathfrak{g}$ &$\Lambda^2 \mathfrak{g}$&$\Lambda^3 \mathfrak{g}$ &Root\\
\hline 
$\mathfrak{sl}_2$& $e_2$ &$-e_3,$&$ -e_1$& \parbox[c]{4.6cm}{\vspace{3pt}$2xy-z^2 = k$, $k \in \mathbb{R}_+$\\$2xy-z^2=0, x^2+y^2\neq 0,k=0$\\$2xy-z^2 = k$, $k \in \mathbb{R}_-$\vspace{3pt}}&
$\mathbb{Z}$&
\parbox[c]{2.9cm}{
\centering
\begin{tikzpicture}[scale = 0.5]
{
\filldraw (0,0) circle (1pt) node[above]{$e_1$} node[below]{$(0)$};
\draw [-] (0,0)--(2,0) node[above]{$e_2$};
\filldraw (2,0) circle (1pt) node[below]{$(1)$};
\draw [-] (0,0)--(-2,0) node[above]{$e_3$};
\filldraw (-2,0) circle (1pt) node[below]{$(-1)$};
}
\end{tikzpicture}
}
&
\parbox[c]{2.9cm}{
\centering
\begin{tikzpicture}[scale = 0.5]
{
\filldraw (0,0) circle (1pt) node[above]{$e_{23}$} node[below]{$(0)$};
\draw [-] (0,0)--(2,0) node[above]{$\mathbf{e_{12}}$};
\filldraw (2,0) circle (1pt) node[below]{$(1)$};
\draw [-] (0,0)--(-2,0) node[above]{$\mathbf{e_{13}}$};
\filldraw (-2,0) circle (1pt) node[below]{$(-1)$};
}
\end{tikzpicture}
}
&
\parbox[c]{1.5cm}{
\centering
\begin{tikzpicture}[scale = 0.5]
{
\filldraw (0,0) circle (1pt) node[above]{$e_{123}$} node[below]{$(0)$};
\draw [-] (0,0)--(1,0);
\draw [-] (0,0)--(-1,0);
}
\end{tikzpicture}
}&Yes
\\\hline
$\mathfrak{su}_2$ &
$e_3$&$-e_2$&$-e_1$& \parbox[c]{4.6cm}{\vspace{3pt}$x^2 + y^2 + z^2 = k$, $k \in \mathbb{R}_{+}$\vspace{3pt}}&
$\mathbb{Z}_2$&
\parbox[c]{2.9cm}{
\centering
\begin{tikzpicture}[scale = 0.5]
{
\filldraw (0,0) circle (1pt) node[above]{$e_a$} node[below]{$(0)$};
\draw [-] (0,0)--(3,0) node[above]{$e_b, e_c$};
\filldraw (3,0) circle (1pt) node[below]{$(1)$};
}
\end{tikzpicture}
}
&
\parbox[c]{2.9cm}{
\centering
\begin{tikzpicture}[scale = 0.5]
{
\filldraw (0,0) circle (1pt) node[above]{$e_{bc}$} node[below]{$(0)$};
\draw [-] (0,0)--(3,0) node[above]{${ e_{ba}, e_{ac}}$};
\filldraw (3,0) circle (1pt) node[below]{$(1)$};
}
\end{tikzpicture}
}
&
\parbox[c]{1.5cm}{
\centering
\begin{tikzpicture}[scale = 0.5]
{
\filldraw (0,0) circle (1pt) node[above]{$e_{abc}$} node[below]{$(0)$};
\draw [-] (0,0)--(1,0);
\draw [-] (0,0)--(-1,0);
}
\end{tikzpicture}
}&No
\\ \hline
$\mathfrak{h}$&
$e_3$& $0$ &$0$& \parbox[c]{4.6cm}{\vspace{3pt}$z\neq 0$, $k = 1$\vspace{3pt}} &$\mathbb{Z}$&
\parbox[c]{2.9cm}{
\centering
\begin{tikzpicture}[scale = 0.5]
{
\filldraw (0,0) circle (1pt) node[above]{$e_{2}$} node[below]{$(2)$};
\draw [-] (0,0)--(2,0) node[above]{$e_{3}$};
\filldraw (2,0) circle (1pt) node[below]{$(3)$};
\draw [-] (0,0)--(-2,0) node[above]{$e_{1}$};
\filldraw (-2,0) circle (1pt) node[below]{$(1)$};
}
\end{tikzpicture}
}
&
\parbox[c]{2.9cm}{
\centering
\begin{tikzpicture}[scale = 0.5]
{
\filldraw (0,0) circle (1pt) node[above]{$\mathbf{e_{13}}$} node[below]{$(4)$};
\draw [-] (0,0)--(2,0) node[above]{$\mathbf{e_{23}}$};
\filldraw (2,0) circle (1pt) node[below]{$(5)$};
\draw [-] (0,0)--(-2,0) node[above]{$e_{12}$};
\filldraw (-2,0) circle (1pt) node[below]{$(3)$};
}
\end{tikzpicture}
}
&
\parbox[c]{1.5cm}{
\centering
\begin{tikzpicture}[scale = 0.5]
{
\filldraw (0,0) circle (1pt) node[above]{$e_{123}$} node[below]{$(6)$};
\draw [-] (0,0)--(1,0);
\draw [-] (0,0)--(-1,0);
}
\end{tikzpicture}
}&No
\\ \hline
$\mathfrak{r}'_{3,0}$ &$
-e_3$&$e_2$&$0$& \parbox[c]{4.6cm}{\vspace{3pt}$x^2 + y^2 > 0$, $k = 1$\vspace{3pt}}&
$\mathbb{Z}$&
\parbox[c]{2.9cm}{
\centering
\begin{tikzpicture}[scale = 0.5]
{
\filldraw (0,0) circle (1pt) node[above]{$e_{1}$} node[below]{$(0)$};
\draw [-] (0,0)--(3,0);
\filldraw (3,0) circle (1pt) node[above]{$e_{2}, e_3$} node[below]{$(1)$};
}
\end{tikzpicture}
}
&
\parbox[c]{2.9cm}{
\centering
\begin{tikzpicture}[scale = 0.5]
{
\filldraw (0,0) circle (1pt) node[above]{$e_{12}, e_{13}$} node[below]{$(1)$};
\draw [-] (0,0)--(3,0);
\filldraw (3,0) circle (1pt) node[above]{$\mathbf{e_{23}}$} node[below]{$(2)$};
}
\end{tikzpicture}
}
&
\parbox[c]{1.5cm}{
\centering
\begin{tikzpicture}[scale = 0.5]
{
\filldraw (0,0) circle (1pt) node[above]{$e_{123}$} node[below]{$(2)$};
\draw [-] (0,0)--(1,0);
\draw [-] (0,0)--(-1,0);
}
\end{tikzpicture}
}&No
\\\hline $\mathfrak{r}_{3,-1}$&$e_2$ &$-e_3$ &$0$& \parbox[c]{4.6cm}{\vspace{3pt}$xy=0$, $x^2+y^2\neq 0$, $k = 1$, \\ $xy \neq 0 \neq 0$, $k = 2$\vspace{3pt}}&
$\mathbb{Z}$&
\parbox[c]{2.9cm}{
\centering
\begin{tikzpicture}[scale = 0.5]
{
\draw [-] (0,0)--(-2,0);
\filldraw (-2,0) circle (1pt) node[above]{$e_{3}$} node[below]{$(-1)$};
\filldraw (0,0) circle (1pt) node[above]{$e_{1}$} node[below]{$(0)$};
\draw [-] (0,0)--(2,0);
\filldraw (2,0) circle (1pt) node[above]{$e_{2}$} node[below]{$(1)$};
}
\end{tikzpicture}
}
&
\parbox[c]{2.9cm}{
\centering
\begin{tikzpicture}[scale = 0.5]
{
\draw [-] (0,0)--(-2,0);
\filldraw (-2,0) circle (1pt) node[above]{$\mathbf{e_{13}}$} node[below]{$(-1)$};
\filldraw (0,0) circle (1pt) node[above]{$e_{23}$} node[below]{$(0)$};
\draw [-] (0,0)--(2,0);
\filldraw (2,0) circle (1pt) node[above]{$\mathbf{e_{12}}$} node[below]{$(1)$};
}
\end{tikzpicture}
}
&
\parbox[c]{1.5cm}{
\centering
\begin{tikzpicture}[scale = 0.5]
{
\filldraw (0,0) circle (1pt) node[above]{$e_{123}$} node[below]{$(0)$};
\draw [-] (0,0)--(1,0);
\draw [-] (0,0)--(-1,0);
}
\end{tikzpicture}
}&Yes
\\\hline
\multirow{2}{*}{\vspace{-0.7cm}$\mathfrak{r}_{3,1}$}&\multirow{2}{*}{\vspace{-0.7cm}$e_2$} &\multirow{2}{*}{\vspace{-0.7cm}$e_3$} &\multirow{2}{*}{\vspace{-0.7cm}$0$}&\multirow{2}{*}{\vspace{-0.7cm}\parbox[c]{4.6cm}{\vspace{3pt}$x^2+y^2\neq 0, z \in \mathbb{R}$, $k = 1$,\\ $x = y = 0, z \neq 0$, $k = 2$\vspace{3pt}}}&
$\mathbb{Z}$&
\parbox[c]{2.9cm}{
\centering
\begin{tikzpicture}[scale = 0.5]
{
\filldraw (0,0) circle (1pt) node[above]{$e_{1}$} node[below]{$(0)$};
\draw [-] (0,0)--(3,0);
\filldraw (3,0) circle (1pt) node[above]{$e_{2}, e_{3}$} node[below]{$(1)$};
}
\end{tikzpicture}
}
&
\parbox[c]{2.9cm}{
\centering
\begin{tikzpicture}[scale = 0.5]
{
\filldraw (0,0) circle (1pt) node[above]{$e_{12}, e_{13}$} node[below]{$(1)$};
\draw [-] (0,0)--(3,0);
\filldraw (3,0) circle (1pt) node[above]{$\mathbf{e_{23}}$} node[below]{$(2)$};
}
\end{tikzpicture}
}
&
\parbox[c]{1.5cm}{
\centering
\begin{tikzpicture}[scale = 0.5]
{
\filldraw (0,0) circle (1pt) node[above]{$e_{123}$} node[below]{$(2)$};
\draw [-] (0,0)--(1,0);
\draw [-] (0,0)--(-1,0);
}
\end{tikzpicture}
}&Yes
\\\cline{6-10}
& & &&&
$\mathbb{Z}^2$&
\parbox[c]{2.9cm}{
\centering
\begin{tikzpicture}[scale = 0.5]
{
\filldraw (0,0) circle (1pt) node[left]{$e_{1}$} node[below]{$(0,0)$};
\draw [-] (0,0)--(3,0);
\draw [-] (0,0)--(0,1);
\filldraw (3,0) circle (1pt) node[right]{$e_{2}$} node[below]{$(1,0)$};
\filldraw (0,1) circle (1pt) node[right]{$e_{3}$} node[left]{$(0,1)$};
}
\end{tikzpicture}
}
&
\parbox[c]{2.9cm}{
\centering
\begin{tikzpicture}[scale = 0.5]
{
\filldraw (3,0) circle (1pt) node[right]{${\bf e_{12}}$} node[left]{$(1,0)$};
\filldraw (0,1) circle (1pt) node[above]{${\bf e_{13}}$} node[below]{$(0,1)$};
\draw [-] (0,1)--(3,1);
\draw [-] (3,0)--(3,1);
\filldraw (3,1) circle (1pt) node[above]{${\bf e_{23}}$} node[right]{$(1,1)$};
}
\end{tikzpicture}
}
&
\parbox[c]{1.5cm}{
\centering
\begin{tikzpicture}[scale = 0.5]
{
\filldraw (0,0) circle (1pt) node[above]{$e_{123}$} node[below]{$(1,1)$};
\draw [-] (0,0)--(1,0);
\draw [-] (0,0)--(-1,0);
}
\end{tikzpicture}
}&No
\\\hline
$\mathfrak{r}_{3}$&
$0$ &$-e_1$ &$e_1+e_2$&\parbox[c]{4.6cm}{\vspace{3pt}$x>0, y=0, z = 0$ $k = 1$,\\ $x < 0, y = 0, z = 0$, $k = 2$, \\ $ y \neq  0, z = 0$, $k = 3$\vspace{3pt}}&
$\mathbb{Z}$&
\parbox[c]{2.9cm}{
\centering
\begin{tikzpicture}[scale = 0.5]
{
\filldraw (0,0) circle (1pt) node[above]{$e_{3}$} node[below]{$(0)$};
\draw [-] (0,0)--(3,0);
\filldraw (3,0) circle (1pt) node[above]{$e_{1}, e_{2}$} node[below]{$(1)$};
}
\end{tikzpicture}
}
&
\parbox[c]{2.9cm}{
\centering
\begin{tikzpicture}[scale = 0.5]
{
\filldraw (0,0) circle (1pt) node[above]{$e_{13}, e_{23}$} node[below]{$(1)$};
\draw [-] (0,0)--(3,0);
\filldraw (3,0) circle (1pt) node[above]{$\mathbf{e_{12}}$} node[below]{$(2)$};
}
\end{tikzpicture}
}
&
\parbox[c]{1.5cm}{
\centering
\begin{tikzpicture}[scale = 0.5]
{
\filldraw (0,0) circle (1pt) node[above]{$e_{123}$} node[below]{$(2)$};
\draw [-] (0,0)--(1,0);
\draw [-] (0,0)--(-1,0);
}
\end{tikzpicture}
}&No
\\\hline
$\mathfrak{r}_{3,\lambda}$ &
$0$ &$-e_1$ &$\lambda e_2$ &\parbox[c]{4.6cm}{\vspace{3pt}$y =0, z \neq 0, x \in \mathbb{R}$, $k = 1$, \\ $z = 0, y \neq 0, x \in \mathbb{R}$, $k = 2$, \\ $y = 0, z = 0, x\in \mathbb{R}$, $k = 3$\vspace{3pt}}&
$\mathbb{R}$&
\parbox[c]{2.9cm}{
\centering
\begin{tikzpicture}[scale = 0.5]
{
\draw [-] (0,0)--(-2,0);
\filldraw (-2,0) circle (1pt) node[above]{$e_{3}$} node[below]{$(0)$};
\filldraw (0,0) circle (1pt) node[above]{$e_{1}$} node[below]{$(1)$};
\draw [-] (0,0)--(2,0);
\filldraw (2,0) circle (1pt) node[above]{$e_{2}$} node[below]{$(\lambda)$};
}
\end{tikzpicture}
}
&
\parbox[c]{2.9cm}{
\centering
\begin{tikzpicture}[scale = 0.5]
{
\draw [-] (0,0)--(-2,0);
\filldraw (-2,0) circle (1pt) node[above]{${\bf e_{13}}$} node[below]{$(1)$};
\filldraw (0,0) circle (1pt) node[above]{${\bf e_{23}}$} node[below]{$(\lambda)$};
\draw [-] (0,0)--(2,0);
\filldraw (2,0) circle (1pt) node[above]{${\bf e_{12}}$} node[below]{$(1+\lambda)$};
}
\end{tikzpicture}
}
&
\parbox[c]{1.5cm}{
\centering
\begin{tikzpicture}[scale = 0.5]
{
\draw [-] (0,0)--(-1,0);
\filldraw (0,0) circle (1pt) node[above]{$e_{123}$} node[below]{$(1+\lambda)$};
\draw [-] (0,0)--(1,0);
}
\end{tikzpicture}
}&Yes
\\\hline
$\mathfrak{r}'_{3,\lambda\neq 0}$&
$0$ &$e_2-\lambda e_1$ &$\lambda e_2+e_1$&\parbox[c]{4.6cm}{\vspace{3pt}$x > 0, y = 0, z = 0$, $k = 1$, \\ $x < 0, y = 0, z = 0$, $k = 2$\vspace{3pt}}&
$\mathbb{Z}$&
\parbox[c]{2.9cm}{
\centering
\begin{tikzpicture}[scale = 0.5]
{
\filldraw (0,0) circle (1pt) node[above]{$e_{3}$} node[below]{$(0)$};
\draw [-] (0,0)--(3,0);
\filldraw (3,0) circle (1pt) node[above]{$e_{1}, e_{2}$} node[below]{$(1)$};
}
\end{tikzpicture}
}
&
\parbox[c]{2.9cm}{
\centering
\begin{tikzpicture}[scale = 0.5]
{
\filldraw (0,0) circle (1pt) node[above]{$e_{13}, e_{23}$} node[below]{$(1)$};
\draw [-] (0,0)--(3,0);
\filldraw (3,0) circle (1pt) node[above]{$\mathbf{e_{12}}$} node[below]{$(2)$};
}
\end{tikzpicture}
}
&
\parbox[c]{1.5cm}{
\centering
\begin{tikzpicture}[scale = 0.5]
{
\filldraw (0,0) circle (1pt) node[above]{$ e_{123}$} node[below]{$(2)$};
\draw [-] (0,0)--(1,0);
\draw [-] (0,0)--(-1,0);
}
\end{tikzpicture}
}&No
\\\hline
\end{tabular}
\caption{Commutation relations, non-zero orbits of equivalent $r$-matrices, and $G$-gradations for three-dimensional Lie algebras and induced descompositions on their Grassmann algebras (only nonzero subspaces written). The letters $a,b,c$ stand for arbitrary different values within $\{1,2,3\}$. As standard in this work, we denote $e_{i_1\ldots i_r}:=e_{i_1}\wedge\ldots\wedge e_{i_r}$, with $i_1,\ldots,i_r\in \overline{1,r}$. Solutions of CYBEs obtained due to the structure of gradations and induced decompositions in $\Lambda^2\mathfrak{g}$ and $\Lambda^3\mathfrak{g}$ are written in bold. We assume $\lambda\in(-1,1)$. The spaces $\mathcal{S}_k$ are the equivalence classes of reduced $r$-matrices related to non-zero cocommutators up to ${\rm Aut}(\mathfrak{g})$.}\label{tabela3w}
\end{table}\end{landscape}

\subsection{General properties}
Let us prove a few results concerning the characterisation of the subspaces $(\Lambda^m\mathfrak{g})^\mathfrak{g}$ and ${\rm Aut}(\mathfrak{g})$. 

\begin{proposition}\label{Aut3} 
Let $\mathfrak{g}$ be such that $\kappa_{\mathfrak{g}}\neq 0$ and $\mathfrak{g}_{1)}\subset \ker \kappa_{\mathfrak{g}}$ is a two-dimensional abelian Lie subalgebra. If $v\notin \mathfrak{g}_{1)}$, then every $T \in {\rm Aut}(\mathfrak{g})$ leaves invariant the set of eigenvectors of ${\rm ad}_v|_{\mathfrak{g}_{1)}}$.
\end{proposition}
\begin{proof} 
Let us prove that if $T\in {\rm Aut}(\mathfrak{g})$, then $Tv\in v+\mathfrak{g}_{1)}$ or $Tv\in -v+\mathfrak{g}_{1)}$ for every $v\notin \mathfrak{g}_{1)}$. Since $T\in {\rm Aut}(\mathfrak{g})$, one has that $\kappa_{\mathfrak{g}}(Tv,Tv)=\kappa_{\mathfrak{g}}(v,v)$. Since $v$ and $\mathfrak{g}_{1)}$ generate $\mathfrak{g}$ and $T$ is an injection, we can write $Tv=\lambda v+h$ for an $h\in \mathfrak{g}_{1)}$ and $\lambda\in \mathbb{K}\backslash\{0\}$. As  $\mathfrak{g}_{1)}\subset \ker\kappa_\mathfrak{g}$, then $\kappa_\mathfrak{g}(Tv,Tv)=\lambda^2\kappa_\mathfrak{g}(v,v)$ for every $v\in \mathfrak{g}$. Since $\kappa_\mathfrak{g}\neq 0$ , one has that $\lambda\in \{\pm 1\}$. Hence, $Tv\in v+\mathfrak{g}_{1)}$ or $Tv\in -v+\mathfrak{g}_{1)}$. Since $\dim\mathfrak{g}_{1)}=2$, $\mathfrak{g}_{1)}$ is an abelian ideal of $\mathfrak{g}$ invariant under automorphisms of $\mathfrak{g}$, one obtains that 
$$
{\rm ad}_{v}|_{\mathfrak{g}_{1)}}=\pm{\rm ad}_{Tv}|_{\mathfrak{g}_{1)}}\Longrightarrow 
{\rm ad}_{v}|_{\mathfrak{g}_{1)}}=\pm T|_{\mathfrak{g}_{1)}}\circ{\rm ad}_{v}|_{\mathfrak{g}_{1)}}\circ T^{-1}|_{\mathfrak{g}_{1)}}.
$$
In consequence, if $e$ is an eigenvector of ${\rm ad}_v|_{\mathfrak{g}_{1)}}$, then $Te$ is a new eigenvector of ${\rm ad}_v|_{\mathfrak{g}_{1)}}$.  
\end{proof}

Proposition \ref{Aut3} can be modified to give a very accurate form of $T|_{\mathfrak{g}^{(1)}}$. For instance, if ${\rm ad}_v|_{\mathfrak{g}^{1)}}$ has two eigenvectors $e_1,e_2$ with different eigenvalues $\lambda_1,\lambda_2$ satisfying that $\lambda_1 +\lambda_2 = 0$ and $Tv \in -v + \mathfrak{g}_{1)}$, then $T|_{\mathfrak{g}^{1)}}$ is an anti-diagonal matrix in the basis $\{e_1,e_2\}$. If $\lambda_1+\lambda_2 \neq 0$ and $\lambda_1\neq\lambda_2$, then $Tv \in v + \mathfrak{g}_{1)}$ and $T|_{\mathfrak{g}^{1)}}$ is diagonal. Several variations of this reasoning can be applied, e.g. when ${\rm ad}_v|_{\mathfrak{g}^{1)}}$ is triangular.

\begin{proposition}\label{Aut2} 
Let $\Omega\in (\Lambda^3\mathfrak{g}^*)\backslash \{0\}$ and assume that $\Upsilon:\Lambda^2\mathfrak{g}\ni r\mapsto \Omega ([r,r]_S)\in \mathbb{R}$ is a semi-definite function different than zero. Then, every automorphism of $\mathfrak{g}$ has positive determinant.
\end{proposition}
\begin{proof} Since $\mathfrak{g}$ is three-dimensional, $\Omega$ is a basis of $\Lambda^3\mathfrak{g}^*$ and there exists a one-element dual basis $\theta\in \Lambda^3\mathfrak{g}$. As $\Omega([r,r]_S)=\Upsilon(r)$, then $[r,r]_S= \Upsilon(r)\theta$. Since $\Upsilon$ is not identically zero, there exists an $r\in \Lambda^2\mathfrak{g}$ such that
$[r,r]_S=\Upsilon(r)\theta\neq 0$. If $T\in {\rm Aut}(\mathfrak{g})$, then
$$
\Upsilon(r) \det (T)\,\theta=\det (T)\, [r, r]_S = \Lambda^3T[r,r]_S=[\Lambda^2Tr,\Lambda^2Tr]_S=\Upsilon(\Lambda^2Tr)\theta.
$$
Hence, $\Upsilon(r)\det (T)=\Upsilon(\Lambda^2Tr)\neq 0$. The semi-definiteness of $\Upsilon$ yields that both sides of the equality must have the same sign and then $\det (T)>0$.
\end{proof}

In the following subsections, we assume that every $\mathfrak{g}$ has a basis $\{e_1,e_2,e_3\}$ satisfying the corresponding commutation relations given in Table \ref{tabela3w}. We choose also the induced bases $\{e_{12},e_{13},e_{23}\}$ and $\{e_{123}\}$ in $\Lambda^2\mathfrak{sl}_2$ and $\Lambda^3\mathfrak{sl}_{2}$, respectively. For each $\mathfrak{g}$, we first analyse $\mathfrak{g}$-invariant elements through gradations, which allows us to determine the shape of mCYBE and reduced $r$-matrices. 

Recall that if $\mathfrak{g}$ admits a $G$-gradation and  $G$ is a group (which happens for all gradations of three-dimensional Lie algebras in this work (cf. Table \ref{tabela3w}),  then $(\Lambda^2\mathfrak{g})^\mathfrak{g}$ is the direct sum of the subspaces of $\mathfrak{g}$-invariant within each homogeneous subspace of $\Lambda^2\mathfrak{g}$. This simplifies the search of $(\Lambda^2\mathfrak{g})^{\mathfrak{g}}$. Meanwhile, Proposition \ref{Prop:Sym} simplifies the derivation of $(\Lambda^3\mathfrak{g})^\mathfrak{g}$ for three-dimensional Lie algebras

\begin{proposition}\label{Prop:Sym}
Each $G$-graded three-dimensional Lie algebra $\mathfrak{g}$ has a unique homogeneous subspace  $(\Lambda^3\mathfrak{g})^{(\alpha)}\neq 0$. Moreover, $[\mathfrak{g}^{(\beta)},(\Lambda^3\mathfrak{g})^{(\alpha)}]=0$ for $\beta\neq 0$. If $\mathfrak{g}$ has a root gradation, then $\Lambda^3\mathfrak{g}=(\Lambda^3\mathfrak{g})^\mathfrak{g}$.\end{proposition}

\subsection{Lie bialgebras on semi-simple Lie algebras} 

There exists only two semi-simple three-dimensional Lie algebras: $\mathfrak{su}_2$ and $\mathfrak{sl}_2$ \cite{SW14}. 

$\bullet$ {\bf Lie bialgebras on \texorpdfstring{$\mathfrak{sl}_2$}{}}

 Since $\mathfrak{sl}_2$ admits a root decomposition giving rise to a root $\mathbb{Z}$-gradation, Proposition \ref{Prop:Sym} yields that $\Lambda^3\mathfrak{sl}_2=(\Lambda^3\mathfrak{sl}_2)^{\mathfrak{sl}_2}$ and every element of $\Lambda^2\mathfrak{sl}_2$ satisfies the mCYBE. Since $\mathfrak{sl}_2$ has a root gradation, Proposition \ref{gin0} gives that  $(\Lambda^2\mathfrak{sl}_2)^{\mathfrak{sl}_2}\subset (\Lambda^2\mathfrak{sl}_2)^{(0)}$. It is then immediate that $(\Lambda^2\mathfrak{sl}_2)^{\mathfrak{sl}_2}=0$. By Proposition \ref{prop:requiv}, every $r\in \Lambda^2\mathfrak{sl}_2$ induces a different cocommutator $\delta_r(\cdot):=[\cdot, r]_{S}$.

A simple calculation and Proposition \ref{proporb} ensure that the dimension of any orbit, $\mathcal{O}_w$, of the action of ${\rm Inn}(\mathfrak{sl}_2)$ on $\Lambda^2\mathfrak{sl}_2$ is  $\dim \Theta_w^2=2$ for $w\in \Lambda^2\mathfrak{sl}_2\backslash\{0\}$ and $\dim \Theta_w^2=0$ otherwise. 
Since ${\rm Inn}(\mathfrak{sl}_{2})$ is connected, the $\mathcal{O}_w$ are two- or zero-dimensional connected immersed submanifolds. Each $\mathcal{O}_w$ must be contained in a connected submanifold of a level set, $S_k$, of the quadratic function $f_{\Lambda^2\mathfrak{sl}_2}:r\in \Lambda^2\mathfrak{sl}_2\mapsto \kappa_{\Lambda^2\mathfrak{sl}_{2}}(r,r)\in \mathbb{R}$. 
If $r=x e_{12}+ y e_{13} + z e_{23}$, then  $f_{\Lambda^2\mathfrak{sl}_2}(r):=8xy - 4z^2$ and $f_{\Lambda^2\mathfrak{sl}_2}$ admits three types of  $S_k$ according to the sign of $k$. If $k<0$, then $S_k$ is a one-sheeted hyperboloid; $S_0$ consists of two cones, one opposite to the other, and the origin of $\Lambda^2\mathfrak{sl}_2$; meanwhile $S_k$ for $k>0$ is a two-sheeted hyperboloid  with two parts contained within the region $x>0,y>0$ and $x<0,y<0$, respectively (see Figure \ref{Sum}).

Each $S_k$ is the union of different orbits $\mathcal{O}_w$, which are two-dimensional except for $\mathcal{O}_0$. Then, each $S_k$, for $k\neq 0$, is an orbit $\mathcal{O}_w$ while $S_0$ has three orbits given by two cones for points with $z>0$ or $z<0$, and $(0,0,0)$. Consequently, there are five inequivalent classes of $r$-matrices on $\mathfrak{sl}_{2}$ relative to the action of ${\rm Inn}(\mathfrak{sl}_{2})$ (cf. \cite{Farinati}). The representatives of each class are $r_0 =0, r=ae_{23}$, with $a>0$ (one-sheeted hyperboloids), $r = a (e_{12}+e_{13})$, with $a\in \mathbb{R}\backslash\{0\}$, (two-sheeted hyperboloids),  and $r = \pm e_{12}$ (cones). 

Now the orbits of the action of ${\rm Aut}(\mathfrak{sl}_2)$ on $\Lambda^2\mathfrak{sl}_2$ can easily be derived. Derivations of $\mathfrak{sl}_2$ are of the form ${\rm ad}_v$ for a certain $v\in \mathfrak{sl}_2$ \cite{Jacobson}. In view of Proposition \ref{aut},  $\mathfrak{inn}(\mathfrak{sl}_2)=\mathfrak{der}(\mathfrak{sl}_2)=\mathfrak{aut}(\mathfrak{sl}_2)$. Hence, ${\rm Inn}(\mathfrak{sl}_2)={\rm Aut}_c(\mathfrak{sl}_2)$ and each orbit of the action of ${\rm Aut}(\mathfrak{sl}_2)$ on $\Lambda^2\mathfrak{sl}_2$ is the sum of some $\mathcal{O}_w$. As $\kappa_{\mathfrak{sl}_2}$ is invariant under the action of ${\rm Aut}(\mathfrak{sl}_2)$, i.e. it is $GL(\widehat{\rm ad})$-invariant, Corollary \ref{ExtInvRep} yields that  $\kappa_{\Lambda^2\mathfrak{sl}_2}$ is invariant under the action of ${\rm Aut}(\mathfrak{sl}_2)$ on $\Lambda^2\mathfrak{sl}_2$ and each of its orbits must be contained in a $S_k$. 

Then, the $T\in {\rm Aut}(\mathfrak{sl}_2)$ such that $T(e_1) := e_1, T(e_2) := -e_2, T(e_3) := -e_3$ can be extended to $\Lambda^2T$ giving rise to a map such that $\Lambda^2T(e_{12})=- e_{12},\Lambda^2T(e_{13})=- e_{13}, \Lambda^2T(e_{23})=e_{23},
$
which connects the two-sheeted hyperboloids within $S_k$ for each fixed $k> 0$. It also maps the two cones contained in $S_0$. Therefore, we have three types of non-zero $r$-matrices up to the action of ${\rm Aut}(\mathfrak{sl}_2)$. It is worth noting that all of them are solutions of the CYBE that can be almost fully derived via gradations as seen in Table \ref{tabela3w}.

Our result agrees with the findings given in \cite{Gomez}, but they do not match the work \cite{Farinati}. This is due to the fact that Farinati and coworkers assume that ${\rm Inn}(\mathfrak{sl}_2)={\rm Aut}(\mathfrak{sl}_2)$ (see \cite[p. 56]{Farinati}), which was proved to be wrong in Section \ref{sect:auto_alg}.

$\bullet$ {\bf Lie bialgebras on \texorpdfstring{$\mathfrak{su}_{2}$}{}}

The $\mathbb{Z}_2$-gradations of $\mathfrak{su}_2$ and their associated decompositions for $\Lambda^3\mathfrak{su}_2$  (see Table \ref{tabela3w} and Example \ref{Gradsu2}) show that the unique non-zero homogeneous space in $\Lambda^3\mathfrak{su}_2$ is invariant under $e_b,e_c$. Since such a homogeneous space is the same for each $\mathbb{Z}_2$-gradation but $e_b,e_c$ are arbitrary, $\Lambda^3\mathfrak{su}_3=(\Lambda^3\mathfrak{su}_2)^{\mathfrak{su}_2}$ and every $r$-matrix is a solution to the mCYBE. By Lemma \ref{LemgDe}, the space $(\Lambda^2\mathfrak{su}_2)^{\mathfrak{su}_2}$ is the linear combination of $\mathfrak{su}_2$-invariant elements within homogeneous spaces of $\Lambda^2\mathfrak{su}_2$. It is then simple to see that $(\Lambda^2\mathfrak{su}_2)^{\mathfrak{su}_2}=\{0\}$. Hence, every $r\in \Lambda^2\mathfrak{su}_2$ induces a different cocomutator and the classification of coboundary cocomutators of $\mathfrak{su}_2$  up to ${\rm Aut}(\mathfrak{su}_2)$ amounts to the classification of their corresponding $r$-matrices.

Let us study the equivalence of $r$-matrices under inner automorphisms by using $\mathfrak{su}_2$-invariant metrics on $\mathfrak{su}_2, \Lambda^2 \mathfrak{su}_2$, and $\Lambda^3 \mathfrak{su}_2$. The Killing metric of $\mathfrak{su}_2$ and its extensions are given by the matrices
$$
[\kappa_{\mathfrak{su}_{2}}] = -2\ \mathbb{I}_{3 \times 3}, \qquad
[\kappa_{\Lambda^2 \mathfrak{su}_{2}}] = 4\ \mathbb{I}_{3 \times 3},\qquad
[\kappa_{\Lambda^3\mathfrak{su}_{2}}]=-8\ \mathbb{I}_{1\times 1}
$$
in our standard bases. Due to Proposition \ref{proporb} and since  ${\rm Inn}(\mathfrak{su}_{2})$ is connected, the orbits of the action of ${\rm Inn}(\mathfrak{su}_{2})$ on $\Lambda^2\mathfrak{su}_{2}$ have a dimension given by 
${\rm Im}\,\Theta^2_w$: two for $w\in \Lambda^2\mathfrak{su}_2\backslash\{0\}$ and zero otherwise.  

The orbits of the action of ${\rm Inn}(\mathfrak{su}_2)$ on $\Lambda^2\mathfrak{su}_2$ are connected immersed submanifolds contained in the level sets, $S_k$, where the quadratic function $f_{\Lambda^2\mathfrak{su}_2}(r):=\kappa_{\Lambda^2\mathfrak{su}_{2}}(r, r)\!\!=4(x^2 + y^2 + z^2)$ takes value $k$. Since the orbits of ${\rm Inn}(\mathfrak{su}_2)$ must be open relative to the topology of each $S_k$ (with $k\geq 0$), which are connected, each orbit of ${\rm Inn}(\mathfrak{su}_2)$ must be the whole $S_k$ for each $k\geq 0$. 
Hence, non-equivalent $r \in \Lambda^2 \mathfrak{su}_2$, with respect to the action of ${\rm Inn}(\mathfrak{su}_2)$, are given by elements $r$ with different modulus, e.g. $r_a=a e_{12}$, with $a\geq 0$. Since the orbits of the action of ${\rm Aut}(\mathfrak{su}_2)$ on $\Lambda^2\mathfrak{su}_2$ are given by the sum of orbits of ${\rm Inn}(\mathfrak{su}_2)$ and they are contained in the surfaces $S_k$, the orbits of the action of ${\rm Aut}(\mathfrak{su}_2)$ in $\Lambda^2\mathfrak{su}_2$ are indeed the spheres $S_k$ with $k>0$ and the point $k=0$ (see Figure \ref{Sum}).

The above retrieves the results given in \cite{Farinati,Gomez} in a new and geometric fashion. 

\subsection{Nilpotent Lie algebras: The 3D-Heisenberg Lie algebra}\label{sect:heisen}
Let us consider the three-dimensional (3D) {\it Heisenberg algebra} $\mathfrak{h}$ \cite{Farinati} described in Table \ref{tabela3w}. This is the only, up to a Lie algebra isomorphism, three-dimensional nilpotent Lie algebra \cite{SW14}.

In view of Proposition \ref{Prop:Sym} and the fact that non-zero homogeneous spaces in $\mathfrak{h}$ are related to non-zero elements, one has that  $(\Lambda^3\mathfrak{h})^{\mathfrak{h}}=\Lambda^3\mathfrak{h}$ and every element of $\Lambda^2\mathfrak{h}$ is a solution to the mCYBE. Since $\mathfrak{g}$ admits a $\mathbb{Z}$-gradation, $(\Lambda^2\mathfrak{h})^\mathfrak{h}$ is the sum of $\mathfrak{h}$-invariant elements on each homogeneous space of $\Lambda^2\mathfrak{g}$, which is easily computable. This gives that $(\Lambda^2\mathfrak{h})^\mathfrak{h}=\langle e_{13},e_{12}\rangle$.  Figure \ref{Rh} depicts the equivalence classes of $\Lambda^2_R\mathfrak{h}$ in $\Lambda^2\mathfrak{h}$.

\begin{minipage}{0.3\textwidth}
	\begin{center}
		\includegraphics[scale=0.25]{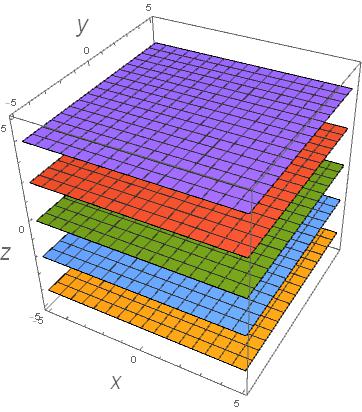}
		\captionof{figure}{Representative orbits of the action of ${\rm Inn}(\mathfrak{h})$ on $\Lambda^2 \mathfrak{h}$}\label{Rh}
	\end{center}
\end{minipage}$\quad$
\begin{minipage}{0.64\textwidth}
	Proposition \ref{prop:requiv} ensures that elements belonging to the same class of $\Lambda^2_R\mathfrak{h}$ give rise to the same Lie bialgebra. Then, to classify coboundary Lie bialgebras, one can restrict oneself to studying reduced $r$-matrices in $\Lambda^2_R\mathfrak{h}$. However, $\mathfrak{h}$ admits the automorphisms $T_\alpha$, with $\alpha \in \mathbb{R}\backslash\{0\}$, given by 
	$
	T_\alpha(e_1):=\alpha e_1, T(e_2):=e_2, T(e_3):= \alpha e_3.
	$ 
	Therefore, $\Lambda^2 T_\alpha (e_{12}) = \alpha e_{12}$ for any $\alpha\neq 0$. Since $(\Lambda^2\mathfrak{h})^\mathfrak{h}$ is invariant under the action of ${\rm Aut}(\mathfrak{h})$, it makes sense to consider the induced action of $\Lambda^2T$ on $\Lambda^2_R\mathfrak{h}$. Then, the induced action of ${\rm Aut}(\mathfrak{h})$ on $\Lambda^2_R\mathfrak{h}$ has two orbits $[0]$ and $[e_{12}]$. Thus, we have only one class of non-zero coboundary coproducts can be represented by the $r$-matrix $r:= e_{12}$. The space of non-equivalent $r$-matrices  is depicted in Figure \ref{Sum}.	Note that the gradation of $\mathfrak{h}$ and the induced decompositions in $\Lambda^2\mathfrak{h}$ and $\Lambda^3\mathfrak{h}$ give easily that $e_{12}$ is a solution to the mCYBE on $\mathfrak{h}$.
\end{minipage}

\subsection{Solvable non-nilpotent Lie algebras}

There exist six classes of solvable but not nilpotent three-dimensional real Lie algebras \cite{SW14}. The following subsections aim at classifying all Lie bialgebras on them.

\subsubsection{The Lie algebra \texorpdfstring{$\mathfrak{r}'_{3,0}$}{}}

Let us  analyse $(\Lambda^2\mathfrak{r}'_{3,0})^{\mathfrak{r}'_{3,0}}$ and $(\Lambda^3\mathfrak{r}'_{3,0})^{\mathfrak{r}'_{3,0}}$. In view of Proposition \ref{Prop:Sym}, the only non-zero homogeneous subspace of $(\Lambda^3\mathfrak{r}'_{3,0})^{\mathfrak{r}'_{3,0}}$ is invariant relative to $e_2,e_3$, which have non-zero degree. The invariance of $\Lambda^3\mathfrak{r}'_{3,0}$ relative to $e_1$ is immediate. Then $(\Lambda^3\mathfrak{r}'_{3,0})^{\mathfrak{r}'_{3,0}}=\Lambda^3\mathfrak{r}'_{3,0}$ and all elements of $\Lambda^2\mathfrak{r}'_{3,0}$ are $r$-matrices.

Recall that $(\Lambda^2\mathfrak{r}'_{3,0})^{\mathfrak{r}'_{3,0}}$ is the sum of homogeneous $\mathfrak{r}'_{3,0}$-invariant elements in $\Lambda^2\mathfrak{r}'_{3,0}$. It easily follows by using the gradations in $\mathfrak{r}'_{3,0}$ that $\langle e_{23}\rangle\subset (\Lambda^2\mathfrak{r}'_{3,0})^{(2)}$ is $\mathfrak{r}'_{3,0}$-invariant. To obtain the $\mathfrak{r}'_{3,0}$-invariant elements within $\Lambda^2(\mathfrak{r}'_{3,0})^{(1)}$, we consider an arbitrary element $e_1\wedge \lambda(e_2,e_3)$ of the space, where $\lambda(e_2,e_3)$ stands for a linear combination of $e_2$ and $e_3$. Then,
$$
[e_2,e_1\wedge \lambda(e_2,e_3)]_S=-e_3\wedge \lambda(e_2,e_3)=0,\qquad [e_3,e_1\wedge \lambda(e_2,e_3)]_S=e_2\wedge \lambda(e_2,e_3)=0.
$$
Hence, $\lambda(e_2,e_3)=0$ and $(\Lambda^2\mathfrak{r}'_{3,0})^{\mathfrak{r}'_{3,0}}=\langle e_{23}\rangle$. 

Let us classify the non-equivalent (up to inner Lie algebra automorphisms of $\mathfrak{r}'_{3,0}$) coboundary coproducts on $\mathfrak{r}'_{3,0}$ by using $\mathfrak{r}'_{3,0}$-invariant metrics on $\Lambda^2_R\mathfrak{r}'_{3,0}$. Let us discuss the existence of $\mathfrak{r}'_{3,0}$-invariant metrics on $\Lambda^2_R \mathfrak{r}'_{3,0}$. Consider the basis $\{[e_{12}], [e_{13}]\}$ of $\Lambda^2_R \mathfrak{r}'_{3,0}$. If $\Lambda^2_R{\rm ad}: v\in \mathfrak{r}'_{3,0} \mapsto [[v],\cdot]_R\in \mathfrak{gl}(\Lambda^2_R \mathfrak{r}'_{3,0})$, where $[\cdot,\cdot]_R$ is the bracket on $\Lambda^2_R \mathfrak{r}'_{3,0}$ induced by the algebraic bracket on $\Lambda^2 \mathfrak{r}'_{3,0}$, then
$$
{\rm Im}\,\Lambda^2_R{{\rm ad}}_{e_1}=\langle [e_{13}],[e_{12}]\rangle,\,\,\ker\,\Lambda^2_R{{\rm ad}}_{e_1}=\langle [0]\rangle,\,\, {\rm Im}\,\Lambda^2_R{{\rm ad}}_{e_2}=\langle [0]\rangle,\,\, \ker\,\Lambda^2_R{{\rm ad}}_{e_2}=\langle [e_{13}]\rangle,$$ 
$$
{\rm Im}\,\Lambda^2_R{{\rm ad}}_{e_3}=\langle [e_{23}]\rangle =\langle [0]\rangle,\qquad \ker \,\Lambda^2_R{{\rm ad}}_{e_3}=\langle [e_{12}]\rangle,
$$
\vskip -0.5cm
\begin{equation*}
\begin{gathered}
b^R_{\Lambda^2\mathfrak{r}'_{3,0}}([e_{12}], [e_{12}]) = b^R_{\Lambda^2\mathfrak{r}'_{3,0}}([[e_1], [e_{13}]]_R, [e_{12}]) = -b^R_{\Lambda^2\mathfrak{r}'_{3,0}}([e_{13}], [[e_1], [e_{12}]]_R)=
b^R_{\Lambda^2\mathfrak{r}'_{3,0}}([e_{13}], [e_{13}]),\\
b^R_{\Lambda^2\mathfrak{r}'_{3,0}}([e_{13}],[e_{12}]) = -b^R_{\Lambda^2\mathfrak{r}'_{3,0}}([[e_1], [e_{12}]]_R, [e_{12}]) = b^R_{\Lambda^2\mathfrak{r}'_{3,0}}([e_{12}],[[e_1],[e_{12}]]_R)=-b^R_{\Lambda^2\mathfrak{r}'_{3,0}}([e_{12}],[e_{13}]).
\end{gathered}
\end{equation*}

Therefore, using again Propositions \ref{prop:sym_form}--\ref{prop:form_cond}, we obtain that $b^R_{\Lambda^2\mathfrak{r}'_{3,0}}$  must be of the form
$$
[b_{\Lambda^2 \mathfrak{r}'_{3,0}}^R] = \left(\begin{array}{cc}
a_1&0\\
0&a_1
\end{array}\right),\qquad \forall a_1\in\mathbb{R}.
$$
Indeed, this is a $\mathfrak{r}'_{3,0}$-invariant metric. For simplicity, we here after assume that $a_1=1$.

\noindent\begin{minipage}{0.35\textwidth}
	\begin{center}
		\includegraphics[scale=0.30]{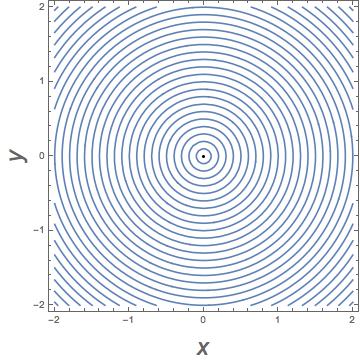}
		\captionof{figure}{Representative orbits of the action of ${\rm Inn}(\mathfrak{r}'_{3,0})$ on $\Lambda^2_R \mathfrak{r}'_{3,0}$}\label{FigSol}
	\end{center}
\end{minipage}$\qquad$
\begin{minipage}{0.60\textwidth} 
	Since all elements of $\Lambda^2\mathfrak{r}'_{3,0}$ give rise to  coboundary cocommutators, their study can be reduced to studying the equivalence classes of $\Lambda^2_R\mathfrak{r}'_{3,0}$. Let us study the equivalence of reduced $r$-matrices up to inner automorphisms of $\mathfrak{r}_{3,0}'$. The equivalence classes in $\Lambda^2_R\mathfrak{r}_{3,0}'$ can be written as $x [e_{12}] + y [e_{13}]$ in the basis $\{[e_{12}],[e_{13}]\}$. It follows that $b^R_{\Lambda^2 \mathfrak{r}'_{3,0}} ([r], [r]) = x^2 + y^2$. The image of $\Theta^2_r$ is one-dimensional for $x^2 +y^2 \neq 0$ and zero-dimensional otherwise. In consequence, the orbits in $\Lambda_R^2\mathfrak{r}_{3,0}'$ relative to the action of ${\rm Aut}(\mathfrak{r}'_{3,0})$ are circles and the central point. Therefore, there exists a nontrivial family of $r$-matrices $r = \mu e_{12}$, $\mu \in \mathbb{R}^+$, giving rise to different non-zero cocommutators, which are not equivalent up to elements  of ${\rm Inn}(\mathfrak{r}_{3,0}')$. 
\end{minipage}

Let us classify coboundary Lie bialgebras on $\mathfrak{r}_{3,0}'$ up to its  Lie algebra automorphisms. Consider the 
automorphisms $T_\alpha\in {\rm Aut}(\mathfrak{r}_{3,0}')$, with $\alpha \in \mathbb{R}\backslash\{0\}$, satisfying  $T_\alpha(e_1):=e_1, T_\alpha(e_2):=\alpha e_2, T_\alpha(e_3):=\alpha e_3.$
These automorphisms induce elements $\Lambda^2T_\alpha\in GL(\Lambda^2\mathfrak{r}_{3,0}')$ such that $\Lambda^2T_\alpha(e_{12})=\alpha e_{12}$. In turn, these automorphisms induce automorphisms $\Lambda^2_RT_\alpha=\alpha{\rm Id}_{\Lambda^2_R\mathfrak{r}'_{3,0}}$ on $\Lambda^2_R\mathfrak{r}'_{3,0}$. The $\Lambda^2_RT_\alpha$ map the circles with different positive radius among themselves. Hence, their sum forms the only orbit of ${\rm Aut}(\mathfrak{r}'_{3,0})$ on $\Lambda^2_R\mathfrak{r}'_{3,0}$ related to a non-zero coboundary coproduct. Hence, there is only one non-zero coboundary coproduct, up to the action of ${\rm Aut}(\mathfrak{r}_{3,0}')$, induced by an $r$-matrix $r = e_{12}$. Figure \ref{Sum} represents the orbits of the action of ${\rm Aut}(\mathfrak{r}_{3,0}')$ on $\Lambda^2_R\mathfrak{r}_{3,0}'$. This matches the results in \cite{Farinati}. Note that the gradation of $\mathfrak{r}'_{3,0}$ easily shows that $e_{12}$ is a solution to the mCYBE.

\subsubsection{The Lie algebra \texorpdfstring{$\mathfrak{r}_{3,-1}$}{}}\label{sect:r3-1}

Since $\mathfrak{r}_{3,-1}$ admits a root gradation (see Table \ref{tabela3w}), Proposition \ref{Prop:Sym} shows that $\Lambda^3\mathfrak{r}_{3,-1}=(\Lambda^3\mathfrak{r}_{3,-1})^{\mathfrak{r}_{3,-1}}$. The root decomposition of  $\mathfrak{r}_{3,-1}$ implies that $(\Lambda^2\mathfrak{r}_{3,-1})^{\mathfrak{r}_{3,-1}}\subset (\Lambda^2\mathfrak{r}_{3,-1})^{(0)}$. It is then immediate that $(\Lambda^2\mathfrak{r}_{3,-1})^{\mathfrak{r}_{3,-1}}=\langle e_{23}\rangle$ and $\Lambda^2_R\mathfrak{r}_{3,-1}=\langle [e_{13}],[e_{12}]\rangle$. 

Let us classify cocommutators on $\mathfrak{r}_{3,-1}$ via  $\mathfrak{r}_{3,-1}$-invariant metrics, $b^R_{\Lambda^2\mathfrak{r}_{3,-1}}$, on $\Lambda^2_R\mathfrak{r}_{3,-1}$. 
Define $\Lambda^2_R{\rm ad}:v\in \mathfrak{r}_{3,-1}\mapsto [[v],\cdot]_R\in \mathfrak{gl}(\Lambda^2_R\mathfrak{r}_{3,-1})$. In the basis $\{[e_{12}], [e_{13}]\}$ of $\Lambda^2_R\mathfrak{r}_{3,-1}$, one gets
$$
\begin{gathered}
b^R_{\Lambda^2\mathfrak{r}_{3,-1}}([[e_1],[e_{12}]]_R,[e_{12}])=b^R_{\Lambda^2\mathfrak{r}_{3,-1}}([e_{12}],[e_{12}]),\,\, b^R_{\Lambda^2\mathfrak{r}_{3,-1}}([[e_1],[e_{13}]]_R,e_{13})=-b^R_{\Lambda^2\mathfrak{r}_{3,-1}}([e_{13}],[e_{13}]).
\end{gathered}
$$
\noindent\begin{minipage}{0.35\textwidth}
	\begin{center}
		\includegraphics[scale=0.270]{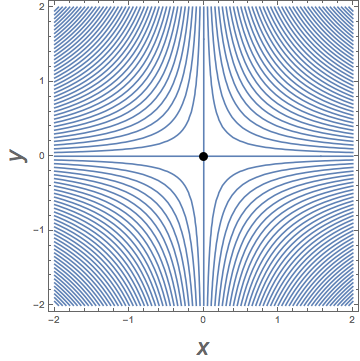}
		\captionof{figure}{Representative orbits of the action of ${\rm Inn}(\mathfrak{r}_{3,-1})$ on $\Lambda^2_R \mathfrak{r}_{3,-1}$}\label{14}
	\end{center}
\end{minipage}$\quad$
\begin{minipage}{0.65\textwidth}
	Then,  $b^R_{\Lambda^2\mathfrak{r}_{3,-1}}$ must be of the form
	\vskip-0.3cm
	$$
	[b^R_{\Lambda^2\mathfrak{r}_{3,-1}}] \!=\! \left(\begin{array}{cc}
	0&\beta\\
	\beta&0
	\end{array}\right), \qquad \beta\in \mathbb{R}.
	$$
	A short calculation shows that
	these are the $\mathfrak{r}_{3,-1}$-invariant metrics on $\Lambda^2_R\mathfrak{r}_{3,-1}$. Let $\{x,y\}$ be the coordinates associated with the basis $\{[e_{12}],[e_{13}]\}$ of $\Lambda^2_R\mathfrak{r}_{3,-1}$. Then $r_R = x [e_{12}] + y[ e_{13}]$ and the quadratic function related to $b^R_{\Lambda^2\mathfrak{r}_{3,-1}}$ reads  $f^R_{\Lambda^2 \mathfrak{r}_{3,-1}}(r_R) = 2xy$.
	The image of $\Theta^2_r$ is one-dimensional for $x^2 + y^2 \neq 0$ and zero-dimensional otherwise. 
	Hence, the orbits of the action of ${\rm Inn}(\mathfrak{r}_{3,-1})$ on $\Lambda^2_R\mathfrak{r}_{3,-1}$ have the representative form presented in Figure \ref{14}. 
\end{minipage}$\quad$

The representatives of inequivalent reduced $r$-matrices, up to the action of inner Lie algebra automorphisms, are given by: 
$$
r^{(\pm, \pm)} =a (\pm[e_{12}] \pm [e_{13}]),\quad  r_2^{(\pm)} = \pm b [e_{12}],\quad 
r_3^{(\pm)} = \pm b [e_{13}],\qquad r_0 = [e_{23}],\quad \forall a, b > 0.
$$
The Lie algebra $\mathfrak{r}_{3,-1}$ satisfies the conditions given in Proposition \ref{Aut3}. Hence, all automorphisms of $\mathfrak{r}_{3,-1}$ must match one of the following automorphisms
\begin{equation*}
\begin{gathered}
T_{\alpha,\beta}(e_1):=e_1+v,\qquad T_{\alpha,\beta}(e_2):=\alpha e_2,\qquad T_{\alpha,\beta}(e_3):=\beta e_3,\qquad \forall \alpha,\beta\in \mathbb{R}\backslash\{0\},\\
T_{\alpha,\beta}'(e_1):=-e_1+v, \qquad T_{\alpha,\beta}'(e_2):=\alpha e_3, \qquad T_{\alpha,\beta}'(e_3):=\beta e_2,\qquad \forall \alpha,\beta\in \mathbb{R}\backslash\{0\},
\end{gathered}
\end{equation*}
for certain $v\in \langle e_2,e_3\rangle$. The extensions $\Lambda^2T_{\alpha,\beta}$ and $\Lambda^2T'_{\alpha,\beta}$ can be restricted to $\Lambda_R^2\mathfrak{r}_{3,-1}$ giving rise to
\begin{equation*}
\begin{gathered}
\Lambda_R^2T_{\alpha,\beta}([e_{12}])=\alpha [e_{12}],\qquad \Lambda^2T_{\alpha,\beta}([e_{13}])=\beta [e_{13}],\qquad \forall \alpha,\beta\in \mathbb{R}\backslash\{0\},\\
\Lambda_R^2T'_{\alpha,\beta}([e_{12}])=-\alpha [e_{13}],\qquad \Lambda_R^2T'_{\alpha,\beta}([e_{13}])=-\beta [e_{12}],\qquad \forall \alpha,\beta\in \mathbb{R}\backslash\{0\}.
\end{gathered}
\end{equation*}

The above transformations do not preserve the connected components of the regions $S_k$ where the function  $f^R_{\Lambda^2\mathfrak{r}_{3,-1}}(r)$ takes a constant value equal to $k$. As a consequence, $T_{\alpha,\beta}$ and $T'_{\alpha,\beta}$ are not inner automorphisms and the non-equivalent non-zero coboundary coproducts on $\mathfrak{r}_{3,-1}$, relative to the action of ${\rm Aut}(\mathfrak{r}_{3,-1})$, are induced by the $r$-matrices: $r = e_{12}, r' = e_{12} - e_{13}$. Indeed, recall that $r_0=e_{23}$ gives rise to a zero coboundary coproduct. Figure \ref{Sum} depicts the orbits of the action of ${\rm Aut}(\mathfrak{r}_{3,-1})$ on $\Lambda^2_R\mathfrak{r}_{3,-1}$.

\subsubsection{The Lie algebra \texorpdfstring{$\mathfrak{r}_{3,1}$}{}}
Since $\mathfrak{r}_{3,1}$ admits a root decomposition and the unique homogeneous space in $\Lambda^3\mathfrak{r}_{3,1}$ is not related to the zero element of the group (see Table \ref{tabela3w}), the Proposition \ref{Prop:Sym} shows that $(\Lambda^3\mathfrak{r}_{3,1})^{\mathfrak{r}_{3,1}}=\{0\}$. Moreover, the root decomposition of $\mathfrak{r}_{3,1}$ tells us that $(\Lambda^2\mathfrak{r}_{3,1})^{\mathfrak{r}_{3,1}}\subset \Lambda^2(\mathfrak{r}_{3,1})^{(0)}$. It is then immediate that $(\Lambda^2\mathfrak{r}_{3,1})^{\mathfrak{r}_{3,1}}=0$.  Therefore, the determination of $r$-matrices demands solving the corresponding mCYBE and every $r$-matrix gives rise to a different coproduct.

 In the coordinates $\{x,y,z\}$  corresponding to the basis $\{e_{12},e_{13},e_{23}\}$ of $\Lambda^2\mathfrak{r}_{3,1}$, one has  $r=xe_{12}+ye_{13}+ze_{23}$ and $[r,r]_S=0$ for every $r\in \Lambda^2\mathfrak{r}_{3,1}$.
Hence, every element of $\Lambda^2\mathfrak{r}_{3,1}$ is an $r$-matrix giving rise to a coboundary coproduct. 

The fundamental vector fields of the action of ${\rm Inn}(\mathfrak{r}_{3,1})$ on $\Lambda^2\mathfrak{r}_{3,1}$ are spanned by
$$
X_1:=x\frac{\partial}{\partial x}+y\frac{\partial}{\partial y}+2z\frac{\partial }{\partial z},\quad X_2:=y\frac{\partial}{\partial z},\quad X_3:=x\frac{\partial}{\partial z}.
$$

\noindent\begin{minipage}{0.4\textwidth}
\begin{center}
\includegraphics[scale=0.30]{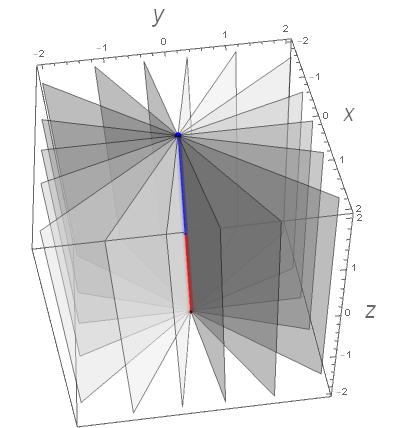}
\captionof{figure}{Representative orbits of the action of ${\rm Inn}(\mathfrak{r}_{3,1})$ on $\Lambda^2 \mathfrak{r}_{3,1}$}\label{Fig7}
\end{center}
\end{minipage}$\quad$
\begin{minipage}{0.58\textwidth}
They generate an integrable two-dimensional distribution off the line $x=y=0$ with integrals given by semi-planes of the form given in Figure \ref{Fig7}. The line $x=y=0$ can also be divided into three orbits of the action of ${\rm Inn}(\mathfrak{r}_{3,1})$ consisting of the points with the same sign of $z$. 

Let us study now the equivalence of $r$-matrices up to the action of ${\rm Aut}(\mathfrak{r}_{3,1})$. Elements of ${\rm Aut}(\mathfrak{r}_{3,1})$ leave the first derived ideal $[\mathfrak{r}_{3,1},\mathfrak{r}_{3,1}]=\langle e_2,e_3\rangle$ invariant. Then, the induced action of ${\rm Aut}(\mathfrak{r}_{3,1})$ on $\Lambda^2\mathfrak{r}_{3,1}$ must leave the subspace $\langle e_{23}\rangle$ invariant and every point within it must be contained in an orbit within $\langle e_{23}\rangle$. Obviously, the $r=0$ is an orbit of the action of ${\rm Aut}(\mathfrak{r}_{3,1})$ on $\Lambda^2\mathfrak{r}_{3,1}$.

\end{minipage}

Moreover, the Lie algebra automorphisms $T_{\alpha,\beta,\gamma,\delta}$ given by 
$$
T_{\alpha,\beta,\gamma,\delta}(e_1):=e_1,\quad T_{\alpha,\beta,\gamma,\delta}(e_2):= \alpha e_2+\beta e_3,\quad
T_{\alpha,\beta,\gamma,\delta}(e_3)=\gamma e_2+\delta e_3,\qquad \alpha\delta-\beta\gamma\neq 0,
$$
are such that the $\Lambda^2T_{\alpha,\beta,\gamma,\delta}$ connect different semi-planes in $\Lambda^2\mathfrak{r}_{3,1}$. Moreover, the above automorphisms connect the parts $z>0$ and $z<0$ of the line $x=y=0$. 
Hence, there exist two non-zero non-equivalent coboundaries induced by the $r$-matrices $r_1=e_{13}$ and $r_2=e_{23}$. It is remarkable that $r_2$ can be shown to be a solution of the CYBE through the gradations in $\mathfrak{r}_{3,1}$ and its induced decompositions in its Grassmann algebra.

\subsubsection{The Lie algebra \texorpdfstring{$\mathfrak{r}_3$}{}} 
In view of Table \ref{tabela3w} and Proposition \ref{Prop:Sym}, the unique non-zero homogeneous space of $\mathfrak{r}_3$ in $\Lambda^3\mathfrak{r}_3$ is invariant under $e_1,e_2$. Nevertheless, a short calculation shows that $[e_3,(\Lambda^3\mathfrak{r}_3)^{(2)}]_S\neq 0$. Hence, the description of coboundary Lie bialgebras on $\mathfrak{r}_3$ requires solving the mCYBE. Since the space of solutions to this equation, let us say YB, is invariant under the action of ${\rm Aut}(\mathfrak{r}_3)$, the classification of such Lie bialgebras reduces to studying of equivalent $r$-matrices in YB.

Let us determine the space $(\Lambda^2\mathfrak{r}_3)^{\mathfrak{r}_3}$ to 
know whether different $r$-matrices induce different coboundary coproducts. Since $\mathfrak{r}_3$ admits a $\mathbb{Z}$-gradation, $(\Lambda^2\mathfrak{r}_3)^{\mathfrak{r}_3}$ is the sum of the $\mathfrak{r}_3$-invariant elements on each homogeneous subspace of  $\Lambda^2\mathfrak{r}_3$. Using the gradation of $\mathfrak{r}_3$ one sees that $[v^{(\alpha)},w^{(\beta)}]_S=0$, for $v^{(\alpha)}\subset \mathfrak{r}_3^{(\alpha)}, w^{(\beta)}\in (\Lambda^2\mathfrak{r}_3)^{(\beta)}$ when $\alpha+\beta\neq 2$. Inspecting remaining commutators, one obtains $(\Lambda^2\mathfrak{r}_3)^{\mathfrak{r}_3}=\{0\}$ and every $r$-matrix induces a different coboundary coproduct.

Let $\{x,y,z\}$ be the coordinates on $\Lambda^2\mathfrak{r}_3$ induced by the basis $\{e_{12},e_{13},e_{23}\}$. The mCYBE, where $r=xe_{12}+ye_{13}+ze_{23}$, reads $[r,r]_S=-2z^2e_{123}$. Hence, $YB=\langle e_{12},e_{13}\rangle$ stands for the space of solutions to the mCYBE, which is presented in Figure \ref{18}. 

A long but simple calculation shows that $\Lambda^2\mathfrak{r}_3$ admits no non-zero $\mathfrak{r}_3$-invariant metrics. Nevertheless, one can still classify $r$-matrices up to the action of ${\rm Inn}(\mathfrak{r}_3)$. The fundamental vector fields of the action of ${\rm Inn}(\mathfrak{r}_3)$ on $\Lambda^2\mathfrak{r}_3$ are spanned by
$$
X_1:=z\frac{\partial}{\partial x},\quad X_2:=(-y+z)\frac{\partial}{\partial x}, \quad X_3:= 2x\frac{\partial}{\partial x} + (y+z)\frac{\partial}{\partial y} + z\frac{\partial}{\partial z}.
$$
Since ${\rm Inn}(\mathfrak{r}_3)$ maps solutions of mCYBE onto new solutions, the above vector fields are tangent to $YB$ and they take on $YB$ the form
$$
X_1|_{YB}=0,\quad X_2|_{YB}=-y\frac{\partial}{\partial x},\qquad X_3|_{YB}=2x\frac{\partial}{\partial x}+y\frac{\partial}{\partial y},
$$

\noindent\begin{minipage}{0.3\textwidth}
\begin{center}
\includegraphics[scale=0.35]{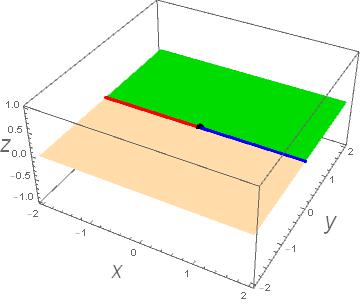}
\captionof{figure}{Orbits of the action of ${\rm Inn}(\mathfrak{r}_3)$ on $YB\subset \Lambda^2 \mathfrak{r}_3$}\label{18}
\end{center}
\end{minipage}
\begin{minipage}{0.7\textwidth}
which span the tangent space to $YB$ when $y \neq 0$;  they span $\langle \partial /\partial x\rangle $ for $y=0$ and $x\neq 0$; and they span a distribution of rank zero for $z=y=x=0$. In consequence, there exist five orbits of ${\rm Inn}(\mathfrak{r}_3)$  depicted in Figure \ref{18}.

Let us accomplish the classification of coboundary cocommutators up to the action of elements of ${\rm Aut}(\mathfrak{r}_{3})$ on $YB$. Since  $\mathfrak{r}_3$ obeys the assumptions of Proposition \ref{Aut3} and  $[r,r]_S$ satisfies the condition in  Proposition \ref{Aut2}, one has that all automorphisms must take the form
$$
T_{\alpha,\beta}(e_3)=e_3+v,\quad T_{\alpha,\beta}(e_1)=\alpha e_1,\qquad T_{\alpha,\beta}(e_2)=\alpha e_2+\beta e_1$$
for all $\alpha \in \mathbb{R}\backslash \{0\},\beta\in \mathbb{R},v\in \langle e_1,e_2\rangle.$ Therefore,

\end{minipage}
$$
\Lambda^2T_{\alpha,\beta}(e_{12})=\alpha^2e_{12},\quad \Lambda^2T_{\alpha,\beta}(e_{13})=\alpha e_{13},\qquad \Lambda^2T_{\alpha,\beta}(e_{23})=\alpha e_{23}+\beta e_{13}+T_{\alpha,\beta}(e_2)\wedge v, 
$$
for all $\alpha\in \mathbb{R}\backslash \{0\}, \beta \in\mathbb{R}, v\in \langle e_2,e_3\rangle
$.\\
\noindent
It was proven in Section \ref{sect:auto_alg} that $[\mathfrak{r}_3,\mathfrak{r}_3]=\langle e_1,e_2\rangle $ is invariant under the action of ${\rm Aut}(\mathfrak{r}_3)$. Thus, $\langle e_{12}\rangle$ is invariant under the action of ${\rm Aut}(\mathfrak{r}_3)$ on $\Lambda^2\mathfrak{r}_3$. In view of the $\Lambda^2T_{\alpha\beta}$, it follows that $\langle e_{12}\rangle$ has three orbits: the $0\in \Lambda^2\mathfrak{r}_3$ and the orbits of $\pm e_{12}$. 
Since $\langle e_{12}\rangle \subset $YB, it is clear that $\pm e_{12}$ are $r$-matrices giving rise to non-zero coproducts. 
Since there exist automorphisms on $\mathfrak{g}$ inverting the coordinate $y$ and leaving $x$ invariant, there exists only one equivalence class of non-zero solutions in YB without $y=z=0$ given by $r_1=e_{13}$.
Hence, we have the equivalence classes related to the $r$-matrices:
$$
r_0 = 0, \qquad r_{\pm}=\pm e_{12},\qquad r= e_{13},
$$
as depicted in color in Figure \ref{Sum}. It is remarkable that $\pm e_{12}$ can be seen to be solutions of the CYBE in view of the gradation of $\mathfrak{r}_3$ and the induced decompositions in its Grassmann algebra.

\subsubsection{The Lie algebra \texorpdfstring{$\mathfrak{r}_{3, \lambda}$}{} (\texorpdfstring{$\lambda \in (-1,1)$}{})}
In view of Proposition \ref{Prop:Sym} and the fact that  $\mathfrak{r}_{3,\lambda}$ admits a root decomposition and the unique non-zero homogeneous space in $\Lambda^3\mathfrak{r}_{3,\lambda}$ has degree three, one has that  $(\Lambda^3\mathfrak{r}_3)^{\mathfrak{r}_{3,\lambda}}=\{0\}$. If we write $r=xe_{12}+ye_{13}+ze_{23}$, then the mCYBE reads $[r,r]_S=2(\lambda-1)yz \,e_{123}$. Hence, the space of $r$-matrices, $YB$, consists of the sum of the plane of points with  $y=0$ and the plane of points with $z=0$. Since $\mathfrak{r}_{3,\lambda}$ admits a root decomposition, $(\Lambda^2\mathfrak{r}_{3,\lambda})^{\mathfrak{r}_{3,\lambda}}\subset (\Lambda^2\mathfrak{r}_{3,\lambda})^{(0)}=\{0\}$.

\noindent\begin{minipage}{0.42\textwidth}
\begin{center}
\includegraphics[scale=0.3]{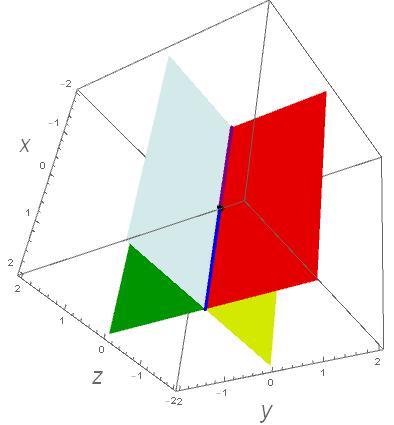}
\captionof{figure}{Orbits of the action ${\rm Inn}(\mathfrak{r}_{3,\lambda})$ on $YB\subset \Lambda^2 \mathfrak{r}_{3, \lambda}$}\label{20}
\end{center}
\end{minipage}$\quad$
\begin{minipage}{0.55\textwidth}
\noindent  As standard, we now accomplish the classification of the coboundary cocommutators up to inner automorphisms of $\mathfrak{r}_{3,\lambda}$. Since $(\Lambda^2\mathfrak{r}_{3,\lambda})^{\mathfrak{r}_{3,\lambda}}=0$, this demands to obtain classes of solutions of the mCYBE (equivalent up to inner automorphisms of $\mathfrak{r}_{3,\lambda}$). Let us now study the solutions to the mCYBE in three subcases: a) $y=0$ with $z\neq 0$, denoted by $YB_1$; b) $z=0$ with $y\neq 0$, denoted by $YB_2$, and c) the line $y=z=0$ denoted by $YB_3$. The desired classification can be achieved by analyzing the fundamental vector fields of the action of  ${\rm Inn}(\mathfrak{r}_{3,\lambda})$ on $\Lambda^2\mathfrak{r}_{3,\lambda}$. These are spanned by
\vskip -0.3cm
$$
Z_1 := z\frac{\partial}{\partial x}, \qquad Z_2:=-\lambda y\frac{\partial}{\partial x},$$
\vskip -0.5cm
$$
Z_3 := (1+\lambda)x\frac{\partial}{\partial x} +y\frac{\partial}{\partial y}+\lambda z\frac{\partial}{\partial z}.
$$
\end{minipage}
\vskip 0.1cm
Assume $\lambda\neq 0$. Let us consider $YB_3$. The distribution $\mathcal{D}$ spanned by $Z_1,Z_2,Z_3$ on $YB_3$ and $x\neq 0$ has rank one. Meanwhile, $\mathcal{D}$ has rank zero at $x=y=z=0$. Hence, $YB_3$ is divided into three orbits for points $(x,0,0)$ with $x>0$,  $x<0$, and $x=0$.

At points of $YB_1$, one has that 
$
Z_1\vert_{YB_1}=z\frac{\partial }{\partial x}, Z_2\vert_{YB_1}=0, Z_3\vert_{YB_1}=(1+\lambda)x\frac{\partial }{\partial x}+\lambda z\frac{\partial}{\partial z}$ 
span the tangent space to $YB_1$. Hence, this gives rise to two orbits of $YB_1$ for its points with $z<0$ and $z>0$, respectively.

The vector fields $Z_1,Z_2,Z_3$ on $YB_2$ read
$
Z_1|_{YB_2}=0, Z_2|_{YB_2}=-\lambda y\frac{\partial}{\partial x}, Z_3|_{YB_2}=(1+\lambda)x\frac{\partial}{\partial x}+y\frac{\partial}{\partial y}
$ and span the tangent space to $YB_2$. Then, we have two orbits of points in $YB_2$ with $y>0$ and $y<0$, correspondingly. Previous results are summarised in Figure \ref{20}.

Let us now classify coboundary coproducts up to the action of ${\rm Aut}(\mathfrak{r}_{3,\lambda})$. Since $[\mathfrak{r}_{3,\lambda},\mathfrak{r}_{3,\lambda}]=\langle e_1,e_2\rangle$ is invariant under ${\rm Aut}(\mathfrak{r}_{3,\lambda})$, the space $\langle e_{12}\rangle $ is also invariant relative to the action of ${\rm Aut}(\mathfrak{r}_{3,\lambda})$. Moreover, the automorphisms of the form
$
T_{\alpha,\beta}(e_1)= \beta e_1,T_{\alpha,\beta}(e_2)= \alpha e_2, T_{\alpha,\beta}(e_3)=e_3,$ for all  $\alpha\in \mathbb{R}\backslash\{0\},
$
are such that the induced $\Lambda^2T_{\alpha,\beta}$ enable us to obtain that $\langle e_{12}\rangle$ has only two equivalence classes: $0$ and $e_{12}$. This finishes the study of solutions with $y=z=0$.

Meanwhile, the $\Lambda^2T_{\alpha,\beta}$ change the sign of $y$ and $z$.  This maps the two semiplane orbits for ${\rm Inn}(\mathfrak{r}_{3,\lambda})$ for the $r$-matrices with $z=0$ and $y=0$. Therefore, we get three classes of inequivalent non-zero coboundary cocommutators (up to the action of ${\rm Aut}(\mathfrak{r}_{3,\lambda}$)) induced by the $r$-matrices: $r_0 = e_{12}$, $r_y = e_{23}$, and $r_z = e_{13}$. This is depicted in Figure \ref{Sum}.

Let us now tackle the case $\lambda=0$. The corresponding Lie algebra is denoted by $\mathfrak{r}_{3,0}$. The analysis of solutions to the mCYBE for aforesaid subcases a) and b) goes similarly as in the previous case. The fundamental vector fields of the action of  ${\rm Inn}(\mathfrak{r}_{3,0})$ read
 $
 Z_1 := z\frac{\partial}{\partial x},  Z_2:=0,Z_3 := x\frac{\partial}{\partial x} +y\frac{\partial}{\partial y}.
 $
 On $YB_3$, the distribution spanned by $Z_1,Z_2,Z_3$ has rank one for $x\neq 0$ and zero for $x=0$. Therefore, we obtain three orbits gathering those points with $z=y=0$ and equal sign of $x$.

Restricting to $YB_1$, we get
 $
 Z_1\vert_{YB_1}=z\frac{\partial }{\partial x}, Z_2\vert_{YB_1}=0, Z_3\vert_{YB_1}=x\frac{\partial }{\partial x},
 $
 which span $\langle \partial / \partial x \rangle$. Thus, the orbits of the action of ${\rm Inn}(\mathfrak{r}_{3,0})$ on this space are lines $(x,0,z_0)$ with a constant value $z_0\neq 0$. Restricting to $YB_2$, i.e. $z=0$ and $y\neq 0$, we get a unique non-zero restriction of $Z_1,Z_2,Z_3$ given by
 $
Z_3\vert_{YB_1}=x{\partial }/{\partial x}+y{\partial}/{\partial y},
 $
 which spans $\langle x\partial / \partial x +y\partial/\partial y\rangle$. Thus, the orbits of the action of ${\rm Inn}(\mathfrak{r}_{3,0})$ on this space are lines $(\mu x,\mu y,0)$ with $\mu>0$ and $y\neq 0$.

 The automorphisms $T_{\alpha,\beta,\gamma}$ such that $
 T_{\alpha,\beta,\gamma}(e_1) := \alpha e_1+\gamma e_2, T_{\alpha,\beta,\gamma}(e_2) := \beta e_2,T_{\alpha,\beta,\gamma}(e_3) := e_3, $ with $\alpha,\beta\in \mathbb{R}\backslash \{0\}$ and $\gamma\in \mathbb{R}$, are such that the $\Lambda^2T_{\alpha,\beta,\gamma}$ identify the lines $(x,0,z_0)$ and $(\mu x_0,\mu y_0,0)$ with different $z_0\neq 0$ and $x_0,y_0\neq0$ among themselves, respectively. Then, we get two $r$-matrices $r_y := e_{13}$ and $r_z:=e_{23}$.

 If $z=y=0$, the automorphisms $\Lambda^2T_{\alpha,\beta,\gamma}$ map points with positive and negative values of $x$. 

We get three classes of inequivalent non-zero coboundary coproducts up to Lie algebra automorphisms of $\mathfrak{r}_{3,0}$ induced by the $r$-matrices given by the non-zero $r$-matrices $r_0 = e_{12}$, $r_y = e_{23}$ and $r_z = e_{13}$, as shown in Figure \ref{Sum}. All of them are trivial solutions of the CYBE in view of the gradation in $\mathfrak{r}_{3,0}$ and the induced decompositions in $\Lambda \mathfrak{r}_{3,0}$.

\subsubsection{The Lie algebra \texorpdfstring{$\mathfrak{r}'_{3, \lambda} (\lambda\neq 0)$}{}}
It stems from Table \ref{tabela3w} that the unique non-zero homogeneous space in $\Lambda^3{\mathfrak{r}'_{3, \lambda}}$ is not invariant relative to the action of $e_3$ and hence $(\Lambda^3{\mathfrak{r}'_{3, \lambda}})^{\mathfrak{r}'_{3, \lambda}}=0$.  The corresponding mCYBE read $[r,r]_S=-2(y^2+z^2)e_{123}$. Hence, the only solutions have $y=z=0$. We denote this space of solutions by $YG$. 

The space $(\Lambda^2{\mathfrak{r}'_{3, \lambda}})^{\mathfrak{r}'_{3, \lambda}}$ can be easily determined as it is spanned by $\mathfrak{r}'_{3, \lambda}$-invariant elements within each homogeneous subspace in $\Lambda^2{\mathfrak{r}'_{3, \lambda}}$. By using the Table \ref{tabela3w} and, eventually, accomplishing easy calculations, one obtains, since $\lambda\neq 0$, that $(\Lambda^2\mathfrak{r}'_{3,\lambda})^{\mathfrak{r}'_{3,\lambda}}=0$.

The invariance of $\kappa_{\mathfrak{r}'_{3,\lambda}}$ under automorphisms and the Lie algebra structure of $\mathfrak{r}'_{3,\lambda}$ show that Proposition \ref{Aut3} applies and all automorphisms have the form $T_\alpha$ with
$
T_\alpha(e_1):=\alpha e_1,T_\alpha(e_2):=\alpha e_2,T_\alpha(e_3):=e_3$, for $\alpha\in \mathbb{R}\backslash\{0\}$. Then, $\Lambda^2T(e_{12})=\alpha^2 e_{12}$ and we obtain three coboundary cocommutators invariant under the action of ${\rm Aut}(\mathfrak{r}'_{3,\lambda})$ given by the $r$-matrices $\pm e_{12}$ and $0$. The result is summarised in Figure \ref{Sum}. All these $r$-matrices are solutions to the CYBE in view of the gradation of $\mathfrak{r}'_{3,\lambda}$ and the induced decompositions in $\Lambda \mathfrak{r}'_{3,\lambda}$.

\section{Conclusion and outlook}
This work has extended methods from Lie algebra theory, like  root decompositions and  $\mathfrak{g}$-invariant maps, to the realm of Grassmann algebras for general Lie algebras. This, along with the use of gradations for Lie algebras and their induced decompositions, opened new ways of determination of coboundary Lie bialgebras, their $\mathfrak{g}$-invariant multivectors, mCYBEs, and their classification up to Lie algebra automorphisms. 

Our techniques have been applied to the classification of coboundary Lie algebras, in general, and the three-dimensional real case has been studied in detail. Our approach simplifies needed calculations to accomplish the classification. For instance, we may skip the determination of all automorphisms of the underlying Lie algebra as in the previous literature \cite{Farinati}. It is remarkable that gradations in Lie algebras and their induced decompositions work well to obtain solutions of mCYBEs and CYBEs. Nevertheless, the gradations are not enough by themselves to analyse the equivalence of their related cocommutators.

Our techniques can be applied to the study of the structure and solutions of mCYBEs for higher-order coboundary and non-coboundary Lie algebras. This will be the goal of future works. 

\section{Acknowledgements}
J. de Lucas acknowledges partial financial support from the contract 1028 financed by the University of Warsaw. D. Wysocki acknowledges a doctoral grant financed by the University of Warsaw and the Kartezjusz program from the University of Warsaw and the Jagiellonian University.

%\newpage
%\phantomsection
%\addcontentsline{toc}{chapter}{Skorowidz}
%\printindex

\end{document}